\documentclass[american, 12pt]{paper}
\usepackage[T1]{fontenc}
\usepackage[utf8]{inputenc}
\usepackage{geometry}
\usepackage{amsthm}
\usepackage[hidelinks]{hyperref}
\usepackage{jn-maxpax}
\usepackage{amsmath}
\usepackage{amssymb}
\usepackage[algo2e, ruled]{algorithm2e}
\usepackage{crossreftools}

\makeatletter
\let\paperclassexample\example

\let\example\relax
\AtBeginDocument{%
  \@ifundefined{example}{\let\example\paperclassexample}{}
}

\newtheorem{thm}{\protect\theoremname}
\newtheorem{defn}[thm]{\protect\definitionname}
\newtheorem{prop}[thm]{\protect\propositionname}
\newtheorem{lem}[thm]{\protect\lemmaname}
\newtheorem{cor}[thm]{\protect\corollaryname}
\newtheorem*{remark}{Remark}

\providecommand{\theoremname}{Theorem}
\providecommand{\definitionname}{Definition}
\providecommand{\propositionname}{Proposition}
\providecommand{\lemmaname}{Lemma}
\providecommand{\corollaryname}{Corollary}

\makeatother
\usepackage{babel}
\usepackage{natbib}

\begin{document}
\title{A Non-asymptotic Analysis for Learning and Applying a Preconditioner
in MCMC}
\author{
  Max Hird\thanks{\texttt{mhird@uwaterloo.ca}}\\
  Department of Statistics and Actuarial Science\\
  University of Waterloo
  \and
  Florian Maire\\
  Département de Mathématiques et de Statistiques\\
  Université de Montréal
  \and
  Jeffrey Negrea\\
  Department of Statistics and Actuarial Science\\
  University of Waterloo
}
\maketitle
\begin{abstract}
Preconditioning is a common method applied to modify Markov chain
Monte Carlo algorithms with the goal of making them more efficient. 
In practice it is often extremely effective, even when the preconditioner
is learned from the chain. We analyse and compare the finite-time
computational costs of schemes which learn a preconditioner based
on the target covariance or the expected Hessian of the target potential
with that of a corresponding scheme that does not use preconditioning.
We apply our results to various algorithms including the Unadjusted Langevin Algorithm (ULA) and the proximal sampler for
an appropriately regular target, establishing non-asymptotic guarantees
for versions of these algorithms that learn and use preconditioners. To do so, we establish non-asymptotic
guarantees on the time taken to collect $N$ approximately independent
samples from the target for schemes that learn their preconditioners
under the assumption that the underlying Markov chain satisfies a
contraction in the Wasserstein-2 distance. This approximate
independence condition, that we formalize, allows us to bridge the
non-asymptotic bounds of modern MCMC theory and classical heuristics
of effective sample size and mixing time, and is needed to amortise
the costs of learning a preconditioner across the samples it
will help to produce. 
\end{abstract}


\section{Introduction}

This article is concerned with the non-asymptotic analysis of Markov chain Monte Carlo (MCMC) algorithms that learn and apply matrix-valued preconditioners. These algorithms are used to estimate expectations with respect to a probability distribution $\pi\in\Pp (\RR^d)$ which we will call the \emph{target}.
Let the target have an unnormalised Lebesgue density of the form $\exp(-U)$ where $U:\RR^d\to\RR$ is called the \emph{potential}. We assume that the potential is twice differential everywhere. At various points we also assume that there exist constants $L\geq m > 0$ such that $m\mathbf{I}_d\preceq\nabla^2U(x)\preceq L\mathbf{I}_d$.
The lower bound on the Hessian implies that $U$ is \emph{$m$-strongly convex} and the upper bound implies that $U$ is \emph{$L$-smooth}. That $U$ has these properties is a standard assumption of the literature on the non-asymptotic theory of MCMC, see e.g. \citet[Theorem 1]{andrieu2024}, \citet[Section 2.2]{wu2022a} and \citet[Equation (1)]{dalalyan2019}. The ratio $\kappa := L/m$ has come to be known as the \emph{condition number} of $\pi$. It is a measure of the anisotropy of $\pi$. Targets with higher condition numbers are more difficult to sample from with standard MCMC samplers that are not adjusted for the anisotropy. This statement is quantified by the dependence on $\kappa$ in bounds on the mixing times of standard MCMC algorithms see e.g. \citet[Theorem 49]{andrieu2024}, \citet[Theorem 1]{cheng2018} and \citet[Theorem 1]{wu2022a}.

\emph{Preconditioning} and \emph{adaptive MCMC} are concerned with learning transformations to the target that make it easier to sample from. Often a linear transformation is learned to make the target more isotropic and therefore reduce the condition number. For example \citet{haario2001} proposes the inverse of the target covariance $\Sigma_\pi^{-1} := \mathrm{Cov}_\pi(X)^{-1}$ and \citet{titsias2023} proposes the `Fisher matrix' $\Ff := \EE_\pi[\nabla^2U(X)]$ as preconditioners. It is known that a well-chosen linear transformation can vastly improve sampler efficiency. However, in practice, such a linear transformation must be learned using information from the chain. This prompts the motivating question behind the present work:
\vspace{6pt}
\begin{center}
\emph{Can a Markov chain-based sampler that learns its own linear preconditioner have a lower total computational complexity than its un-preconditioned counterpart?}
\par\end{center}
\vspace{6pt}

Theorems \ref{thm:total_ULA_complexities}, \ref{thm:total_unadjusted_underdamped_complexities}, \ref{thm:total_unadjusted_HMC_complexities}, and \ref{thm:total_proximal_complexities} allow us to answer in the affirmative for various sampling algorithms, as long as the time over which we use the preconditioner is sufficiently long, and if the target is linearly preconditionable with $\Sigma_\pi^{-1}$ or $\Ff$.

To compare MCMC performance with and without  preconditioning, we posit an appropriate termination condition. Prior work on non-asymptotic MCMC performance focuses on the time until the law of final sample from the chain is $\varepsilon$-close to $\pi$. Whilst informative, this condition does not mirror the actual practice of MCMC in which an ensemble of samples is collected after a burn-in. Nor does it make sense to measure the success of a preconditioner in this way: an up-front cost is incurred in learning the preconditioner and must therefore be \emph{amortised} over multiple iterations. Therefore we posit the following condition:
\begin{defn}\label{defn:approximately_IID_condition}
    Let $\varepsilon>0$. A random element $\{ X_{t}\} _{t=1}^{N}$ taking values in $\mathbb{R}^{d\times N}$
is $\sqrt{N}\varepsilon$-approximately IID from $\pi$ in $W_{2}$
if
\begin{equation}
W_{2}\left(\Ll\left(\left\{ X_{t}\right\} _{t=1}^{N}\right),\pi^{\otimes N}\right)\leq\sqrt{N}\varepsilon
\end{equation}
where $W_2:\Pp(\RR^{d \times N}) \times \Pp(\RR^{d \times N}) \to \RR^+$ is the Wasserstein-2 distance of the Frobenius norm.
\end{defn}
 We use the terms `$\sqrt{N}\varepsilon$-AIID' or simply `AIID' where the context allows. Independent and identically distributed (IID) samples from $\pi$ are `ideal' for Monte Carlo estimators, and hence this condition quantifies the distance from the output of the MCMC to a sample of ideal quality. The AIID condition controls the error of empirical averages constructed using the MCMC output, see Propositions \ref{prop:L2_rootNepsilon_mean_error} and \ref{prop:rootNepsilon_bounded_estimator_variance}. This condition bridges classical MCMC results, where performances of algorithms are measured by the asymptotic variance of their estimators or by the effective sample size, and newer, non-asymptotic results, where performance of algorithms are measured by how quickly they achieve an $\varepsilon$-approximate sample. The AIID condition naturally combines the error due to both bias and variance of the MCMC outputs, and hence can be assessed even for a kernel that does not have $\pi$ as its invariant distribution. See Section \ref{section:rootNepsilonAIID} for further discussion of the condition and its implications.

To ensure that the output of the algorithm is $\sqrt{N}\varepsilon$-AIID, we examine MCMC algorithms whose kernels are contractive in the Wasserstein-2 distance in the following sense:
\begin{defn}\label{def:W_2_contraction}Let $K:\mathbb{R}^{d}\times\mathcal{B}(\mathbb{R}^{d})\to[0,1]$
be a Markov kernel. For $\gamma>0$,
$b\geq0$ and $\Gamma>0$, $K$ is a $(\Gamma, \gamma,b)$-$W_{2}$
contraction to $\pi$ when, for all $k\in\mathbb{N}$ and $\mu\in\mathcal{P}(\mathbb{R}^{d})$, it holds that:
\begin{equation}
    W_{2}\left(\pi,\mu K^{k}\right)\leq\Gamma\exp\left(-\gamma k\right)W_{2}\left(\pi,\mu\right)+b.
\end{equation}

\end{defn}
Many Markov kernels have such a property: see \citet[Theorem 1]{dalalyan2019}, \citet{cheng2018} and \citet{bou-rabee2025}. The presence of $b\geq 0$ indicates an MCMC algorithm that is possibly biased or inexact for $\pi$.

The MCMC kernels we examine will satisfy Definition \ref{def:W_2_contraction}. Thus, in order to produce an $\sqrt{N}\varepsilon$-approximately IID from $\pi$ in $W_2$ sample, we output samples in following way:
\begin{enumerate}[noitemsep]
    \item Burn in for $t_\mathrm{burn}\in \ZZ^+$ iterations. Discard the burn-in states.
    \item Collect a new sample for the output every $t_\mathrm{thin} \in \ZZ^+$ iterations.
\end{enumerate}
Intuitively, step 1 ensures that the first state in the output ensemble is within $\varepsilon$ of $\pi$ in $W_2$, and step 2, commonly known as \emph{thinning}, ensures that the total output ensemble is close to an IID sample from $\pi$. This result is made formal in Theorem \ref{thm:Wasserstein_contraction_complexity}. Thinning, while technically inefficient, only suffers a constant factor of inefficiency when the optimal thinning period is used. Thus our bounds will be based on the optimal thinning period, and can be extended and applied to the unthinned chain with negligible losses. The preconditioned algorithms will do the above two steps twice: first, to learn their preconditioners, and, second, to collect the output. Thus, this method can be seen as simplified version of the adaptive strategy used in the HMC sampler from the popular software package \texttt{stan} \citep{carpenter2017a}. The analysis in the present work could be applied multiple times, in stages, to reflect the true implementation in \texttt{stan} or to better reflect other instantiations of adaptive MCMC.

The main contributions of this article are to state the total number of FLOPS for preconditioned and unpreconditioned algorithms sufficient for their outputs to be $\sqrt{N}\varepsilon$-AIID from $\pi$ in $W_2$. Stating complexity in terms of FLOPS is necessary because preconditioned algorithms incur at least an additional matrix-vector product per iteration over their unpreconditioned counterparts, and  this additional cost must be counted for a fair comparison. Although our results concern generic MCMC kernels that satisfy a $W_2$ contraction stated in Definition \ref{def:W_2_contraction}, we focus on the Unadjusted Langevin Algorithm (ULA) \citep[Equation (2)]{dalalyan2019}, the Unadjusted Underdamped Langevin Algorithm \citep[Algorithm 1]{cheng2018}, a variant of unadjustem Hamiltonian Monte Carlo (HMC) introduced by \citet{bou-rabee2025} and the proximal sampler of \citet{lee2021a}. We examine the covariance based preconditioner as introduced by \citet{haario2001} and the so-called `Fisher matrix' preconditioner introduced by \citet{titsias2023}.

\subsection{Related Work}

The past ten years has seen a proliferation of non-asymptotic MCMC theory under the assumption of a smooth and strongly convex potential. \citet{mitra2025} study the mixing time in the class of $\Phi$-divergences of ULA and the Proximal sampler of \citet{lee2021a} under various assumptions that are implied by the strong convexity of the potential. Under similar assumptions \citet{liang25} study the time between two approximately independent samples of ULA as measured by their $\Phi$-mutual information. Our $\sqrt{N}\varepsilon$-AIID condition in Definition \ref{defn:approximately_IID_condition} controls both time to mix, and time between approximately independent samples and so our objective combines the objectives of \citet{mitra2025} and \citet{liang25}.

Adaptive MCMC has its origins in \citet{haario2001}. Non-asymptotic theory for its instantiations is rare although \citet{hofstadler2025} give a non-asymptotic rate of convergence of the Monte Carlo estimator produced by the Adapted Increasingly Rarely MCMC method of \citet{chimisov2018a} up to an unknown constant. The body of literature on the theory of preconditioning is also meagre although see \citet{titsias2023} for a justification for the Fisher preconditioner via the expected squared jump distance of ULA, \citet{hird2025a} for an investigation of the effects on the condition number of various preconditioners and \citet{negrea2019} for a statement of the preconditioner that maximises the spectral gap of the limiting diffusion associated with a time rescaling of the RWM.

Corollary \ref{cor:preconditioner_learn_complexity} bounds the number of iterations needed for an MCMC algorithm to estimate $\Sigma_\pi$ or $\Ff$ with moderate accuracy, with high probability. The cost of estimating $\Sigma_\pi$ using a $\pi$-invariant Markov kernel is also given by \citet{kook2024}. In contrast, our results apply to both exact and inexact methods. \citet{nakakita2026} also gives complexity bounds for covariance estimation using ULA, in which they use smoothness and strong convexity of the target to show a log-Sobolev inequality holds for the (biased) invariant distribution. Our Corollary \ref{cor:preconditioner_learn_complexity}, part 1, uses only strong convexity and a lower bound on $\Sigma_\pi$ to learn the target covariance, and part 2 only uses smoothness and a lower bound on $\Ff$ to learn the Fisher information. We then also apply these results to evaluate algorithms that are preconditioned with the subsequent estimates of $\Sigma_\pi$ and $\Ff$.

\subsection{Notation}

Let $\Pp(\RR^d)$ be the set of probability measures on $\RR^d$. Unless otherwise stated, $\|\cdot\|$ is the $2$-norm. The Wasserstein distance $W_p:\Pp(\RR^n)\times \Pp(\RR^n)\to[0,\infty)$ is defined by $W_p(\nu,\pi)^p := \inf_{\gamma\in\Cc(\nu,\pi)}\EE_{(X, Y)\sim\gamma}[\|X - Y\|^p]$ for $\nu,\pi\in\Pp(\RR^n)$, where $\Cc(\nu,\pi)$ is the set of couplings of $\nu$ and $\pi$. When $\nu$ and $\pi$ are defined on spaces of matrices the norm is the Frobenius norm. $\mathrm{PD}_{d\times d}$ denotes the set of positive definite matrices in $\RR^{d\times d}$. For a symmetric matrix $A\in\RR^{d\times d}$, $\lambda_i(A)$ is its $i$th-largest eigenvalue. We denote the condition number of a positive definite matrix $A \in \mathrm{PD}_{d\times d}$ with $\kappa(A) := \lambda_1(A) / \lambda_d(A)$. For a probability measure $\nu$, the $N$-fold product measure of $\nu$ is denoted $\nu^{\otimes N}$. The pushforward of $\nu$ through a function $T$ is denoted $T_\sharp\nu$. The covariance matrix of $\nu$ is denoted $\Sigma_\nu$. Let $M \in \mathrm{PD}_{d\times d}$ be a linear preconditioner. We define $L_M$, $m_M$, and $\kappa_M$ to be the smoothness constant, strong-convexity constant and condition number of $(M^{1/2})_\sharp\pi$ where $M^{1/2}$ is the positive square root i.e.
\[
m_M := \inf_{x\in\RR^d} \lambda_d(M^{-1/2}\nabla^2 U(x)M^{-1/2})\textup{, }L_m := \sup_{x\in\RR^d}\lambda_1(M^{-1/2}\nabla^2 U(x)M^{-1/2})\textup{ and }\kappa_M := \frac{L_M}{m_M}.
\]

We define time-homogeneous Markov kernels $K:\RR^d \times \Bb(\RR^d) \to [0,1]$ with the following notation:
\[
K(x \to A) := \PP(X_{t + 1} \in A | X_t = x)
\]
for all $x\in \RR^d$ and $A \in \Bb(\RR^d)$. For $k \in \NN$ we define $K^k(x \to A) := \PP(X_{t + k} \in A | X_t = x)$. Markov kernels $K$ act on probability distributions $\mu \in \Pp(\RR^d)$, producing new probability distributions, according to:
\[
\mu K(A) := \int\mu(\dee x)K(x \to A) \qquad \text{for all }A\in\Bb(\RR^d).
\]

\section{Sampling Algorithms}

The MCMC algorithms we examine follow the generic form laid out in Algorithm \ref{alg:thinned_Markov_chain}. The algorithms that learn and use a preconditioner follow the form in Algorithm \ref{alg:preconditioned_Markov_chain}.

\begin{algorithm}
    \caption{Thinned Markov chain sampler.}\label{alg:thinned_Markov_chain}
    \SetKwInOut{Input}{input}\SetKwInOut{Output}{output}
    \Input{Markov kernel $K:\Reals^d \times \Bb(\RR^d) \to [0, 1]$, an initial distribution $\mu \in \Pp(\RR^d)$, a burn-in time $t_\mathrm{burn}\in\ZZ^+$, a thinning period $t_\mathrm{thin} \in \ZZ^+$ and an output size $N\in\ZZ^+$.}
    \Output{$\{X_t\}_{t = 1}^{N}\in\RR^{d \times N}$}
    \begin{enumerate}[noitemsep]
        \item Sample $X_1 \sim \mu K^{t_\mathrm{burn}}$.
        \item For $t \in [N-1]$ : Sample $X_{t+1} \sim K^{t_\mathrm{thin}}(X_t \to \cdot)$.
    \end{enumerate}
\end{algorithm}
\begin{algorithm}
    \caption{Preconditioned Markov chain sampler.}\label{alg:preconditioned_Markov_chain}
    \SetKwInOut{Input}{input}\SetKwInOut{Output}{output}
    \Input{A family of $\pi$-informed $M$-preconditioned Markov kernels $\{K_{M, \pi}:\RR^d \times \Bb(\RR^d) \to [0, 1]:M\in\mathrm{PD}_{d\times d}\}$, sample size for preconditioner learning $N_\mathrm{learn}\in \ZZ^+$, preconditioner estimator $\Mm : \RR^{d \times N_\mathrm{learn}} \to \mathrm{PD}_{d \times d}$, 
    initial distribution $\mu \in \Pp(\RR^d)$, burn-in times $t_\mathrm{burn}, t_\mathrm{burn}^\prime\in\NN$, thinning periods $t_\mathrm{thin}, t_\mathrm{thin}^\prime \in \NN$ and output size $N\in\NN$.
    }
    \Output{$\{X_t\}_{t = 1}^N\in\RR^{d \times N}$}
    \begin{enumerate}[noitemsep]
        \item Collect $\{Z_t\}_{t=1}^{N_\mathrm{learn}}$ using Algorithm \ref{alg:thinned_Markov_chain} with input $(K_{\mathbf{I}_d, \pi},\mu,t_\mathrm{burn}, t_\mathrm{thin}, N_\mathrm{learn})$.
        \item Construct the preconditioner $M = \Mm(\{Z_t\}_{t=1}^{N_\mathrm{learn}})$.
        \item Output $\{X_t\}_{t=1}^N\in\RR^{d \times N}$ using Algorithm \ref{alg:thinned_Markov_chain} with input $(K_{M, \pi},\mu,t_\mathrm{burn}^\prime, t_\mathrm{thin}^\prime, N)$.
    \end{enumerate}
\end{algorithm}

An example of the kind of Markov kernel we will use is the preconditioned ULA kernel: given a step-size $h^2 > 0$ and a preconditioner $M\in \mathrm{PD}_{d \times d}$ it is defined as
\[\label{eq:preconditioned_ULA_kernel}
K_{M, \pi}(x \to \cdot) := \mathcal{N}\left(x - \frac{h^2}{2}M^{-1}\nabla U(x), h^2M^{-1}\right)
\]
for all $x\in\RR^d$ and $A \in \Bb(\RR^d)$. It is the Euler--Maruyama integrator of the Langevin diffusion with preconditioner $M$ and step size $h^2$. It outputs samples that are approximately distributed according to the target $\pi$ (for $h^2$ sufficiently small). However the preconditioner $M$ means that the algorithm will output a chain that is isomorphic in the sense of \citet[Appendix A]{johnson12} to a non-preconditioned ULA kernel on a $M^{1/2}_\sharp \pi$ target \citep[Proposition 26]{hird2025}. This means that the $K_{M, \pi}$ kernel will output samples approximately distributed to $\pi$ but will perform as though it `sees' the smoothness and strong convexity of $M^{1/2}_\sharp \pi$, which may be improved compared to $\pi$.

In general we define $K_{M,\pi}$ using an underlying unpreconditioned $\pi$-informed kernel $K_\pi$ as follows:

\begin{defn}\label{def:preconditioned_kernel}
    Let $K_\pi:\RR^d \times \Bb(\RR^d) \to [0, 1]$ be a Markov kernel that is informed by $\pi \in \Pp(\RR^d)$. For $M \in \mathrm{PD}_{d\times d}$ we define the $M$-preconditioned kernel $K_{M,\pi}:\RR^d \times \Bb(\RR^d) \to [0, 1]$ with
    \[
    K_{M, \pi}(x \to A) := K_{M^{1/2}_\sharp\pi}(M^{1/2}x \to M^{1/2}A)
    \]
    for all $x\in\RR^d$ and $A \in \Bb(\RR^d)$ where $M^{1/2}\in\mathrm{PD}_{d\times d}$ is the positive square root of $M$.
\end{defn}

For instance if $K_\pi$ is the unpreconditioned ULA kernel, then $K_{M,\pi}$ is the preconditioned ULA kernel in equation (\ref{eq:preconditioned_ULA_kernel}). One may also generate the commonly used preconditioned versions of the Metropolis adjusted Langevin Algorithm, random walk Metropolis, unadjusted HMC and HMC in this fashion \citep[Lemma 78]{hird2025}. Note that the tuning parameters of $K_{M,\pi}$ such as step-size or integration length will be determined by $M^{1/2}_\sharp \pi$. One may use other forms of square root in the place of the positive square root of $M$: we use the positive root for simplicity. The time complexity of using $K_{M, \pi}$ as opposed to $K_\pi$ will, by construction, only incur a few extra matrix vector products (once $M^{1/2}$ and $M^{-1/2}$ are calculated) since $M^{1/2}_\sharp \pi$ is usually accessed through its potential (hence no determinant of $M$ needs calculating).

Preconditioning i.e. going from $K_\pi$ to $K_{M,\pi}$ is done to improve mixing. In this paper we concern ourselves with Markov kernels that mix in $W_2$ in the precise way governed by the constants $(\Gamma, \gamma, b)$ in Definition \ref{def:W_2_contraction}. These constants will depend on the target $\pi$. Preconditioning with $K_{M,\pi}$ works therefore because it will adjust the constants in the following way:

\begin{prop}\label{prop:preconditioned_W_2_contraction}
    Let $K_{M^{1/2}_\sharp\pi}$ be a $(\Gamma(M^{1/2}_\sharp\pi), \gamma(M^{1/2}_\sharp\pi), b(M^{1/2}_\sharp\pi))$-$W_2$ contraction to $M_\sharp^{1/2}\pi$ as defined in Definition \ref{def:W_2_contraction}. Then $K_{M,\pi}$ is a $(\tilde{\Gamma}, \tilde{\gamma}, \tilde{b})$-$W_2$ contraction to $\pi$ where
    \[
    \tilde{\Gamma} = \sqrt{\kappa(M)}\Gamma(M^{1/2}_\sharp \pi),\;\tilde{\gamma} = \gamma(M^{1/2}_\sharp \pi)\textup{ and }\tilde{b} = \|M^{-1/2}\|b(M^{1/2}_\sharp \pi).
    \]
\end{prop}

Note that the new mixing rate $\tilde{\gamma}$ now depends on $M^{1/2}_\sharp\pi$ \emph{without} suffering any inflation caused by changing coordinates according to $M$. A proof of Proposition \ref{prop:preconditioned_W_2_contraction} can be found in Section \ref{proof:preconditioned_W_2_contraction}.

\section{Main Results}\label{section:main_results}

Here we present two main results: the first states the total number of iterations Algorithm \ref{alg:thinned_Markov_chain} must run for to achieve an output which is $\sqrt{N}\varepsilon$-approximately IID from $\pi$ when the underlying Markov kernel satisfies the Wasserstein contraction condition in Definition \ref{def:W_2_contraction}. The second assumes that one has an ensemble of samples that is $\sqrt{N}\varepsilon$-approximately IID from $\pi$ and gives value of $N$ and $\varepsilon$ such that the ensemble can be used to form `good' preconditioners. We define what it means for a preconditioner to be good in equations \ref{eqn:preconditioner_estimate_goodness_condition_covariance} and \ref{eqn:preconditioner_estimate_goodness_condition_Fisher}. Note that the results are independent of each other and are only connected by the $\sqrt{N}\varepsilon$-approximately IID condition defined in Definition \ref{defn:approximately_IID_condition}.

Recently many results have been posited concerning the convergence of MCMC kernels in the Wasserstein-2 distance, see e.g \citet{bou-rabee2025}, \citet{durmus2019}, \citet{cheng2018}. Motivated by this we state a generic result that gives the iteration complexity of an MCMC kernel satisfying the Wasserstein contraction condition in Definition \ref{def:W_2_contraction} to achieve an output that is $\sqrt{N}\varepsilon$-approximately IID from $\pi$ in $W_2$.

\begin{thm}\label{thm:Wasserstein_contraction_complexity}
    Let the underlying MCMC kernel for the thinned Markov chain sampler in Algorithm \ref{alg:thinned_Markov_chain} satisfy the Wasserstein contraction condition in Definition \ref{def:W_2_contraction} with $\Gamma \geq 1$. Let $\varepsilon>0$ such that $\varepsilon \in ( 3\Gamma^2b, \sqrt{3\mathrm{tr}(\Sigma_\pi)}]$. The cost, in iterations of the underlying Markov chain, to output a $\sqrt{N}\varepsilon$-approximately IID from $\pi$ in $W_2$ is $O(\gamma^{-1}N\log((\varepsilon - 3\Gamma^2b)^{-1}))$.
\end{thm}

For a proof see Section \ref{proof:Wasserstein_contraction_complexity}. The lower bound on $\varepsilon$ implies one cannot achieve arbitrary precision with a biased algorithm ($b>0$). With this, we can substitute in sampler dependent values for $(\Gamma,\gamma,b)$ to determine the iteration complexities of unpreconditioned algorithms, recovering results from the works cited, in terms of their dependence on $\kappa$, $d$ and $\varepsilon$ after setting $N=1$. Since $\gamma$ and $b$ will depend on $L$ and $m$ we can also give iteration complexities (and therefore total complexities in FLOPS) of preconditioned algorithms, so long as we can quantify the effect of the preconditioner on $L$ and $m$. To do so, for a preconditioner $M\in\mathrm{PD}_{d\times d}$, we define the preconditioned smoothness and strong convexity constants
\begin{equation}
    L_M:= \sup_{x\in\RR^d} \lambda_1\left(M^{-1/2}\nabla^2U\left(x\right)M^{-1/2}\right)\textup{ and }m_M := \inf_{x\in\RR^d} \lambda_d\left(M^{-1/2}\nabla^2U\left(x\right)M^{-1/2}\right).
\end{equation}
We define the condition number after preconditioning with $M$ as $\kappa_M := L_M /m_M$.
We focus on the preconditioners $\Sigma^{-1}_\pi := \mathrm{Cov}_\pi(X)^{-1}$ introduced by \citet{haario2001} and $\Ff := \mathrm{Cov}_\pi(\nabla\log\pi(X))$ introduced by \citet{titsias2023} referred to as the `Fisher matrix'. In general we will not know these preconditioners in advance. Thus, we must compute estimates $\widehat{\Sigma}_\pi$ and $\widehat{\Ff}$ of $\Sigma_\pi$ and $\Ff$. We consider our estimates to be sufficiently good if
\begin{equation}\label{eqn:preconditioner_estimate_goodness_condition_covariance}
    \left\|\Sigma_\pi^{-1/2}\widehat{\Sigma}_\pi\Sigma_\pi^{-1/2} - \mathbf{I}_d\right\|\leq \Delta\textup{ with probability }\geq 1-\delta\text{, or}
\end{equation}
\begin{equation}\label{eqn:preconditioner_estimate_goodness_condition_Fisher}
    \left\|\Ff^{-1/2}\widehat{\Ff}\Ff^{-1/2} - \mathbf{I}_d\right\|\leq \Delta\textup{ with probability }\geq 1-\delta\hphantom{\quad\text{, or}}
\end{equation}
In particular, if these inequalities hold, we can control $\kappa_{\widehat{\Sigma}_\pi^{-1}}$ with $\kappa_{\Sigma_\pi^{-1}}$ or $\kappa_{\widehat{\Ff}}$ with $\kappa_{\Ff}$, see Corollary \ref{cor:conditioning_comparison_covariance_case}. Here $\delta,\Delta >0$ are parameters to be chosen later, although in general $\Delta$ needn't be excessively small: $\Delta \approx 0.5$ suffices, see the Remark at the end of Section \ref{section:application_to_existing_algorithms}. This essentially means we only need a moderately accurate estimate of the desired preconditioner to realize its benefits (up to constants).

Say we have an output $\{X_t\}_{t=1}^N$ that is $\sqrt{N}\varepsilon$-approximately IID from $\pi$ in $W_2$. Let
\[
    \widehat{\Sigma}_\pi & :=\frac{1}{N}\sum_{t=1}^N \left(X_t - \bar{X}\right)\left(X_t - \bar{X}\right)^\top  && \text{and} &
    \widehat{\Ff} & :=\frac{1}{N}\sum_{t=1}^N \nabla\log\pi\left(X_t\right)\nabla\log\pi\left(X_t\right)^\top,
\]
where \smash{$\bar{X} := \frac{1}{N}\sum_{t=1}^NX_t$}.
Next, we give conditions ensuring \cref{eqn:preconditioner_estimate_goodness_condition_covariance} and \cref{eqn:preconditioner_estimate_goodness_condition_Fisher} hold:
\begin{thm}\label{thm:preconditioner_learn_complexity}
    Let $\{X_t\}_{t=1}^N$ be $\sqrt{N}\varepsilon$-approximately IID from $\pi$ in $W_2$ and let $\delta, \Delta > 0$.
    \begin{enumerate}
        \item (Covariance Preconditioner Learning) Let $U$ be $m$-strongly convex, and let $\beta_{\Sigma_\pi}>0$ with $\beta_{\Sigma_\pi}\mathbf{I}_d \preceq \Sigma_\pi$. Then it suffices to take
        \[
            N&=\max\left\{ \frac{5d}{\delta\Delta},\frac{2CK_\mathrm{cov}^{2}\left(d+\log\left(4/\delta\right)\right)\sqrt{CK_\mathrm{cov}^{2}+2\Delta}}{\Delta}\right\}\;\mathrm{ and}\\
            \varepsilon &=  \frac{\sqrt{2}}{120}\frac{\delta\Delta}{\sqrt{d}}\sqrt{\beta_{\Sigma_\pi}}
        \]
    for $\{X_t\}_{t=1}^N$ to satisfy \cref{eqn:preconditioner_estimate_goodness_condition_covariance} where $K_\mathrm{cov} := m^{-1/2}_{\Sigma_\pi^{-1}}$ and $C>0$ is a universal constant. 
    \item (Fisher Preconditioner Learning) Let $U$ be $L$-smooth, and let $\alpha_\Ff>0$ with $\alpha_\Ff \mathbf{I}_d \preceq \Ff$. Then it suffices to take
        \[
            N&=\frac{2cK_\mathrm{Fisher}^{2}\left(d+\log(4/\delta)\right)\sqrt{cK_\mathrm{Fisher}^{2}+2\Delta}}{\Delta}\;\mathrm{ and}\\
            \varepsilon &=  \frac{3}{8}\frac{\delta\Delta}{L\sqrt{d}}\sqrt{\alpha_\Ff}
        \]
    for $\{X_t\}_{t=1}^N$ to satisfy \cref{eqn:preconditioner_estimate_goodness_condition_Fisher} where $K_\mathrm{Fisher} := \sqrt{L_\Ff}$ and $c>0$ is a universal constant.
    \end{enumerate}
\end{thm}
For a proof see Sections \ref{proof:preconditioner_learn_complexity_part_1} and \ref{proof:preconditioner_learn_complexity_part_2}. To learn the covariance preconditioner we assume strong convexity and a positive lower bound on the spectrum of $\Sigma_\pi$. To learn the Fisher preconditioner we assume smoothness and a positive lower bound on the spectrum of $\Ff$. These assumptions therefore relax the usual assumptions of both strong convexity and smoothness. We state in Appendix \ref{section:preconditioner_learning_corollary} a corollary to \cref{thm:Wasserstein_contraction_complexity} and \cref{thm:preconditioner_learn_complexity} where we give the total number of iterations necessary for the output of the thinned Markov chain sampler in Algorithm \ref{alg:thinned_Markov_chain} to achieve \cref{eqn:preconditioner_estimate_goodness_condition_covariance} or \cref{eqn:preconditioner_estimate_goodness_condition_Fisher}.

\section{Application to Existing Algorithms}\label{section:application_to_existing_algorithms}

We may now state the total complexity of unpreconditioned and preconditioned MCMC algorithms. For the unpreconditioned algorithms we apply Theorem \ref{thm:Wasserstein_contraction_complexity} directly. For the preconditioned algorithms we use Theorem \ref{thm:Wasserstein_contraction_complexity} with the values of $N$ and $\varepsilon$ from Theorem \ref{thm:preconditioner_learn_complexity} to give us the complexity of estimating the preconditioners in step 1 of Algorithm \ref{alg:preconditioned_Markov_chain} and Theorem \ref{thm:Wasserstein_contraction_complexity} with new values of $\gamma$ and $b$ corresponding to the preconditioned chain to determine the complexity of collecting $\{X_t\}_{t=1}^N$ in step 3 of Algorithm \ref{alg:preconditioned_Markov_chain}. Note that we use Theorem \ref{thm:Wasserstein_contraction_complexity} \emph{twice} in the proof of the time complexity of the preconditioned algorithms. See Figure \ref{fig:timelines} for a visual breakdown of the results.

\begin{figure}[H]
    \centering
    \includegraphics[scale = 0.9]{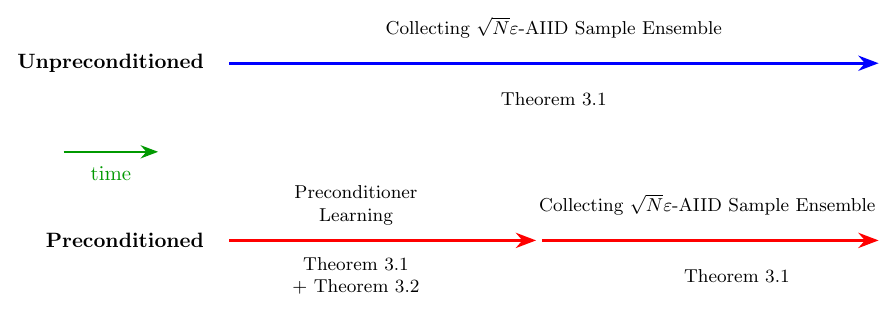}
    \caption{Timelines of activities in the unpreconditioned Algorithm \ref{alg:thinned_Markov_chain} (top) and the preconditioned Algorithm \ref{alg:preconditioned_Markov_chain} (bottom) along with the Theorems we use to quantify their complexities.}
    \label{fig:timelines}
\end{figure}

We first apply our results to the ULA kernel, where the effect of preconditioning on $\gamma$ and $b$ is explicit.

\begin{thm}\label{thm:total_ULA_complexities}
    Let $U$ be $m$-strongly convex and $L$-smooth. Let $T$ be the minimum total number of FLOPS required to output $\{X_t\}_{t=1}^N$ that is $\sqrt{N}\varepsilon$-AIID from $\pi$ in $W_2$. Let $\mathtt{G}$ denote the computational complexity in FLOPS of a call to $\nabla\log\pi$.
    \begin{enumerate}[noitemsep]
        \item For the Algorithm \ref{alg:thinned_Markov_chain}, with unpreconditioned ULA as the Markov kernel, if 
        \[
            \varepsilon
                & \leq\min\left\{10\kappa L^{-1/2}\sqrt{d},\sqrt{3\mathrm{tr}\left(\Sigma_\pi\right)}\right\} 
                    &&\text{ then }
                     & T \in \tilde{O}\left(m^{-1}\rbra{d + \mathtt{G}}d\kappa^2 \frac{N}{\varepsilon^{2}}\right).
        \]
        \item For Algorithm \ref{alg:preconditioned_Markov_chain} with the covariance preconditioner and ULA as the Markov kernel, if
        \begin{equation}
            \varepsilon\leq\min\cbra[2]{\frac{10}{3\sqrt{2}}\frac{\kappa_{\Sigma_\pi^{-1}}\sqrt{d}}{\sqrt{L_{\Sigma_\pi^{-1}}}},\sqrt{2d}},
        \end{equation}
        then, for $K_\mathrm{cov} := m_{\Sigma_\pi^{-1}}^{-1/2}$, with probability $\geq 1 - \delta$, $T = T_\mathrm{learn} + T_\mathrm{collect}$ where
        \[
        T_\mathrm{learn} &\in \tilde{O}\left( d^3\left(d + \mathtt{G}\right)\frac{\kappa^3}{\delta^{2}}\max\left\{\frac{1}{\delta},K^3_\mathrm{cov}\right\}\right)\\&\qquad\text{ and }\\T_\mathrm{collect} &\in \tilde{O}\left(m^{-1}_{\Sigma_\pi^{-1}}\left(d^2 + \mathtt{G}\right)d\kappa_{\Sigma_\pi^{-1}}^2\frac{N}{\varepsilon^2}\right)\\
        \]
        \item For Algorithm \ref{alg:preconditioned_Markov_chain} with the Fisher preconditioner and ULA as the Markov kernel, if
        \begin{equation}
            \varepsilon \leq \min\left\{ 2 \kappa_{\Ff}L_{\Ff}^{-1/2}\sqrt{d},\sqrt{\frac{3}{2}d}\right\},
        \end{equation}
        then, for $K_\mathrm{Fisher} := \sqrt{L_\Ff}$, with probability $1-\delta$, $T = T_\mathrm{learn} + T_\mathrm{collect}$ where
        \[
        T_\mathrm{learn} & \in \tilde{O}\left(d^3\left(d + \mathtt{G}\right)\frac{\kappa^4}{\delta^{2}} K_\mathrm{Fisher}^3\right)\\&\qquad\text{ and }\\
        T_\mathrm{collect} &\in \tilde{O}\left(m^{-1}_\Ff\left(d^2 + \mathtt{G}\right)d\kappa_\Ff^2\frac{N}{\varepsilon^2}\right)\\
        \]
    \end{enumerate}
\end{thm}
See Sections \ref{subsection:proof_of_unpreconditioned_ULA_complexity},  \ref{proof:total_ULA_complexities_part_2} and \ref{proof:total_ULA_complexities_part_3} for proofs. Note that the complexities of the preconditioned algorithms decompose into two additive terms $T_\mathrm{learn}$ and $T_\mathrm{collect}$: the first quantifies how long is needed to learn the preconditioner and the second is how long it takes to produce the final ensemble of samples. Therefore, if we did not have to learn the preconditioners, the time complexity of the algorithms would simply be the second additive terms. The difference between the times to learn the preconditioners is a factor of $\kappa$. This is the same factor of $\kappa$ we discussed following the statement of Corollary \ref{cor:preconditioner_learn_complexity} and we suspect is an artifact of the analysis. The $d^2$ in the time to produce the final ensemble of samples with preconditioning is due to the matrix-vector multiplications necessary for a linearly preconditioned MCMC algorithm.

One can use the results in \citet{cheng2018} to show that the unadjusted underdamped Langevin algorithm with step size $h\leq 1$ is a $W_2$ contraction to $\pi$ with parameters $\Gamma = 4$, $\gamma = h / (2\kappa)$ and
\begin{equation}
    b = 16\kappa\sqrt{\frac{2\Ee_K}{5}}h\textup{ and }\Ee_K:=26\left(\frac{d}{m} + \Dd^2\right)
\end{equation}
where $\Dd^2$ is an upper bound on the distance squared between the initial state and the mode of $\pi$.

To get the time complexity for the preconditioned Algorithm \ref{alg:preconditioned_Markov_chain} with unadjusted underdamped Langevin as its underlying kernel, we must examine how algorithm specific parameters such as $\Ee_K$ change with preconditioning. Therefore define
\begin{equation}
    \Ee_K^{(\Sigma_\pi)} := 26\left(\frac{2d}{3m_{\Sigma_\pi^{-1}}} + \Dd^2\right)\textup{ and }\Ee_K^{(\Ff)} := 26\left(\frac{d}{2m_{\Ff}} + \Dd^2\right)
\end{equation}
where each $\Dd$ is an upper bound on the distance from the initial state of Algorithm \ref{alg:preconditioned_Markov_chain} step 3. to the mode of the target. We adjust the factors of $d / m$ to make sure that $\Ee_K^{(\Ff)} \leq \Ee_K^{(\widehat{\Ff})}$ and $\Ee_K^{(\Sigma_\pi)} \leq \Ee_K^{(\widehat{\Sigma}_\pi)}$ using Corollaries \ref{cor:conditioning_comparison_covariance_case} and \ref{cor:conditioning_comparison_Fisher_case}.

\begin{thm}\label{thm:total_unadjusted_underdamped_complexities}
    Let $U$ be $m$-strongly convex and $L$-smooth. Let $T$ be the minimum total number of FLOPS to gain an output $\{X_t\}_{t=1}^N$ that is $\sqrt{N}\varepsilon$-approximately IID from $\pi$ in $W_2$.
    \begin{enumerate}[noitemsep]
        \item
        Suppose that
        \begin{equation}
            \varepsilon\leq\min\left\{1536\kappa\sqrt{\frac{2\Ee_K}{5}},\sqrt{3\mathrm{tr}\left(\Sigma_\pi\right)}\right\}.
        \end{equation}
        Then
        \begin{equation}
            T \in \tilde{O}\left(\left(d + \mathtt{G}\right)\kappa^2\sqrt{\Ee_K}\frac{N}{\varepsilon}\right)
        \end{equation}
        for the Algorithm \ref{alg:thinned_Markov_chain} with the unpreconditioned unadjusted underdamped Langevin algorithm as the underlying Markov kernel.
        \item
        Suppose that
        \begin{equation}
            \varepsilon\leq\min\left\{512\kappa\sqrt{\frac{2\Ee_K^{(\Sigma_\pi)}}{5}},\sqrt{2d}\right\}.
        \end{equation}
        and $L\Dd^2\geq 1$. Then $T=T_\mathrm{learn} + T_\mathrm{collect}$ where
        \[
        T_\mathrm{learn} &\in \tilde{O}\left({\frac{d^{3/2}\left(d + \mathtt{G}\right)\kappa^2\sqrt{L\Ee_K}}{\delta}\max\left\{\delta^{-1},K^3_\mathrm{cov}\right\}}\right)\\
        &\qquad\text{ and }\\
        T_\mathrm{collect} &\in \tilde{O}\left(\left(d^2 + \mathtt{G}\right)\kappa_{\Sigma_\pi^{-1}}^2\sqrt{\Ee_K^{(\Sigma_\pi)}}\frac{N}{\varepsilon}\right)
        \]
        with probability $\geq 1 - \delta$ for Algorithm \ref{alg:preconditioned_Markov_chain} with the covariance-based preconditioner and unadjusted underdamped Langevin as the underlying Markov kernel where $K_\mathrm{cov} := m_{\Sigma_\pi^{-1}}^{-1/2}$
        \item Suppose that
        \begin{equation}
            \varepsilon \leq \min\left\{ 512\kappa\sqrt{\frac{2\Ee_K^{(\Ff)}}{5}},\sqrt{\frac{3}{2}d}\right\}
        \end{equation}
        and $\sqrt{dL\kappa}\geq 1$. Then $T=T_\mathrm{learn} + T_\mathrm{collect}$ where
        \[
        T_\mathrm{learn} &\in \tilde{O}\left(\frac{d^{3/2}\left(d + \mathtt{G}\right)\kappa^{5/2}\sqrt{L\Ee_K}}{\delta}K_\mathrm{Fisher}^3\right)\\
        &\qquad\text{ and }\\
        T_\mathrm{collect} &\in \tilde{O}\left(\left(d^2 + \mathtt{G}\right)\kappa_\Ff^2\sqrt{\Ee_K^{(\Ff)}}\frac{N}{\varepsilon}\right)\\
        \]
        with probability $\geq 1 - \delta$ for Algorithm \ref{alg:preconditioned_Markov_chain} with the Fisher preconditioner and unadjusted underdamped Langevin as the underlying Markov kernel where $K_\mathrm{Fisher} := \sqrt{L_\Ff}$.
    \end{enumerate}
    Here $\mathtt{G}$ is the computational complexity of a call to $\nabla\log\pi$.
\end{thm}

For a proof see Section \ref{proof:total_unadjusted_underdamped_complexities}.

\citet{bou-rabee2025} show that a variant of unadjusted HMC with duration $T \leq 1 / (\sqrt{8L})$ and step-size $h \leq T$ is a $W_2$ contraction with parameters $\Gamma = 1$, $\gamma = mT^2 / 6$ and
\begin{equation}
    b = 1704\frac{L^{1/4}\sqrt{d\kappa}}{mT^2}h^{3/2}.
\end{equation}

\begin{thm}\label{thm:total_unadjusted_HMC_complexities}
    Let $U$ be $m$-strongly convex and $L$-smooth. Let $T$ be the minimum total number of FLOPS to gain an output $\{X_t\}_{t=1}^N$ that is $\sqrt{N}\varepsilon$-approximately IID from $\pi$ in $W_2$.
    \begin{enumerate}
        \item Suppose that
        \[
        \varepsilon \leq \min \left\{\frac{40897L^{1/2}\sqrt{d\kappa}}{8^{3/4}m},\sqrt{3\mathrm{tr}(\Sigma_\pi)}\right\}.
        \]
        Then \[
        T \in \tilde{O}\left((d + \mathtt{G})\frac{\kappa^2d^{1/3}}{\sqrt{L}}\frac{N}{\varepsilon^{2/3}}\right)
        \] for Algorithm \ref{alg:thinned_Markov_chain} with the unpreconditioned variant of unadjusted HMC from \citet{bou-rabee2025} as the underlying Markov kernel.
        \item Suppose that
        \[
        \varepsilon\leq \min\left\{\frac{81794}{3\sqrt{6}\times 8^{3/4}}\frac{L_{\Sigma_\pi^{-1}}^{1/2}\sqrt{d\kappa_{\Sigma_\pi^{-1}}}}{m_{\Sigma_\pi^{-1}}},\sqrt{2d}\right\}.
        \]
        Then $T = T_\mathrm{learn} + T_\mathrm{collect}$ where
        \[
        T_\mathrm{learn} &\in \tilde{O}\left(\frac{(d + \mathtt{G})\kappa^2d^{5/3}}{\delta^{2/3}}\max\left\{\frac{1}{\delta}, K_\mathrm{cov}^3\right\}\right)\\
        &\qquad\text{ and }\\
        T_\mathrm{collect} &\in \tilde{O}\left((d^2 + \mathtt{G})\frac{\kappa^2_{\Sigma_\pi^{-1}}d^{1/3}}{\sqrt{L_{\Sigma_\pi^{-1}}}}\frac{N}{\varepsilon^{2/3}}\right)
        \]
        with probability $\geq 1 - \delta$ for Algorithm \ref{alg:preconditioned_Markov_chain} with the covariance-based preconditioner and the variant of unadjusted HMC from \citet{bou-rabee2025} as the underlying Markov kernel where $K_\mathrm{cov}:=m^{-1/2}_{\Sigma_\pi^{-1}}$.
        \item Suppose that
        \[
        \varepsilon\leq \min\left\{\frac{40897\sqrt{2}}{6\times 8^{3/4}}\frac{L_\Ff^{1/2}\sqrt{d\kappa_{\Ff}}}{m_{\Ff}},\sqrt{\frac{3}{2}d}\right\}.
        \]
        Then $T = T_\mathrm{learn} + T_\mathrm{collect}$ where
        \[
        T_\mathrm{learn} &\in \tilde{O}\left(\frac{(d + \mathtt{G})\kappa^{7/3}d^{5/3}}{\delta^{2/3}}K_\mathrm{Fisher}^3\right)\\
        &\qquad\text{ and }\\
        T_\mathrm{collect} &\in \tilde{O}\left((d^2 + \mathtt{G})\frac{\kappa^2_{\Ff}d^{1/3}}{\sqrt{L_{\Ff}}}\frac{N}{\varepsilon^{2/3}}\right)
        \]
        with probability $\geq 1-\delta$ for Algorithm \ref{alg:preconditioned_Markov_chain} with the Fisher preconditioner and the variance of unadjusted HMC from \citet{bou-rabee2025} as the underlying Markov kernel where $K_\mathrm{Fisher} = \sqrt{L_\Ff}$.
    \end{enumerate}
    Here $\mathtt{G}$ is the computational complexity of a call to $\nabla\log\pi$.
\end{thm}

For a proof see Section \ref{proof:total_unadjusted_HMC_complexities}.

\citet{lee2021a} has that the proximal sampler with a restricted Gaussian oracle is a $W_2$ contraction with parameters $\Gamma = 1$, $\gamma = \log(1+mh)$ and $b = 0$ where $h$ is the step size. One must replace oracle access with some actually implementable sampler if one wishes to get a realistic computational complexity. Hence in \citet[Section 4.2]{chen2022a} the authors use a rejection sampler to replace the oracle, which has a random complexity. Therefore the complexities for the proximal sampler in the following result are in expectation.

\begin{thm}\label{thm:total_proximal_complexities}
    Let $U$ be $m$-strongly convex and $L$-smooth. Let $T$ be the expected minimum total number of FLOPS to gain an output $\{X_t\}_{t=1}^N$ that is $\sqrt{N}\varepsilon$-approximately IID from $\pi$ in $W_2$.
    \begin{enumerate}
        \item Suppose that $\varepsilon \leq \sqrt{3\mathrm{tr}(\Sigma_\pi)}$. Then we have that
        \[
        T \in O\left((d + \mathtt{U})d\kappa N\log\varepsilon^{-1}\right)
        \]
        for Algorithm \ref{alg:thinned_Markov_chain} with the proximal sampler of \citet[Corollary 6]{chen2022a} as the underlying Markov kernel.
        \item Suppose that $\varepsilon \leq \sqrt{2d}$. Then $T = T_\mathrm{learn} + T_\mathrm{collect}$ where
        \[
        T_\mathrm{learn} &\in \tilde{O}\left((d + \mathtt{U})d^2 \max\left\{\frac{1}{\delta}, K_\mathrm{cov}^3\right\}\right)\\
        &\qquad\text{ and }\\
        T_\mathrm{collect} &\in O\left((d^2 + \mathtt{U})d \kappa_{\Sigma_\pi^{-1}}N\log\varepsilon^{-1}\right)
        \]
        with probability $\geq 1-\delta$ for Algorithm \ref{alg:preconditioned_Markov_chain} with the covariance-based preconditioner and the proximal sampler of \citet[Corollary 6]{chen2022a} as the underlying Markov kernel where $K_\mathrm{cov} = m^{-1/2}_{\Sigma_\pi^{-1}}$.
        \item Suppose that $\varepsilon \leq \sqrt{(3/2)d}$. Then $T = T_\mathrm{learn} + T_\mathrm{collect}$ where
        \[
        T_\mathrm{learn} &\in \tilde{O}\left((d + \mathtt{U})d^2\kappa K_\mathrm{Fisher}^3\right)\\
        &\qquad\text{ and }\\
        T_\mathrm{collect} &\in O\left((d^2 + \mathtt{U})d\kappa_{\Ff}N\log\varepsilon^{-1}\right)
        \]
        with probability $\geq 1-\delta$ for Algorithm \ref{alg:preconditioned_Markov_chain} with the Fisher preconditioner and the proximal sampler of \citet[Corollary 6]{chen2022a} as the underlying Markov kernel where $K_\mathrm{Fisher} = \sqrt{L_\Ff}$.
    \end{enumerate}
    Here $\mathtt{U}$ is the computational complexity of a call to $\log\pi$.
\end{thm}

For a proof see Section \ref{proof:total_proximal_complexities}. Note that for each iteration of the algorithm one must find the mode of the restricted Gaussian oracle distribution, which is strongly log-concave and smooth. The authors of \citet{chen2022a} omit the complexity of this step, and so do we.

\begin{remark}\label{rem:polynomial_complexity}
    Note that the unpreconditioned iteration complexities to collect the $\sqrt{N}\varepsilon$-approximately IID ensemble of samples are all polynomial in the quantities that characterise the geometry of the target such as $\kappa$, $m$ and $L$:
    \[
    T_\mathrm{ULA} \in \tilde{O}\left(d\frac{\kappa^2}{m} \frac{N}{\varepsilon}\right)\textup{, }T_\mathrm{Under} \in \tilde{O}\left(\kappa^2\sqrt{\Ee_K}\frac{N}{\varepsilon}\right)\textup{, }T_\mathrm{HMC} \in \tilde{O}\left(\frac{\kappa^2d^{1/3}}{\sqrt{L}}\frac{N}{\varepsilon^{2/3}}\right)\textup{, }T_\mathrm{Prox}\in O\left(d\kappa N\log\frac{1}{\varepsilon}\right)
    \]
    where $T_\mathrm{ULA}$ is the time for ULA, $T_\mathrm{Under}$ is the time for unadjusted underdamped, $T_\mathrm{HMC}$ is the time for the variant of unadjusted HMC in \citet{bou-rabee2025} and $T_\mathrm{Prox}$ is the expected time for the proximal sampler of \citet[Corollary 6]{chen2022a}.
    
    When the preconditioner $\widehat{M}\in\RR^{d\times d}$ estimated in the initial run and the preconditioner $M\in\RR^{d\times d}$ it's estimating are related via
    \[
    \|M^{-1/2}\widehat{M}M^{-1/2} - \mathbf{I}_d\| \leq \Delta\textup{ for }\Delta \in [0, 1)
    \]
    the quantities $m_{\widehat{M}}$, $L_{\widehat{M}}$ and $\kappa_{\widehat{M}}$ are only a multiple of $1 + \Delta$ or $1 - \Delta$ away from the corresponding $m_{M}$, $L_{M}$ and $\kappa_{M}$, see Corollaries \ref{cor:conditioning_comparison_covariance_case} and \ref{cor:conditioning_comparison_Fisher_case}. Given its polynomial dependence on these quantities, the iteration complexity of an algorithm preconditioned with $M$ is only a constant multiple (depending on $\Delta$) away from the iteration complexity of the same algorithm preconditioned with $\widehat{M}$. Therefore the $\Delta$ parameter needn't be very small, as we find in our theoretical results in which we set $\Delta = 1/2$. Not also that the degrees of the polynomials concerned are low, and so the constant penalty factor that depends on $\Delta$ will not be overly severe.
\end{remark}

Concretely, comparison of these total complexities depends on the FLOP complexity to evaluate the log-density and its gradient. This quantity is problem-specific. If $\pi$ is Gaussian with arbitrary covariance we have $\mathtt{G} = O(d^2)$ (assuming the Cholesky factorisation of the precision matrix is precomputed). The value of $\mathtt{G}$ may be higher than this: in Bayesian inverse modelling one must integrate a system of differential equations to evaluate $\nabla\log\pi$, see \citet{stuart2010} for an overview.

\section{Technical Overview}

Here we offer sketches of the proofs of the main results stated in Section \ref{section:main_results} and Theorem \ref{thm:total_ULA_complexities} parts 1. and 2. in Section \ref{section:application_to_existing_algorithms}. We omit the sketch proofs of Theorems \ref{thm:total_unadjusted_underdamped_complexities}, \ref{thm:total_unadjusted_HMC_complexities} and \ref{thm:total_proximal_complexities} since they work in a very similar way to that of Theorem \ref{thm:total_ULA_complexities}. We also omit the sketch proof of Theorem \ref{thm:total_ULA_complexities} part 3. because it works similarly to Theorem \ref{thm:total_ULA_complexities} part 2.

\subsection{Proof Sketch of Theorem \ref{thm:Wasserstein_contraction_complexity}}

Let $\mu_N \in \Pp(\RR^{d\times N})$ be the distribution of the output $\{X_t\}_{t=1}^N$ of the thinned Markov chain sampler in Algorithm \ref{alg:thinned_Markov_chain} given by %
\begin{equation}
    \mu_{N}\left(\mathrm{d}x\right):=\mu_{0}K^{t_{\textup{burn}}}\left(\mathrm{d}x_{1}\right)K^{t_{\textup{thin}}}\left(x_{1}\to\mathrm{d}x_{2}\right)\cdots K^{t_{\textup{thin}}}\left(x_{N-1}\to\mathrm{d}x_{N}\right).
\end{equation}
We wish to bound the Wasserstein-2 distance between $\mu_N$ and $\pi^{\otimes N}$ in terms of $t_\mathrm{burn}$ and $t_\mathrm{thin}$. We decompose the Wasserstein-2 distance as follows:
\begin{prop}\label{prop:Wasserstein_decomposition}
    Let $\mu_N \in \Pp(\RR^{d\times N})$ be the distribution of the output $\{X_t\}_{t=1}^N$ of the thinned Markov chain sampler in Algorithm \ref{alg:thinned_Markov_chain} and suppose $\pi$ has a Lebesgue density on $\RR^d$. Then
    \begin{equation}\label{eq:W_2_decomposition}
    W_{2}\left(\pi^{\otimes N},\mu_{N}\right)^{2}\leq\mathbb{E}\left[W_{2}\left(\pi,\mu_{0}K^{t_{\textup{burn}}}\right)^{2}+\sum_{t=1}^{N-1}W_{2}\left(\pi,K^{t_{\textup{thin}}}\left(X_{t}\to\cdot\right)\right)^{2}\right]\,.
    \end{equation}
\end{prop}
The proof, in Section \ref{proof:Wasserstein_contraction_complexity}, iteratively constructs a coupling between $\pi^{\otimes N}$ and $\mu_N$. First a coupling between $\pi$ and $\mu_0K^{t_\mathrm{burn}}$ is posited, and then a conditional coupling between $\pi$ and $K^{t_\mathrm{thin}}(x_1 \to \cdot)$ is posited, then conditional coupling between $\pi$ and $K^{t_\mathrm{thin}}(x_2 \to \cdot)$, etc. The construction is similar to the Knothe-Rosenblatt coupling \citep{carlier2010}.
Next, we apply the Wasserstein-2 contraction of Definition \ref{def:W_2_contraction} that $K$ satisfies. Finally, $t_\mathrm{burn}$ is chosen such that the first term in the expectation in \cref{eq:W_2_decomposition} is less than $\varepsilon^2$ and we $t_\mathrm{thin}$ is chosen such that each of the summands are less than $\varepsilon^2$.

\subsection{Proof Sketch of Theorem \ref{thm:preconditioner_learn_complexity}}

We bound how close the learned preconditioner is to an `ideally learned' preconditioner that is constructed using IID samples from $\pi$. This can be done if the samples used to construct the learned preconditioner are $\sqrt{N}\varepsilon$-approximately IID from $\pi$ in $W_2$. Then we bound how close the `ideally learned' preconditioner is to the preconditioner $\Ff$ or $\Sigma_\pi$ using standard empirical covariance concentration results, namely \citep[Exercise 4.7.3]{vershynin2018}. We choose values of $\varepsilon$ and $N$ to satisfy conditions \cref{eqn:preconditioner_estimate_goodness_condition_covariance} or \cref{eqn:preconditioner_estimate_goodness_condition_Fisher}.

\subsubsection{Proof Sketch of Theorem \ref{thm:preconditioner_learn_complexity} Part 1.}\label{subsubsection:proof_sketch_of_covariance_learn}

Let $\{X_t\}_{t=1}^N$ be the $\sqrt{N}\varepsilon$-AIID output of Algorithm \ref{alg:thinned_Markov_chain}, and let $X\in\RR^{d\times N}$ be the matrix whose $j$th column is $X_j - \mu_\pi$ for $j\in [N]$ where $\mu_\pi:=\EE_\pi[Y]$. Then
\begin{equation}
    \widehat{\Sigma}_\pi = \frac{1}{N}XX^\top - \left(\bar{X} - \mu_\pi\right)\left(\bar{X} - \mu_\pi\right)^\top.
\end{equation}
Define $\{Y_t\}_{t=1}^N \sim \pi^{\otimes N}$ as the `ideal' IID sample from $\pi$, and let $Y \in \RR^{d\times N}$ be the matrix whose $j$th column is $Y_j - \mu_\pi$ for $j \in [N]$. Then we have
\begin{equation}\label{eq:covariance_bound_decomposition}
    \begin{split}
        \|\Sigma_\pi^{-1/2}\widehat{\Sigma}_\pi\Sigma_\pi^{-1/2}-\mathbf{I}_d\|&\leq \|\frac{1}{N}\Sigma_\pi^{-1/2}XX^\top\Sigma_\pi^{-1/2} - \frac{1}{N}\Sigma_\pi^{-1/2}YY^\top\Sigma_\pi^{-1/2}\| \\&\qquad+ \|\frac{1}{N}\Sigma_\pi^{-1/2}YY^\top\Sigma_\pi^{-1/2} - \mathbf{I}_d\|+ \|\Sigma_\pi^{-1/2}(\bar{X} - \EE_\pi[X])\|^2.
    \end{split}
\end{equation}
Note that all $\RR^d$-valued random variables in the above expression are pre-multiplied by $\Sigma_\pi^{-1/2}$. Therefore we can work in the space transformed to via $\Sigma_\pi^{-1/2}$, adjusting constants as necessary. For instance, the $\sqrt{N}\varepsilon$-AIID condition transforms by multiplying $\varepsilon$ by the Lipschitz constant of premultiplying by $\Sigma_\pi^{-1/2}$, see Proposition \ref{prop:rootNIID_Lipschitz_transformation}. Pre-multiplying by $\Sigma_\pi^{-1/2}$ is $\beta_{\Sigma_\pi}^{-1/2}$-Lipschitz and hence we pick up a factor of $\beta_{\Sigma_\pi}^{-1/2}$ on $\varepsilon$. To bound the first term on the right-hand side of \cref{eq:covariance_bound_decomposition} we have Proposition \ref{prop:rootNIID_implies_empirical_covariance_bound}. The proof works by using the decomposition:
\begin{equation}
    \frac{1}{N}XX^\top - \frac{1}{N}YY^\top = \frac{1}{N}\left(X-Y\right)Y^\top + \frac{1}{N}\left(X-Y\right)\left(X - Y\right)^\top + \frac{1}{N}Y\left(X - Y\right)^\top
\end{equation}
where $X,Y\in\RR^{d\times N}$ are matrices with columns $X_t$ and $Y_t$ respectively. Taking the Frobenius norm and expectation, we bound the $X - Y$ terms above using the AIID condition (with the updated $\varepsilon$ caused by transforming by $\Sigma_\pi^{-1/2}$) and we can bound the $Y$ terms using $\mathrm{tr}(\Sigma_\pi)$.

The second term on the right hand side of \cref{eq:covariance_bound_decomposition} is bounded via a result on concentration of empirical covariances estimated from IID samples \citep[Exercise 4.7.3]{vershynin2018}. 
\begin{equation}
    \|\frac{1}{N}\Sigma_\pi^{-1/2}YY^\top\Sigma_\pi^{-1/2} - \mathbf{I}_d\| \leq CK_\mathrm{cov}^2\rbra[3]{\sqrt{\frac{d+u}{N}} + \frac{d+u}{N}}
\end{equation}
for all $u\geq0$ with probability at least $1 - 2\exp(-u)$, where $K_\mathrm{cov} \propto \sqrt{1 / m_{\Sigma_\pi^{-1}}}$.

For the third term on the right-hand side of \cref{eq:covariance_bound_decomposition} we use Proposition \ref{prop:L2_rootNepsilon_mean_error}.
In summary we have upper bounds on the expectations of the first and third terms of the right-hand side of \cref{eq:covariance_bound_decomposition} and we have a high probability bound on the second term. Using Markov's inequality and a union bound, these can be combined to yield a bound of the form
\begin{equation}
    \PP\left(\|\Sigma_\pi^{-1/2}\widehat{\Sigma}_\pi\Sigma_\pi^{-1/2}-\mathbf{I}_d\| \leq f\left(u,N,\varepsilon\right)\right) \geq 1 - g\left(u,N,\varepsilon\right)
\end{equation}
where $f(u,N,\varepsilon)$ is increasing in $u$ and $g(u, N,\varepsilon)$ is decreasing in $u$. Both functions are explicit in $N$ and $\varepsilon$. We choose a sensible value for $u$, and then choose
\begin{equation}
    \varepsilon =\frac{\sqrt{2}}{120}\frac{\delta\Delta}{\sqrt{d}}\sqrt{\beta_{\Sigma_\pi}}\textup{ and }N=\max\left\{ \frac{5d}{\delta\Delta},\frac{2CK_\mathrm{cov}^{2}\left(d+\log\left(4/\delta\right)\right)\sqrt{CK_\mathrm{cov}^{2}+2\Delta}}{\Delta}\right\} .
\end{equation}to satisfy $\PP(\|\Sigma_\pi^{-1/2}\widehat{\Sigma}_\pi\Sigma_\pi^{-1/2}-\mathbf{I}_d\| \leq \Delta) \geq 1 - \delta$.

\subsubsection{Proof Sketch of Theorem \ref{thm:preconditioner_learn_complexity} Part 2.}\label{}

The proof for the Fisher preconditioner works similarly to the proof for the covariance preconditioner in Section \ref{subsubsection:proof_sketch_of_covariance_learn}. Let $\{X_t\}_{t=1}^N$ be the output of the thinned Markov chain sampler in Algorithm \ref{alg:thinned_Markov_chain}. Let $S \in \RR^{d \times N}$ be the matrix whose $j$th column is $\nabla\log\pi(X_j)$ for $j\in[N]$ and let $T \in \RR^{d\times N}$ be the matrix whose $j$th column is $\nabla\log\pi(Y_t)$ with $\{Y_t\}_{t=1}^N\sim\pi^{\otimes N}$. We already know that the mean of each column of $S$ in stationarity is 0 and hence we needn't analyse the difference in means term in \cref{eq:covariance_bound_decomposition}. Instead we have
\begin{equation}\label{eq:Fisher_bound_decomposition}
    \begin{split}
        \|\Ff^{-1/2}\widehat{\Ff}\Ff^{-1/2} - \mathbf{I}_d\| &\leq \|\frac{1}{N}\Ff^{-1/2}SS^\top\Ff^{-1/2} - \frac{1}{N}\Ff^{-1/2}TT^\top\Ff^{-1/2}\|\\
        &\qquad + \|\frac{1}{N}\Ff^{-1/2}TT^\top\Ff^{-1/2} - \mathbf{I}_d\|
    \end{split}
\end{equation}

The potential of $\pi$ is $L$-smooth and so $\nabla\log\pi$ is $L$-Lipschitz. Hence $\Ff^{-1/2}\nabla\log\pi$ is $L / \sqrt{\alpha_\Ff}$-Lipschitz (using the fact that $\alpha_\Ff\mathbf{I}_d \preceq \Ff$) and so, according to Proposition \ref{prop:rootNIID_Lipschitz_transformation} we pick up a factor of $L / \sqrt{\alpha_\Ff}$ on $\varepsilon$. Note the difference with the covariance preconditioner in which we picked up a factor of $1 / \sqrt{\beta_{\Sigma_\pi}}$. This is where the additional factor of $L$ comes from when we move from the time complexity to learn the covariance in Theorem \ref{thm:preconditioner_learn_complexity} part 1 to the time complexity to learn the Fisher preconditioner in Theorem \ref{thm:preconditioner_learn_complexity} part 2. We believe that in some cases this can be avoided. Take $\pi = \Nn(0,\Sigma_\pi)$ for instance. Pre-multiplying by $\Sigma_\pi^{-1/2}$ is still $1 / \sqrt{\beta_{\Sigma_\pi}}$-Lipschitz and $\Ff^{-1/2}\nabla\log\pi(x) = \Sigma_\pi^{-1/2}x$ is \emph{also} $1 / \sqrt{\beta_{\Sigma_\pi}}$-Lipschitz. In any case we use Proposition \ref{prop:rootNIID_implies_empirical_covariance_bound} to bound the expectation of the first term on the right of \cref{eq:Fisher_bound_decomposition}. The rest of the proof works similarly to the covariance case.

\subsection{Proof Sketch of Theorem \ref{thm:total_ULA_complexities} Parts 1. and 2.}\label{subsection:proof_sketch_of_total_ULA_complexity}

From \citet[Theorem 1a)]{dalalyan2019}, for step size $h \leq 2 / (L + m)$, the ULA kernel is a $(\Gamma, \gamma, b)$-contraction to $\pi$ with $\Gamma = 1$, $\gamma = -\log(1 - mh)$ and $b = (33/20)\kappa\sqrt{dh}$. To make things simpler we use $\gamma = mh$ since this value serves as a lower bound to $\gamma$ as initially stated. The upshot of this is that we can apply Theorems \ref{thm:Wasserstein_contraction_complexity} and \ref{thm:preconditioner_learn_complexity}.

\subsubsection{Proof Sketch of Theorem \ref{thm:total_ULA_complexities} Part 1.}\label{subsubsection:proof_sketch_unpreconditioned_ULA_complexity}

Choose $h = 100^{-1}d^{-1}\kappa^{-2}\varepsilon^2$. The upper bound $\varepsilon \leq 10\kappa\sqrt{d}L^{-1/2}$ makes $h\leq2 / (L + m)$. Therefore the unpreconditioned ULA is a $(\Gamma, \gamma, b)$-$W_2$ contraction to $\pi$ by \citet[Theorem 1a)]{dalalyan2019} with the values of $\Gamma$, $\gamma$, and $b$ stated above. The choice of $h$ also ensures $3\Gamma^2b<\varepsilon$, satisfying the lower bound constraint from Theorem \ref{thm:Wasserstein_contraction_complexity}. The upper bound $\varepsilon \leq \sqrt{3\mathrm{tr}(\Sigma_\pi)}$ ensures we can apply Theorem \ref{thm:Wasserstein_contraction_complexity}.

\subsubsection{Proof Sketch of Theorem \ref{thm:total_ULA_complexities} Part 2.}\label{subsubsection:proof_sketch_ULA_complexity_covariance_preconditioned}

We only give the covariance case here because the Fisher case works analogously.

To collect the samples with which we construct the preconditioner we use Algorithm \ref{alg:preconditioned_Markov_chain} with ULA as the underlying kernel. Let $\Delta,\delta \in (0,1)$ set the precision and probability with which we learn the preconditioner. With the aim of outputting a suitable $\sqrt{N}\varepsilon$-AIID sample from which to construct the preconditioner, we let $\varepsilon = \frac{\sqrt{2}}{120}\frac{\delta\Delta}{\sqrt{dL}}$
in accordance with Theorem \ref{thm:preconditioner_learn_complexity}, where we note that, under $L$-smoothness of the potential, $\beta_{\Sigma_\pi} \geq L^{-1}$ satisfying the assumption of the Theorem.

Set $h_1 = k\varepsilon^2$ and set an upper bound on $k$ such that $h_1 \leq 2 / (L + m)$. Then the unpreconditioned ULA kernel is a $(\Gamma, \gamma, b)$-$W_2$ contraction to $\pi$ by \citet[Theorem 1a)]{dalalyan2019} with the values of $\Gamma$, $\gamma$, and $b$ stated above. The bias $b$ is parameterised by $h_1$ so choose $k$ to ensure that $3\Gamma^2b = (1/2)\varepsilon < \varepsilon$. One may check that the resulting value of $k$ satisfies the upper bound on it stated above. The fact that $\delta\Delta < 1$ and $\mathrm{tr}(\Sigma_\pi) \geq L^{-1}d$ automatically give $\varepsilon\leq \sqrt{3\mathrm{tr}(\Sigma_\pi)}$. Therefore Theorem \ref{thm:Wasserstein_contraction_complexity} may be applied, yielding an iteration complexity for the thinned ULA chain to output a $\sqrt{N}\varepsilon$-AIID ensemble. Finally choose
\begin{equation}
    N = \max\cbra[2]{5\delta^{-1}\Delta^{-1}d,\ 2C\Delta^{-1}K_\mathrm{cov}^{2}\left(d+\log\left(4/\delta\right)\right)\sqrt{CK_\mathrm{cov}^{2}+2\Delta}}
\end{equation}
in line with Theorem \ref{thm:preconditioner_learn_complexity}. For these $\varepsilon$ and $N$, Theorem \ref{thm:preconditioner_learn_complexity} implies
\begin{equation}
            \norm[1]{\Sigma_\pi^{-1/2}\widehat{\Sigma}_\pi\Sigma_\pi^{-1/2} - \mathbf{I}_d}\leq \Delta\textup{ with probability }\geq 1-\delta.
\end{equation}
Choosing $\Delta = 1/2$ gives the preconditioner learning complexity of $\tilde{O}(d^3(d + \mathtt{G})\kappa^3\delta^{-2}\max\left\{\delta^{-1},K^3_\mathrm{cov}\right\})$ after substituting in the values of $N$, $\gamma$, $\Gamma$, $b$, and the computational costs of each iteration.

Now that the preconditioner has been constructed we can collect the $\sqrt{N}\varepsilon$-AIID output $\{X_t\}_{t=1}^N$. We are sampling using the preconditioned kernel $K_{\widehat{\Sigma}_\pi^{-1}, \pi}$ with the preconditioner learned in step 1. of Algorithm \ref{alg:preconditioned_Markov_chain}. However \citet{dalalyan2019} only give us the conditions under which $K_{(\widehat{\Sigma}_\pi^{-1/2})_\sharp\pi}$ (i.e. the unpreconditioned kernel on a preconditioned target) is a $W_2$ contraction to $(\widehat{\Sigma}_\pi^{-1/2})_\sharp\pi$. Therefore we use Corollary \ref{cor:approximate_preconditioned_kernel_efficiency} to get the contraction coefficients for $K_{\widehat{\Sigma}_\pi^{-1}, \pi}$ from those of $K_{(\widehat{\Sigma}_\pi^{-1/2})_\sharp\pi}$. No change is made to the contraction parameter $\gamma$ and hence the two kernels $K_{(\widehat{\Sigma}_\pi^{-1/2})_\sharp\pi}$ and $K_{\widehat{\Sigma}_\pi^{-1}, \pi}$ end up having the same complexity (in FLOPS) up to logarithmic factors.

As mentioned we first assert that $K_{(\widehat{\Sigma}_\pi^{-1/2})_\sharp\pi}$ is a contraction to $(\widehat{\Sigma}_\pi^{-1/2})_\sharp\pi$. Similarly to the proof of the complexity of the unpreconditioned kernel we choose $h = 100^{-1}d^{-1}\tilde{\kappa}^{-2}\varepsilon^2$ where $\tilde{\kappa}$ is the condition number of $(\widehat{\Sigma}_\pi^{-1/2})_\sharp\pi$. We wish to verify that
\begin{equation}
    \varepsilon \leq \frac{10\kappa_{\widehat{\Sigma}_\pi^{-1}}\sqrt{d}}{\sqrt{L_{\widehat{\Sigma}_\pi^{-1}}}}
\end{equation}
in order to achieve $h \leq 2 / (L_{\widehat{\Sigma}_\pi^{-1}} + m_{\widehat{\Sigma}_\pi^{-1}})$ to get a $(\Gamma, \gamma, b)$-$W_2$ contraction to $(\widehat{\Sigma}_\pi^{-1/2})_\sharp\pi$ of $K_{(\widehat{\Sigma}_\pi^{-1/2})_\sharp\pi}$. Corollary \ref{cor:conditioning_comparison_covariance_case} gives us that $L_{\widehat{\Sigma}_\pi^{-1}} \leq (1-\Delta)^{-1}L_{{\Sigma}_\pi^{-1}}$ and $\kappa_{\Sigma_\pi^{-1}} \leq (1 + \Delta)(1 - \Delta)^{-1}\kappa_{\widehat{\Sigma}_\pi^{-1}}$ with $\Delta = 0.5$, and so the requirement
\begin{equation}
    \varepsilon\leq \frac{10}{3\sqrt{2}}\frac{\kappa_{\Sigma_\pi^{-1}}\sqrt{d}}{\sqrt{L_{\Sigma_\pi^{-1}}}}
\end{equation}
from the Theorem statement gives us the upper bound on $\varepsilon$ we wish to verify. Therefore $K_{(\widehat{\Sigma}_\pi^{-1/2})_\sharp\pi}$ is a $(\Gamma,\gamma,b)$-$W_2$ contraction to $(\widehat{\Sigma}_\pi^{-1/2})_\sharp\pi$ with
\[
\Gamma = 1\textup{, }\gamma = m_{\widehat{\Sigma}_\pi^{-1}}h\textup{ and }b = \frac{33}{20}\kappa_{\widehat{\Sigma}_\pi^{-1}}.
\]
Corollaries \ref{cor:conditioning_comparison_covariance_case} and \ref{cor:approximate_preconditioned_kernel_efficiency} then give us that $K_{\widehat{\Sigma}_\pi^{-1}, \pi}$ is a $(\Gamma, \gamma, b)$-$W_2$ contraction to $\pi$ with
\[
\Gamma = \sqrt{3\kappa(\Sigma_\pi)}\textup{, }\gamma = \frac{2}{3}m_{\Sigma_\pi^{-1}}h\textup{ and }b =3\sqrt{2}\|\Sigma_\pi^{1/2}\| \frac{33}{20}\kappa_{\Sigma_\pi^{-1}}.
\]
Finally to apply Theorem \ref{thm:Wasserstein_contraction_complexity} we wish to verify that $\varepsilon \leq \sqrt{3\mathrm{tr}(\Sigma_{\tilde{\pi}})}$ where $\tilde{\pi}$ is the pushforward of $\pi$ through $\widehat{\Sigma}_\pi^{-1/2}$. For this we can use the fact that $\widehat{\Sigma}_\pi$ and $\Sigma_\pi$ are quantifiably close:
    \begin{equation}\begin{split}
\mathrm{tr}\left(\Sigma_{\tilde{\pi}}\right) & =\mathrm{tr}\left(\widehat{\Sigma}_{\pi}^{-1/2}\Sigma_{\pi}\widehat{\Sigma}_{\pi}^{-1/2}\right)\\
 & =\sum_{i=1}^{d}\lambda_{i}\left(\widehat{\Sigma}_{\pi}^{-1/2}\Sigma_{\pi}\widehat{\Sigma}_{\pi}^{-1/2}\right)\\
 & =\sum_{i=1}^{d}\lambda_{i}\left(\widehat{\Sigma}_{\pi}^{1/2}\Sigma_{\pi}^{-1}\widehat{\Sigma}_{\pi}^{1/2}\right)^{-1}\\
 & =\sum_{i=1}^{d}\lambda_{i}\left(\Sigma_{\pi}^{-1/2}\widehat{\Sigma}_{\pi}\Sigma_{\pi}^{-1/2}\right)^{-1}\\
 & \geq\frac{1}{1+\Delta}d\\
 & =\frac{2}{3}d.
\end{split}\end{equation}
Therefore the condition $\varepsilon \leq \sqrt{2d}$ from the Theorem statement implies $\varepsilon \leq \sqrt{3\mathrm{tr}(\Sigma_{\tilde{\pi}})}$. Applying Theorem \ref{thm:Wasserstein_contraction_complexity} gives a time (in iterations) to collect the $\sqrt{N}\varepsilon$-AIID samples as
\[
T_\mathrm{collect} \in \tilde{O}\left(\frac{1}{m_{\Sigma_\pi^{-1}}}d\kappa^2_{\widehat{\Sigma}_\pi^{-1}}\frac{N}{\varepsilon^2}\right)
\]
and hence applying Corollary \ref{cor:conditioning_comparison_covariance_case} to transform $\kappa_{\widehat{\Sigma}_\pi^{-1}}$ to $\kappa_{\Sigma_\pi^{-1}}$ gives the result.

\section{\texorpdfstring{The $\sqrt{N}\varepsilon$-approximately IID from $\pi$ in $W_2$ condition}{The ε√N-approximately IID from π in W₂ condition}}\label{section:rootNepsilonAIID}

The $\sqrt{N}\varepsilon$-AIID condition in Definition \ref{defn:approximately_IID_condition} bridges the gap between traditional measures of MCMC performance, such as asymptotic variance, and contemporary measures such as mixing time. 
We present implications of the $\sqrt{N}\varepsilon$-AIID condition that are informative as to the quality of MCMC output. Proofs of these results can be found in Section \ref{subsection:proofs_for_rootNepsilonIID}. 

Firstly, the condition transforms stably under Lipschitz perturbations:

\begin{prop}
\label{prop:rootNIID_Lipschitz_transformation}If $\{ X_{t}\} _{t=1}^{N}\in\mathbb{R}^{d\times N}$
be $\sqrt{N}\varepsilon$-AIID from $\pi$ in $W_{2}$ and $F:\mathbb{R}^{d}\to\mathbb{R}^{m}$ is $L$-Lipschitz,
then $\{ F(X_{t})\} _{t=1}^{N}$
is $\sqrt{N}L\varepsilon$-approximately IID from $F_{\sharp}\pi$
in $W_{2}$.
\end{prop}
Thus, we may replace $\{X_t\}_{t=1}^N$ with $\{F(X_t)\}_{t=1}^N$ where $F$ is Lipschitz, adjusting $\varepsilon$ and $\pi$ accordingly.
Monte Carlo algorithms produce estimators that are averages:
$
    \overline{X}:=\frac{1}{N}\sum_{t=1}^NX_t.
$
That $\{X_t\}_{t=1}^N$ is $\sqrt{N}\varepsilon$-AIID informs us on the quality of these averages via their $L^2$ error:
\begin{prop}\label{prop:L2_rootNepsilon_mean_error}
    Let $\{X_t\}_{t=1}^N\in \RR^{d\times N}$ be $\sqrt{N}\varepsilon$-AIID from $\pi$ in $W_2$ where $\pi$ has a finite covariance matrix $\Sigma_\pi\in\RR^{d\times d}$ and mean $\mu_\pi\in\RR^d$. Then $\EE[\|\overline{X}-\mu_\pi\|^2] \leq (\varepsilon + \sqrt{N^{-1}\mathrm{tr}(\Sigma_\pi)})^2$.
\end{prop}
and their variance:
\begin{prop}\label{prop:rootNepsilon_bounded_estimator_variance}
    Let $\{X_t\}_{t=1}^N\in \RR^{d\times N}$ be $\sqrt{N}\varepsilon$-approximately IID from $\pi$ in $W_2$ where $\pi$ has a finite covariance matrix $\Sigma_\pi\in\RR^{d\times d}$ and mean $\mu_\pi\in\RR^d$. Then
    \begin{equation}
        \mathrm{tr}\left(\mathrm{Var}\left(\overline{X}\right)\right) \leq \frac{1}{N}\left(\sqrt{\mathrm{tr}\left(\Sigma_\pi\right)} + \varepsilon\right)^2 + \left(2\sqrt{\mathrm{tr}\left(\Sigma_\pi\right)} + \varepsilon\right)\varepsilon.
    \end{equation}
\end{prop}
We can also compare the empirical distributions of $\{X_t\}_{t=1}^N$ and of an IID sample from $\pi$:
\begin{prop}\label{prop:empirical_distribution_distance_rootNepsilonIID}
    If $\{X_t\}_{t=1}^N\in \RR^{d\times N}$ is $\sqrt{N}\varepsilon$-AIID from $\pi$ in $W_2$ and $\{Y_t\}_{t=1}^N\distas \pi^{\otimes N}$, then the $W_2$-optimal coupling between $\{X_t\}_{t=1}^N$ and $\{Y_t\}_{t=1}^N$ satisfies:
    \begin{equation}
        \EE\left[W_2\left(\frac{1}{N}\sum_{t=1}^N\delta_{X_t},\frac{1}{N}\sum_{t=1}^N\delta_{Y_t}\right)\right] \leq \varepsilon
    \end{equation}
\end{prop}
Finally, the $\sqrt{N}\varepsilon$-IID condition allows one to bound the distance between empirical covariances constructed with $\{X_t\}_{t=1}^N$ and those constructed with $\{Y_t\}_{t=1}^N \sim \pi^{\otimes N}$:

\begin{prop}
\label{prop:rootNIID_implies_empirical_covariance_bound}
Suppose that $\pi$ is mean zero. 
If $\{X_t\}_{t=1}^N\in \RR^{d\times N}$ is $\sqrt{N}\varepsilon$-AIID from $\pi$ in $W_2$ and $\{Y_t\}_{t=1}^N\distas \pi^{\otimes N}$, then the $W_2$-optimal coupling between $\{X_t\}_{t=1}^N$ and $\{Y_t\}_{t=1}^N$ has:
\[
\mathbb{E}\left[\left\Vert \frac{1}{N}\sum_{t=1}^{N}X_{t}X_{t}^{\top}-\frac{1}{N}\sum_{t=1}^{N}Y_{t}Y_{t}^{\top}\right\Vert _{F}\right]\leq2\sqrt{\mathrm{tr}\left(\Sigma_{\pi}\right)}\varepsilon+\varepsilon^{2}.
\]
\end{prop}

Note that all results recover the correct corresponding statements about the classical Monte Carlo estimator that uses independent samples as $\varepsilon \searrow 0$.

\section{Considerations for Practical Implementation}

\subsection{When Should one Learn and Use a Preconditioner?}\label{subsection:when_to_use_a_preconditioner}

The time complexities in Theorems \ref{thm:total_ULA_complexities}, \ref{thm:total_unadjusted_underdamped_complexities}, \ref{thm:total_unadjusted_HMC_complexities} and \ref{thm:total_proximal_complexities} all follow a similar pattern. The complexity to achieve a $\sqrt{N}\varepsilon$-approximately IID sample with an unpreconditioned algorithm is linear in $N$, whereas the time complexity for the preconditioned algorithms is affine in $N$. This pattern holds even if one includes the log factors omitted in the statements of the results. Therefore if one knew the problem dependent quantities such as $\kappa$ and $d$ one could use the results to calculate whether one should use the unpreconditioned Algorithm \ref{alg:thinned_Markov_chain} or the preconditioned Algorithm \ref{alg:preconditioned_Markov_chain}. However we stress that this calculation is problem dependent and will hence have no universal answer.

As a demonstration of the kind of calculation a user might undertake, in Figure \ref{fig:time_comparison} we plot the total complexity $T$ (excluding log factors) as a function of the number of approximately IID samples $N$ for the unpreconditioned algorithm and for the Fisher preconditioned algorithm. We use the proximal sampler with a rejection component of \cite[Section 4.2]{chen2022a} as the example kernel, and assume that the cost in FLOPs of evaluating the log-target density is $O(d^2)$.

\begin{figure}
    \centering
    \includegraphics[scale = 1]{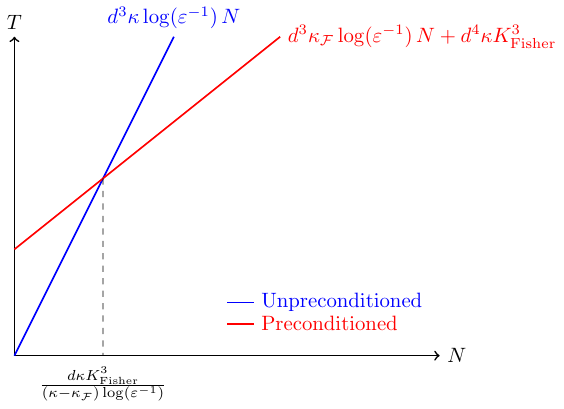}
    \caption{A comparison based on Theorem \ref{thm:total_proximal_complexities} of the time complexity $T$ (up to log factors) of the time to achieve a $\sqrt{N}\varepsilon$-approximately IID ensemble using an unpreconditioned proximal sampler (blue) versus that of the Fisher preconditioned proximal sampler (red) assuming that the cost in FLOPS of a log-density evaluation is $O(d^2)$.}
    \label{fig:time_comparison}
\end{figure}

The vertical dashed line indicates the value of $N$ (up to log factors) Theorem \ref{thm:total_proximal_complexities} Part 3. suggests that using Algorithm \ref{alg:preconditioned_Markov_chain} with a Fisher preconditioner will be more efficient that using the unpreconditioned Algorithm \ref{alg:thinned_Markov_chain}. Thus the diagram and the results show that, if the target is preconditionable in such a way that using $\Sigma_\pi^{-1}$ or $\Ff$ will improve its conditioning, there will always be some value of $N$ above which preconditioning will be more effective than not.

\subsection{When Can One Improve the Conditioning with the Covariance and Fisher Preconditioners?}\label{section:improving_the_conditioning}

The time complexities in Theorems \ref{thm:total_ULA_complexities}, \ref{thm:total_unadjusted_underdamped_complexities}, \ref{thm:total_unadjusted_HMC_complexities} and \ref{thm:total_proximal_complexities} are written in terms of $\kappa_{\Sigma_\pi^{-1}}$ and $\kappa_\Ff$ i.e. the condition number after preconditioning with the covariance and Fisher matrices. Therefore we would like to know when these condition numbers are guaranteed to be smaller than $\kappa$. The following result gives an example of a form of target distribution under which preconditioning with $\Ff$ will improve the condition number:

\begin{lem}\label{lem:Fisher_improves_multiplicative_Hessian}
Let the potential of the target be such that $\nabla^{2}U\left(x\right)=Z^{\top}\Lambda\left(x\right)Z$
where $Z\in\mathbb{R}^{n\times d}$ for some $n\in\mathbb{N}$ and
assume that $\Lambda\left(x\right)$ is diagonal and that there exist
$C\geq c>0$ such that $\sup_{x\in\mathbb{R}^{d}}\lambda_{i}\left(\Lambda\left(x\right)\right)=C$
and $\inf_{x\in\mathbb{R}^{d}}\lambda_{i}\left(\Lambda\left(x\right)\right)=c$
for all $i\in\left[d\right]$. Then
\[
\kappa=\frac{C}{c}\kappa\left(Z^{\top}Z\right)
\]
and
\[
\kappa_{\mathcal{F}}\leq\frac{C}{c}\kappa\left(\mathbb{E}_{\pi}\left[\Lambda\left(X\right)\right]\right)\leq\frac{C^{2}}{c^{2}}.
\]
\end{lem}

For a proof see \cref{proof:Fisher_improves_multiplicative_Hessian}. Therefore if $\kappa(\EE_\pi[\Lambda(X)])$ or $C / c$ is less than $\kappa(Z^\top Z)$ the Fisher preconditioner reduces the condition number. The Hessian of the potential of the target assumes this form when, for instance, the potential can be written as
\[\label{eq:generalised_linear_potential}
U(x) = \sum_{k=1}^n \ell_{Y_k}(Z_k^\top x) + \frac{\lambda}{2}(x-\mu)^\top Z^\top WZ(x-\mu)
\]
where $\{(Y_k, Z_k)\}_{k = 1}^n$ are observations with $Y_k \in \RR$ and $Z_k \in \RR^d$ for all $k\in[n]$ and $W \in \RR^{n \times n}$ is diagonal. Such a potential is typical in generalised linear models that use the generalised $g$-prior of \citet{hanson2014a}. In that case $\ell_{Y_k}(Z_k^\top x)$ is the negative log-likelihood associated with observation $k$ and the prior is $\Nn(\mu,(\lambda Z^\top W Z)^{-1})$. Clearly a Gaussian distribution has a potential in the form of \cref{eq:generalised_linear_potential}, although in that case we have $\kappa_{\Sigma_\pi^{-1}} = \kappa_\Ff = 1$.

For the covariance preconditioner we may use the following result:

\begin{lem}\label{lem:covariance_condition_bounded_by_Fisher}
Let the potential of the target be strongly convex and smooth. We have
\[
\kappa_{\Sigma^{-1}_{\pi}}\leq\lambda_{1}\left(\mathcal{F}^{1/2}\mathbb{E}_{\pi}\left[\nabla^{2}U\left(X\right)^{-1}\right]\mathcal{F}^{1/2}\right)\kappa_{\mathcal{F}}.
\]
\end{lem}

For a proof see \cref{proof:covariance_condition_bounded_by_Fisher}. Therefore if the Fisher matrix is a good preconditioner and $\EE_\pi[\nabla^2U]$ is close to $\EE_\pi[\nabla^2U^{-1}]^{-1}$ then the covariance is a good preconditioner.

\section{Discussion}
In the present work we provided the first, to our knowledge, non-asymptotic characterization of costs and benefits of learning a preconditioner for use in MCMC. 
This is in contrast to prior work that is typically only able to show that preconditioning, or adapting parameters, in MCMC makes the efficiency only not too much worse, and (perhaps) is asymptotically more efficient.
We demonstrated that preconditioners based on empirical covariance or empirical Fisher information matrices only need to be estimated within moderate $O(1)$ accuracy in order for the benefits of preconditioning to be realized. We showed that this estimation can be performed efficiently, and that the computational cost of estimation can be amortized when a large number of approximately IID samples from the target distribution are required. Thus when a large number of samples are required, one typically benefits from preconditioning. This mirrors practice as, for example, software packages like \texttt{stan} include preconditioner learning as a default tactic.

We also introduced the approximate IID condition used in our derivations. We believe this is a useful tool as it allows one to work simultaneously with errors due to approximation and non-asymptotics in MCMC analysis. We believe that this framework can be used to study adapative MCMC methods more generally---the present work can be seen as adaptive MCMC with a single adaptation of a special parameter (the preconditioner). In future work we hope to apply it more broadly, both a) to tune other types of hyperparameters such as the scale/step size of methods like ULA or MALA, the integration time for HMC, etc, and b) to extend it to settings more typical of practice with continuous or periodic adaptation.

A limitation of the present work is that our results are based on only the $W_2$-contraction assumption. It has been shown that ULA has an invariant distribution which admits a log-Sobolev inequality. Recently \citet{nakakita2026} used this to demonstrate that, for covariance estimation using ULA, tighter dependence on the failure probability ($\delta$ in Theorem \ref{thm:preconditioner_learn_complexity}) is possible. While true for ULA, this does not follow from the $W_2$-contraction assumption in general. Their result can be applied to tighten the estimation cost contribution of our covariance-preconditioning bounds for ULA. If we coupled the $W_2$-contraction assumption with a stronger assumption on the invariant distribution of the underlying chain, then we expect to obtain bounds with tighter failure probabilities for estimating both covariance and Fisher information estimation, for a wider range of algorithms.

\bibliographystyle{plainnat}
\bibliography{bibliography}

\appendix

\section{Results Concerning the Condition Number and Linear Preconditioning}\label{section:results_concerning_the_condition_number}

We show the effect of preconditioning with the estimates $\widehat{\Sigma}_\pi^{-1}$ and $\widehat{\Ff}$ are is within a constant of that for $\Sigma_\pi^{-1}$ and $\Ff$ respectively. Per equations \cref{eqn:preconditioner_estimate_goodness_condition_covariance} and \cref{eqn:preconditioner_estimate_goodness_condition_Fisher}, good covariance and Fisher preconditioners will have, with high probability,
\begin{equation}\label{eq:good_preconditioners}
    \left\|\Sigma_\pi^{-1/2}\widehat{\Sigma}_\pi\Sigma_\pi^{-1/2} - \mathbf{I}_d\right\|\leq \Delta\textup{ and }\left\|\Ff^{-1/2}\widehat{\Ff}\Ff^{-1/2} - \mathbf{I}_d\right\|\leq \Delta.
\end{equation}
for a chosen $\Delta \in [0,1)$. Then
\begin{equation}
\lambda_i\left(\Sigma_\pi^{-1/2}\widehat{\Sigma}_\pi\Sigma_\pi^{-1/2}\right) \in \left[1 - \Delta, 1 + \Delta\right]\textup{ and }\lambda_i\left(\Ff^{-1/2}\widehat{\Ff}\Ff^{-1/2}\right) \in \left[1 - \Delta, 1 + \Delta\right]
\end{equation}
for all $i \in [d]$.
First, a result comparing effects of different preconditioners. For $M \in \mathrm{PD}_{d\times d}$, let $m_{M} = \inf_{x\in\Reals^d}\lambda_d(M^{-1/2} \hess U(x) M^{-1/2})$ and $L_{M} = \sup_{x\in\Reals^d}\lambda_1(M^{-1/2} \hess U(x) M^{-1/2})$
\begin{prop}\label{prop:preconditioner_comparisons}
    Let $M_1,M_2 \in \mathrm{PD}_{d\times d}$. Then
\[
           \lambda_d(M_1^{-1/2}M_2M_1^{-1/2})m_{M_2} \leq m_{M_1} \leq L_{M_1} &\leq\lambda_1(M_1^{-1/2}M_2M_1^{-1/2})L_{M_2}
\]
    
\end{prop}
\begin{proof}
    By the definition of the smoothness and strong convexity constants we have that
    \begin{equation}
        \begin{split}
            L_{M_{1}} & :=\sup_{x\in\mathbb{R}^{d}}\lambda_{1}\left(M_{1}^{-1/2}\nabla^{2}U\left(x\right)M_{1}^{-1/2}\right)\\
 & =\sup_{x\in\mathbb{R}^{d}}\lambda_{1}\left(M_{1}^{-1/2}M_{2}^{1/2}M_{2}^{-1/2}\nabla^{2}U\left(x\right)M_{2}^{-1/2}M_{2}^{1/2}M_{1}^{-1/2}\right)\\
 & \leq\lambda_{1}\left(M_{1}^{-1/2}M_{2}M_{1}^{-1/2}\right)\sup_{x\in\mathbb{R}^{d}}\lambda_{1}\left(M_{2}^{-1/2}\nabla^{2}U\left(x\right)M_{2}^{-1/2}\right)\\
 & =\lambda_{1}\left(M_{1}^{-1/2}M_{2}M_{1}^{-1/2}\right)L_{M_{2}}
        \end{split}
    \end{equation}
\begin{equation}
    \begin{split}
        m_{M_{1}} & :=\inf_{x\in\mathbb{R}^{d}}\lambda_{d}\left(M_{1}^{-1/2}\nabla^{2}U\left(x\right)M_{1}^{-1/2}\right)\\
 & =\inf_{x\in\mathbb{R}^{d}}\lambda_{d}\left(M_{1}^{-1/2}M_{2}^{1/2}M_{2}^{-1/2}\nabla^{2}U\left(x\right)M_{2}^{-1/2}M_{2}^{1/2}M_{1}^{-1/2}\right)\\
 & \geq\lambda_{d}\left(M_{1}^{-1/2}M_{2}M_{1}^{-1/2}\right)\inf_{x\in\mathbb{R}^{d}}\lambda_{d}\left(M_{2}^{-1/2}\nabla^{2}U\left(x\right)M_{2}^{-1/2}\right)\\
 & =\lambda_{d}\left(M_{1}^{-1/2}M_{2}M_{1}^{-1/2}\right)m_{M_{2}}
    \end{split}
\end{equation}
where, in both cases, we used Ostrowski's theorem \citep[Theorem 1]{ostrowski1959}. One could also get a corresponding lower (resp. upper) bound on $L_{M_1}$ (resp. $m_{M_1}$) if one uses the other sides of the inequality in the statement of Ostrowski's theorem.
\end{proof}
Define $\kappa(A) := \lambda_1(A)/\lambda_d(A)$ to be the condition number of a matrix $A\in\mathrm{PD}_{d\times d}$
\begin{cor}\label{cor:condition_number_comparison}
    Let $M_1,M_2 \in \mathrm{PD}_{d\times d}$ be preconditioners. Then \smash{$\kappa_{M_1} \leq \kappa(M_1^{-\frac{1}{2}}M_2M_1^{-\frac{1}{2}})\kappa_{M_2}$}.
\end{cor}
Applying these results with $M_1 = \widehat{\Sigma}_\pi^{-1}$ and $M_2 = \Sigma_\pi^{-1}$ (similarly for $\widehat{\Ff}$ and $\Ff$) gives:
\begin{cor}\label{cor:conditioning_comparison_covariance_case}
    Assume $\widehat{\Sigma}_\pi$ and $\Sigma_\pi$ satisfy \cref{eq:good_preconditioners} for $\Delta \in [0,1)$. Then
\[
    (1+\Delta)^{-1}m_{\Sigma_\pi^{-1}} \leq m_{\widehat{\Sigma}_\pi^{-1}} &\leq L_{\widehat{\Sigma}_\pi^{-1}} \leq (1-\Delta)^{-1}L_{\Sigma_\pi^{-1}} && 
    \textup{and } & \kappa_{\widehat{\Sigma}_\pi^{-1}}&\leq \frac{1 + \Delta}{1 - \Delta}\kappa_{\Sigma_\pi^{-1}}.
\]
\end{cor}
\begin{cor}\label{cor:conditioning_comparison_Fisher_case}
    Assume $\widehat{\Ff}$ and $\Ff$ satisfy \cref{eq:good_preconditioners} for $\Delta \in [0,1)$. Then
    \[
    (1-\Delta)m_{\Ff} \leq m_{\widehat{\Ff}} &\leq L_{\widehat{\Ff}} \leq (1+\Delta)L_{\Ff} && 
    \textup{and } & \kappa_{\Ff}&\leq \frac{1 + \Delta}{1 - \Delta}\kappa_{\Ff}.
    \]
\end{cor}
These corollaries also hold if we reverse the roles of $\Sigma_\pi$ and $\widehat{\Sigma}_\pi$ (similarly for $\widehat{\Ff}$ and $\Ff$). Given Proposition \ref{prop:preconditioned_W_2_contraction} we would like to compare $\kappa(\widehat{\Sigma}_\pi^{-1})$ and $\|\widehat{\Sigma}_\pi^{1/2}\|$ with $\kappa(\Sigma_\pi^{-1})$ and $\|\Sigma_\pi^{1/2}\|$ (likewise for $\Ff$ and $\widehat{\Ff}$).

\begin{prop}\label{cor:matrix_conditioning_comparison}
    Assume $\widehat{\Sigma}_\pi$ and $\Sigma_\pi$ satisfy \cref{eq:good_preconditioners} for $\Delta \in [0,1)$. Then
    \[
    (1 + \Delta)^{-1}\lambda_d(\Sigma_\pi^{-1}) \leq \lambda_d(\widehat{\Sigma}_\pi^{-1}) \leq \lambda_1(\widehat{\Sigma}_\pi^{-1}) \leq (1 - \Delta)^{-1}\lambda_1(\Sigma_\pi^{-1})\textup{ and } \kappa(\widehat{\Sigma}_\pi^{-1}) \leq \frac{1 + \Delta}{1 - \Delta}\kappa(\Sigma_\pi^{-1})
    \]
    and
    \[
    \|\widehat{\Sigma}_\pi^{1/2}\| \leq \frac{1}{\sqrt{1 - \Delta}}\|\Sigma_\pi^{1/2}\|.
    \]
    Assume $\widehat{\Ff}$ and $\Ff$ satisfy \cref{eq:good_preconditioners} for $\Delta \in [0,1)$. Then
    \[
    (1 - \Delta)\lambda_d(\Ff) \leq \lambda_d(\widehat{\Ff}) \leq \lambda_1(\widehat{\Ff}) \leq (1 + \Delta)\lambda_1(\Ff)\textup{ and } \kappa(\widehat{\Ff}) \leq \frac{1 + \Delta}{1 - \Delta}\kappa(\Ff)
    \]
    and
    \[
    \|\widehat{\Ff}^{-1/2}\| \leq \sqrt{1 - \Delta}\|\Ff^{-1/2}\|.
    \]
\end{prop}

The corollary also holds if we reverse the roles of $\Sigma_\pi^{-1}$ and $\widehat{\Sigma}_\pi^{-1}$ (and substitute $\Ff$ and $\widehat{\Ff}$).

\begin{proof}
    We have that
    \[
    \lambda_1(\widehat{\Sigma}_\pi^{-1}) &= \lambda_1(\Sigma_\pi^{-1/2}\Sigma_\pi^{1/2}\widehat{\Sigma}_\pi^{-1}\Sigma_\pi^{1/2}\Sigma_\pi^{-1/2})\\
    &\leq \lambda_1(\Sigma_\pi^{1/2}\widehat{\Sigma}_\pi^{-1}\Sigma_\pi^{1/2})\lambda_1(\Sigma_\pi^{-1})\\
    &= \frac{\lambda_1(\Sigma_\pi^{-1})}{\lambda_1(\Sigma_\pi^{-1/2}\widehat{\Sigma}_\pi\Sigma_\pi^{-1/2})}\\
    &\leq \frac{\lambda_1(\Sigma_\pi^{-1})}{1 - \Delta}
    \]
    where in the second line we use Ostrowski's Theorem. An analogous process suffices to prove the lower bound on $\lambda_d(\widehat{\Sigma}_\pi^{-1})$. The result concerning the matrix condition numbers follows straightforwardly. The exact same arguments work if we exchange $\widehat{\Sigma}_\pi^{-1}$ with $\Sigma_\pi^{-1}$.
    
    For the matrix norm:
    \[
    \|\widehat{\Sigma}_\pi^{1/2}\| &= \lambda_1(\widehat{\Sigma}_\pi^{1/2})\\
    &= \lambda_1(\widehat{\Sigma}_\pi^{-1})^{-1/2}\\
    &\leq \frac{1}{\sqrt{1 - \Delta}}\lambda_1(\Sigma_\pi^{-1})^{-1/2}\\
    &= \frac{1}{\sqrt{1 - \Delta}}\|\Sigma_\pi^{1/2}\|.\\
    \]
    The above proofs work analogously for $\Ff$ and $\widehat{\Ff}$.
\end{proof}

We would like to know how efficient the preconditioned kernels $K_{\widehat{\Sigma}_\pi^{-1},\pi}$ and $K_{\widehat{\Ff},\pi}$ are. Applying \ref{prop:preconditioned_W_2_contraction} and \ref{cor:matrix_conditioning_comparison} gives the following:

\begin{cor}\label{cor:approximate_preconditioned_kernel_efficiency}
    Assume $\widehat{\Sigma}_\pi$ and $\Sigma_\pi$ satisfy \cref{eq:good_preconditioners} for $\Delta \in [0,1)$. Let $K_{(\widehat{\Sigma}_\pi^{-1/2})_\sharp\pi}$ be a $(\Gamma, \gamma, b)$-$W_2$ contraction to $(\widehat{\Sigma}_\pi^{-1/2})_\sharp\pi$. Then $K_{\widehat{\Sigma}_\pi^{-1},\pi}$ is a $(\tilde{\Gamma}, \tilde{\gamma}, b)$-$W_2$ contraction to $\pi$ with
    \[
    \tilde{\Gamma} = \sqrt{\frac{1 + \Delta}{1 - \Delta}\kappa(\Sigma_\pi)}\Gamma\textup{ , }\tilde{\gamma} = \gamma\textup{ and }\tilde{b} = \frac{1}{\sqrt{1 - \Delta}}\|\Sigma_\pi^{1/2}\|b.
    \]
    Assume $\widehat{\Ff}$ and $\Ff$ satisfy \cref{eq:good_preconditioners} for $\Delta \in [0,1)$. Let $K_{(\widehat{\Ff}^{1/2})_\sharp\pi}$ be a $(\Gamma, \gamma, b)$-$W_2$ contraction to $(\widehat{\Ff}^{1/2})_\sharp\pi$. Then $K_{\widehat{\Ff},\pi}$ is a $(\tilde{\Gamma}, \tilde{\gamma}, b)$-$W_2$ contraction to $\pi$ with
    \[
    \tilde{\Gamma} = \sqrt{\frac{1 + \Delta}{1 - \Delta}\kappa(\Ff)}\Gamma\textup{ , }\tilde{\gamma} = \gamma\textup{ and }\tilde{b} = \sqrt{1 - \Delta}\|\Ff^{-1/2}\|b.
    \]
\end{cor}

\section{Corollary to Theorems \ref{thm:Wasserstein_contraction_complexity} and \ref{thm:preconditioner_learn_complexity} for the Complexity of Learning the Covariance and Fisher Preconditioners}\label{section:preconditioner_learning_corollary}

\begin{cor}\label{cor:preconditioner_learn_complexity}
    Suppose that the underlying Markov kernel within the thinned Markov chain sampler in Algorithm \ref{alg:thinned_Markov_chain} satisfies the Wasserstein contraction condition in Definition \ref{def:W_2_contraction}.
    \begin{enumerate}[noitemsep]
        \item \textup{(Covariance Preconditioner Learning Complexity)} Let $U$ be $m$-strongly convex, and let $\beta_{\Sigma_\pi}>0$ with $\beta_{\Sigma_\pi}\mathbf{I}_d \preceq \Sigma_\pi$. Assume that $\delta,\Delta>0$ satisfy
        \begin{equation}
            360\Gamma^2b\sqrt{d}\beta_{\Sigma_\pi}^{-1/2} < \delta\Delta \leq 120\Gamma\sqrt{\mathrm{tr}\left(\Sigma_\pi\right)d}\beta_{\Sigma_\pi}^{-1/2}.
        \end{equation}
        Then the complexity, in iterations of the underlying Markov chain, to output $\{X_t\}_{t=1}^N$ satisfying \cref{eqn:preconditioner_estimate_goodness_condition_covariance} is:
        \[
        \tilde{O}\left(\frac{d}{\Delta \gamma} \rbra{\frac{1}{\delta}\vee K_\mathrm{cov}^3}\log\left(\rbra{\sqrt{2}\sqrt{\beta_{\Sigma_\pi}}\delta\Delta - 360\Gamma^2b\sqrt{d}}^{-1}\right)\right),
        && \text{for }
        K_\mathrm{cov} :=  m_{\Sigma_\pi^{-1}}^{-1/2}.
        \]
        \item \textup{(Fisher Preconditioner Learning Complexity)}  Let $U$ be $L$-smooth, and let $\alpha_\Ff>0$ with $\alpha_\Ff \mathbf{I}_d \preceq \Ff$. Assume that $\delta,\Delta>0$ satisfy
        \begin{equation}
            8\Gamma^2bL\sqrt{d}\alpha_\Ff^{-1/2} < \delta\Delta \leq \frac{8}{3}\Gamma L\sqrt{3\mathrm{tr}\left(\Sigma_\pi\right)d}\alpha_\Ff^{-1/2}.
        \end{equation}
        Then the complexity, in iterations of the underlying Markov chain, to output $\{X_t\}_{t=1}^N$ satisfying \cref{eqn:preconditioner_estimate_goodness_condition_Fisher} is:
        \[
        \tilde{O}\rbra[3]{\frac{dK_\mathrm{Fisher}^3}{\gamma\Delta}\log\left(\rbra{\sqrt{\alpha_\Ff}\delta\Delta - 8\Gamma^2bL\sqrt{d}}^{-1}\right)}, && \text{where } K_\mathrm{Fisher} := \sqrt{L_\Ff}.
        \]
    \end{enumerate}
\end{cor}

The bounds on $\delta\Delta$ force $\varepsilon$ satisfy the assumptions of Theorem \ref{thm:Wasserstein_contraction_complexity}, and so we are guaranteed that $\{X_t\}_{t=1}^N$ is $\sqrt{N}\varepsilon$-AIID with the appropriate values of $N$ and $\varepsilon$ if we take the correct number of iterations. How precisely we can guarantee conditions \cref{eqn:preconditioner_estimate_goodness_condition_covariance} and \cref{eqn:preconditioner_estimate_goodness_condition_Fisher} depends on the bias $b$ via the lower bounds on $\delta\Delta$ in the above. This is an impediment on achieving a good preconditioner with a biased algorithm, although the $\Delta$ parameter needn't necessarily be arbitrarily small for the purposes of learning a good preconditioner. Note that coefficient pre-multiplying $b$ is a factor of $L$ larger in the Fisher case than in the covariance case. We believe this is an artifact of our proof technique and not something intrinsic to the learning tasks. See the paragraph below \cref{eq:Fisher_bound_decomposition} for a discussion.

\section{Proofs}

\subsection{Proof of Proposition \ref{prop:preconditioned_W_2_contraction}}\label{proof:preconditioned_W_2_contraction}

We first prove the following Lemma:

\begin{lem}\label{lem:preconditioned_kernel_pushforward_definition}
    Let $K_\pi$ and $K_{M,\pi}$ be defined as in Definition \ref{def:preconditioned_kernel}. Then for all $\mu\in \Pp(\RR^d)$ and $k \in \NN \cup \{0\}$
    \[
    \mu K_{M,\pi}^k = M^{-1/2}_\sharp\left(\left(M^{1/2}_\sharp\mu\right)K^k_{M^{1/2}_\sharp\pi}\right)
    \]
    i.e. to sample from $\mu K_{M,\pi}^k$ sample from $\mu$, transform using $M^{1/2}$, take $k$ steps of $K$ targetting $M^{1/2}_\sharp\pi$ and then transform back using $M^{-1/2}$.
\end{lem}

\begin{proof}
    Let $A \in \Bb(\RR^d)$. Then
    \[
    \mu K_{M,\pi}^k(A) &= \int_{\RR^d}\mu(\dee x)K_{M,\pi}^k(x\to A)\\
    &= \int_{\RR^d}\mu(\dee x)K_{M^{1/2}_\sharp\pi}^k(M^{1/2}x\to M^{1/2}A)\\
    &= \int_{\RR^d}M^{1/2}_\sharp\mu(\dee z)K_{M^{1/2}_\sharp\pi}^k(z\to M^{1/2}A)\\
    &= \left(\left(M^{1/2}_\sharp\mu\right)K^k_{M^{1/2}_\sharp\pi}\right)(M^{1/2}A)\\
    &= M^{-1/2}_\sharp\left(\left(M^{1/2}_\sharp\mu\right)K^k_{M^{1/2}_\sharp\pi}\right)(A)\\
    \]
    where in the third line we use the change of variables $z = M^{1/2}x$.
\end{proof}

Given Lemma \ref{lem:preconditioned_kernel_pushforward_definition} we can use the stability of $W_2$ under Lipschitz transformations:

\[
W_2(\pi, \mu K_{M,\pi}^k) &= W_2\left(M^{-1/2}_\sharp M^{1/2}_\sharp\pi, M^{-1/2}_\sharp\left(\left(M^{1/2}_\sharp\mu\right)K^k_{M^{1/2}_\sharp\pi}\right)\right)\\
&\leq \|M^{-1/2}\|W_2\left( M^{1/2}_\sharp\pi, \left(\left(M^{1/2}_\sharp\mu\right)K^k_{M^{1/2}_\sharp\pi}\right)\right)\\
&\leq \|M^{-1/2}\|\left(\Gamma(M^{1/2}_\sharp\pi) \exp(-\gamma(\Gamma(M^{1/2}_\sharp\pi)k)W_2(M^{1/2}_\sharp\pi, M^{1/2}_\sharp\mu) + b(M^{1/2}_\sharp\pi)\right)\\
&\leq \|M^{-1/2}\|\left(\|M^{1/2}\|\Gamma(M^{1/2}_\sharp\pi) \exp(-\gamma(\Gamma(M^{1/2}_\sharp\pi)k)W_2(\pi, \mu) + b(M^{1/2}_\sharp\pi)\right)\\
\]
where in the first line we use Lemma \ref{lem:preconditioned_kernel_pushforward_definition}, in the second and fourth line we use the stability of $W_2$ under Lipschitz pushforwards and in the third line we use the fact that $K_{M^{1/2}_\sharp \pi}$ is a $W_2$ contraction to $M^{1/2}_\sharp \pi$. 

\subsection{Proof of Theorem \ref{thm:Wasserstein_contraction_complexity}}\label{proof:Wasserstein_contraction_complexity}

First, we prove Proposition \ref{prop:Wasserstein_decomposition} by constructing a coupling using a transport map between $\pi^{\otimes N}$
and $\mu_{N}$. Since the Wasserstein-2 distance is the infimum over
couplings it can be bounded using this map. Let $y_{a:b}=(y_a,y_{a+1},\dots, y_b)$.
Let $\left\{ Y_{t}\right\} _{t=1}^{N}\sim\pi^{\otimes N}$ and let
$T_{1}$ be the Monge map between $\pi$ and
$\mu_{0}K^{t_\mathrm{burn}}$ such that $T_{1}\left(Y_{1}\right)\sim\mu_{0}K^{t_\mathrm{burn}}$ with optimal quadratic cost. By Brenier's
Theorem \citep{brenier1991}, since $\pi$ a Lebesgue density, $\left(\mathrm{Id}\times T_{1}\right)_{\sharp}\pi$
is the $W_{2}$-optimal coupling between $\pi$ and $\mu_{0}K^{t_\mathrm{burn}}$. For
$i\in\left\{ 2,\ldots,N\right\} $ recursively define $T_{i}\left(y_{1:(i-1)},\bullet\right):\mathbb{R}^{d}\to\mathbb{R}^{d}$
for $\pi^{\otimes i -  1}$-almost all $y_{1:(i-1)}\in(\mathbb{R}^{d})^{i-1}$
as the optimal quadratic cost Monge map between $\pi$ and $K^{t_\mathrm{thin}}\left(T_{i-1}\left(y_{1:(i-1)}\right)\to\bullet\right)$
such that $T_{i}\left(y_{1:(i-1)},Y_{i}\right)\sim K^{t_\mathrm{thin}}\left(T_{i-1}\left(y_{1:(i-1)}\right)\to\bullet\right)$.
Applying Brenier's Theorem again, we have $\left(\mathrm{Id}\times T_{i}\left(y_{1:(i-1)},\bullet\right)\right)_{\sharp}\pi$
is the $W_{2}$-optimal coupling of $\pi$ and $K^{t_\mathrm{thin}}\left(T_{i-1}\left(y_{1:(i-1)}\right)\to\bullet\right)$,
for almost all $y_{1:(i-1)}\in(\mathbb{R}^{d})^{i-1}$.

Now, $T_{i}\left(Y_{1:i}\right)\sim\mu_{0}K^{t_\mathrm{burn} + (i - 1)t_\mathrm{thin}}$
for all $i\in\left\{ 1,\ldots,N\right\} $ and we can build the transport
\[
T\left(y_{1:N}\right):=\left(T_{1}\left(y_{1}\right),T_{2}\left(y_{1:2}\right),\ldots,T_{N}\left(y_{1:N}\right)\right)\textup{ for all }\text{\ensuremath{y_{1:N}}\ensuremath{\ensuremath{\in}}\ensuremath{(\mathbb{R}^{d})^{N}}}
\]
such that $T_{\sharp}\pi^{\otimes N}=\mu_{N}$ (If $d=1$ this would
be the Knothe-Rosenblatt rearrangement \citep{carlier2010}). We upper bound the $W_{2}$-distance between $\pi^{\otimes N}$
and $\mu_{N}$ as follows:
\[
W_{2}\left(\pi^{\otimes N},\mu_{N}\right)^{2}
& :=\inf_{\gamma\in\mathcal{C}\left(\pi^{\otimes N},\mu_{N}\right)}\int\left\Vert y-x\right\Vert ^{2}\gamma\left(\mathrm{d}y,\mathrm{d}x\right)\\
 & \leq\int\left\Vert y-T\left(y\right)\right\Vert ^{2}\pi^{\otimes N}\left(\mathrm{d}y\right)\\
 & =\int\left(\left\Vert y_{1}-T_{1}\left(y_{1}\right)\right\Vert ^{2}+\sum_{i=1}^{N}\left\Vert y_{i}-T_{i}\left(y_{1:i}\right)\right\Vert ^{2}\right)\pi^{\otimes N}\left(\mathrm{d}y\right)\\
 & =\int\norm{y_{1}-T_{1}(y_{1})}^{2} \pi(\mathrm{d}y_{1})
 +\sum_{i=1}^{N}\int\left\Vert y_{i}-T_{i}\left(y_{1:i}\right)\right\Vert ^{2}\pi^{\otimes N}\left(\mathrm{d}y\right)
\]
where $\mathcal{C}\left(\pi^{\otimes N},\mu_{N}\right)$ is the
space of couplings between $\pi^{\otimes N}$ and $\mu_{N}$. The
first term on the right hand side is $W_{2}\left(\pi,\mu_{0}K^{t_\mathrm{burn}}\right)^{2}$.
The summands are
\[
&\hspace{-2em}
\int\left\Vert y_{i}-T_{i}\left(y_{1},\ldots,y_{i}\right)\right\Vert ^{2}\pi\left(\mathrm{d}y_{1}\right)\ldots\pi\left(\mathrm{d}y_{i}\right) \\
    & =\int\left(\int\left\Vert y_{i}-T_{i}\left(y_{1},\ldots,y_{i}\right)\right\Vert ^{2}\pi\left(\mathrm{d}y_{i}\right)\right)\pi\left(\mathrm{d}y_{1}\right)\ldots\pi\left(\mathrm{d}y_{i-1}\right)\\
 & =\int W_{2}\left(\pi,K^{t_\mathrm{thin}}\left(T_{i-1}\left(y_{1},\ldots,y_{i-1}\right)\to\cdot\right)\right)^{2}\pi\left(\mathrm{d}y_{1}\right)\ldots\pi\left(\mathrm{d}y_{i-1}\right)\\
 & =\mathbb{E}_{\pi^{\otimes i}}\left[W_{2}\left(\pi,K^{t_\mathrm{thin}}\left(T_{i-1}\left(Y_{1},\ldots,Y_{i-1}\right)\to\cdot\right)\right)^{2}\right]\\
 & =\mathbb{E}\left[W_{2}\left(\pi,K^{t_\mathrm{thin}}\left(X_{i-1}\to\cdot\right)\right)^{2}\right]
\]
where the final line is due to the fact that $T_{i-1}\left(Y_{0},\ldots,Y_{i-1}\right)\stackrel{d}{=}X_{i-1}$
for all $i\in\left\{ 1,\ldots,N\right\} $.

Theorem \ref{thm:Wasserstein_contraction_complexity} is a rephrasing of the following Lemma, which we prove below.

\begin{lem}
\label{lem:thinned_Markov_chain_approximately_iid_W_2}Let $\left\{ X_{t}\right\} _{t=1}^{N}$
be sampled with a thinned Markov chain sampler \ref{alg:thinned_Markov_chain}
whose undelying Markov kernel is a $\left(\Gamma, \gamma,b\right)$-$W_{2}$ contraction to $\pi$. Let $3\Gamma^{2}b<\varepsilon$, $\Gamma\geq1$,
and $\varepsilon^{2}\leq\Gamma^{2}\left(3\mathrm{tr}\left(\Sigma_{\pi}\right)+4b^{2}\right)+2b^{2}$
(assuming that $\pi\in\mathcal{P}\left(\mathbb{R}^{d}\right)$ has
an elementwise-finite covariance $\Sigma_{\pi}\in\mathbb{R}^{d\times d}$).
Setting
\[
k_{\textup{burn}}\geq\frac{1}{\gamma}\log\frac{3\Gamma^{3}W_{2}\left(\pi,\mu_0\right)}{\varepsilon-3\Gamma^{2}b}\textup{ and }k_{\textup{thin}}\geq\frac{1}{2\gamma}\log\frac{2\Gamma^{2}\left(3\mathrm{tr}\left(\Sigma_{\pi}\right)+4b^{2}\right)}{\varepsilon^{2}-2b^{2}}
\]
ensures that $\left\{ X_{t}\right\} _{t=1}^{N}$ is $\sqrt{N}\varepsilon$-approximately
IID from $\pi$ in $W_{2}$.
\end{lem}
Let $\mu_{N}\in\mathcal{P}\left(\mathbb{R}^{d\times N}\right)$ be
the distribution of the collection of samples from the thinned Markov
chain sampler i.e.
\[
\mu_{N}\left(\mathrm{d}x\right):=\mu_{0}K^{k_{\textup{burn}}}\left(dx_{1}\right)K^{k_{\textup{thin}}}\left(x_{1}\to\mathrm{d}x_{2}\right)\cdots K^{k_{\textup{thin}}}\left(x_{N-1}\to\mathrm{d}x_{N}\right).
\]
Then using Lemma \ref{prop:Wasserstein_decomposition} we have that
\[
W_{2}\left(\pi^{\otimes N},\mu_{N}\right)^{2}\leq\mathbb{E}\left[W_{2}\left(\pi,\mu_{0}K^{k_{\textup{burn}}}\right)^{2}+\sum_{t=1}^{N-1}W_{2}\left(\pi,K^{k_{\textup{thin}}}\left(X_{t}\to\cdot\right)\right)^{2}\right]
\]
where the expectation is with respect to the randomness in the chain.

First we bound $\mathbb{E}\left[W_{2}\left(\pi,K^{k_{\textup{thin}}}\left(X_{t}\to\cdot\right)\right)^{2}\right]$.
Note that with probability one
\begin{equation}
    \begin{split}
        W_{2}\left(\pi,K^{k_{\textup{thin}}}\left(X_{t}\to\cdot\right)\right)^{2} & =W_{2}\left(\pi,\delta_{X_{t}}K^{k_{\textup{thin}}}\right)^{2}\\
 & \leq\left(\Gamma\exp\left(-\gamma k_{\textup{thin}}\right)W_{2}\left(\pi,\delta_{X_{t}}\right)+b\right)^{2}\\
 & \leq2\Gamma^{2}\exp\left(-2\gamma k_{\textup{thin}}\right)W_{2}\left(\pi,\delta_{X_{t}}\right)^{2}+2b^{2}
    \end{split}
\end{equation}
where in the second line we use the Wasserstein-2 contractivity of
$K$ and in the final line we use the fact that $\left(a+b\right)^{2}\leq2a^{2}+2b^{2}$.
Now note that
\begin{equation}
    \begin{split}
        W_{2}\left(\pi,\delta_{X_{t}}\right)^{2} & =\inf_{\gamma\in\mathcal{C}\left(\pi,\delta_{X_{i}}\right)}\int\left\Vert y_t-x_t\right\Vert ^{2}\gamma\left(\mathrm{d}x_t,\mathrm{d}y_t\right)\\
 & =\int\left\Vert y_t-T\left(y\right)\right\Vert ^{2}\pi\left(\mathrm{d}y_t\right)\\
 & =\int\left\Vert y_t-X_{t}\right\Vert ^{2}\pi\left(\mathrm{d}y_t\right)\\
 & =\mathbb{E}_{Y\sim\pi}\left[\left\Vert Y-X_{t}\right\Vert ^{2}\right]\\
 & =\mathrm{tr}\left(\Sigma_{\pi}\right)+\left\Vert X_{t}-\mu_{\pi}\right\Vert ^{2}
    \end{split}
\end{equation}
where $\Sigma_{\pi}:=\mathrm{Cov}_{\pi}\left(Y\right)$, $\mu_{\pi}:=\mathbb{E}_{\pi}\left[Y\right]$.
We will need to take the expectation of the last line with respect
to the coupling in the proof of the Wasserstein-2 decomposition result,
and so we need to bound $B_{t}:=\mathbb{E}\left[\left\Vert X_{t}-\mu_{\pi}\right\Vert ^{2}\right]$,
and we will do so in terms of $W_{2}\left(\pi,\mu_{0}K^{i}\right)$
where $i=k_{\textup{burn}}+\left(t-1\right)k_{\textup{thin}}$ since
then $\mu_{0}K^{i}$ is the distribution of $X_{t}$:
\begin{equation}
    \begin{split}
        W_{2}\left(\pi,\mu_{0}K^{i}\right)^{2} & :=\inf_{\gamma\in\mathcal{C}\left(\pi,\mu_{0}K^{i}\right)}\int\left\Vert y_i-x_{i}\right\Vert ^{2}\gamma\left(\mathrm{d}y_i,\mathrm{d}x_{i}\right)\\
 & =\int\left\Vert y_i-x_{i}\right\Vert ^{2}\gamma^{*}\left(\mathrm{d}y_i,\mathrm{d}x_{i}\right)\\
 & =\mathrm{tr}\left(\Sigma_{\pi}\right)-2\int\left\langle y_y-\mu_{\pi},x-\mu_{\pi}\right\rangle \gamma^{*}\left(\mathrm{d}y_i,\mathrm{d}x_{i}\right)+B_{i}.
    \end{split}
\end{equation}
Therefore 
\begin{equation}
    \begin{split}
        B_{i} & =W_{2}\left(\pi,\mu_{0}K^{i}\right)^{2}-\mathrm{tr}\left(\Sigma_{\pi}\right)+2\int\left\langle y_i-\mu_{\pi},x_i-\mu_{\pi}\right\rangle \gamma^{*}\left(\mathrm{d}y_i,\mathrm{d}x_{i}\right)\\
 & \leq W_{2}\left(\pi,\mu_{0}K^{i}\right)^{2}-\mathrm{tr}\left(\Sigma_{\pi}\right)+2\sqrt{\mathbb{E}_{\gamma^{*}}\left[\left\Vert Y_i-\mu_{\pi}\right\Vert ^{2}\right]}\sqrt{\mathbb{E}_{\gamma^{*}}\left[\left\Vert X_{i}-\mu_{\pi}\right\Vert ^{2}\right]}\\
 & =W_{2}\left(\pi,\mu_{0}K^{i}\right)^{2}-\mathrm{tr}\left(\Sigma_{\pi}\right)+2\sqrt{\mathrm{tr}\left(\Sigma_{\pi}\right)}\sqrt{B_{i}}
    \end{split}
\end{equation}
where in the second line we use Cauchy-Schwarz, and so
\begin{equation}
    \begin{split}
         \left(\sqrt{B_{i}}-\sqrt{\mathrm{tr}\left(\Sigma_{\pi}\right)}\right)^{2}=  B_{i}-2\sqrt{\mathrm{tr}\left(\Sigma_{\pi}\right)}\sqrt{B_{i}}+\mathrm{tr}\left(\Sigma_{\pi}\right) & \leq W_{2}\left(\pi,\mu_{0}K^{i}\right)^{2}\\
    \end{split}
\end{equation}
If $\sqrt{B_{i}}\leq\sqrt{\mathrm{tr}\left(\Sigma_{\pi}\right)}$
then we have an upper bound on $B_{i}$. If $\sqrt{B_{i}}\geq\sqrt{\mathrm{tr}\left(\Sigma_{\pi}\right)}$.
Then
\[
\left(\sqrt{B_{i}}-\sqrt{\mathrm{tr}\left(\Sigma_{\pi}\right)}\right)^{2} & \leq W_{2}\left(\pi,\mu_{0}K^{i}\right)^{2} &&
\Rightarrow&\sqrt{B_{i}} & \leq W_{2}\left(\pi,\mu_{0}K^{i}\right)+\sqrt{\mathrm{tr}\left(\Sigma_{\pi}\right)} \\
&&&
\Rightarrow & B_{i} & \leq2W_{2}\left(\pi,\mu_{0}K^{i}\right)^{2}+2\mathrm{tr}\left(\Sigma_{\pi}\right).
\]
Taking expectations and substituting this in gives
\[
& \hspace{-1em}  \mathbb{E}\left[W_{2}\left(\pi,K^{k_{\textup{thin}}}\left(X_{t}\to\cdot\right)\right)^{2}\right] \\ & \leq2\Gamma^{2}\exp\left(-2\gamma k_{\textup{thin}}\right)\left(2W_{2}\left(\pi,\mu_{0}K^{k_{\textup{burn}}+\left(t-1\right)k_{\textup{thin}}}\right)^{2}+3\mathrm{tr}\left(\Sigma_{\pi}\right)\right)+2b^{2}.
\]
Applying the contraction assumption to $W_{2}\left(\pi,\mu_{0}K^{k_{\textup{burn}}+\left(t-1\right)k_{\textup{thin}}}\right)$
and $\left(a+b\right)^{2}\leq2a^{2}+2b^{2}$:
\begin{equation}
    \begin{split}
        &\hspace{-2em}\mathbb{E}\left[W_{2}\left(\pi,K^{k_{\textup{thin}}}\left(X_{t}\to\cdot\right)\right)^{2}\right] \\
& \leq2\Gamma^{2}\exp\left(-2\gamma k_{\textup{thin}}\right)(4\Gamma^{2}\exp\left(-2\gamma\left(t-1\right)k_{\textup{thin}}\right)W_{2}(\pi,\mu_{0}K^{k_{\textup{burn}}})^{2}\\
 & \qquad+4b^{2}+3\mathrm{tr}\left(\Sigma_{\pi}\right))+2b^{2}
    \end{split}
\end{equation}
Summing over $t\in\left[N-1\right]$ and adding $W_{2}\left(\pi,\mu_{0}K^{k_{\textup{burn}}}\right)^{2}$
gives:
\begin{equation}
    \begin{split}
        &\hspace{-1em}W_{2}\left(\pi^{\otimes N},\mu_{N}\right)^{2} \\
& \leq W_{2}\left(\pi,\mu_{0}K^{k_{\textup{burn}}}\right)^{2}+8\Gamma^{4}\exp\left(-2\gamma k_{\textup{thin}}\right)W_{2}\left(\pi,\mu_{0}K^{k_{\textup{burn}}}\right)^{2}\sum_{t=1}^{N-1}\exp\left(-2\gamma\left(t-1\right)k_{\textup{thin}}\right)\\
 & \qquad+\left(2\Gamma^{2}\exp\left(-2\gamma k_{\textup{thin}}\right)\left(4b^{2}+3\mathrm{tr}\left(\Sigma_{\pi}\right)\right)+2b^{2}\right)\left(N-1\right)\\
 & \leq\left(1+8\Gamma^{4}\frac{\exp\left(-2\gamma k_{\textup{thin}}\right)}{1-\exp\left(-2\gamma k_{\textup{thin}}\right)}\right)W_{2}\left(\pi,\mu_{0}K^{k_{\textup{burn}}}\right)^{2}\\
 & \qquad+\left(2\Gamma^{2}\exp\left(-2\gamma k_{\textup{thin}}\right)\left(4b^{2}+3\mathrm{tr}\left(\Sigma_{\pi}\right)\right)+2b^{2}\right)\left(N-1\right)
    \end{split}
\end{equation}
where, for the second inequality, the finite sum is bounded with an infinite sum.

To ensure $\left\{ X_{t}\right\} _{t=1}^{N}$ is $\sqrt{N}\varepsilon$-AIID from $\pi$ in $W_{2}$ it suffices to choose $k_{\textup{burn}}$
and $k_{\textup{thin}}$ with
\begin{equation}
    \begin{split}
        \left(1+8\Gamma^{4}\frac{\exp\left(-2\gamma k_{\textup{thin}}\right)}{1-\exp\left(-2\gamma k_{\textup{thin}}\right)}\right)W_{2}\left(\pi,\mu_{0}K^{k_{\textup{burn}}}\right)^{2} & \leq\varepsilon^{2}\textup{ and}\\
2\Gamma^{2}\exp\left(-2\gamma k_{\textup{thin}}\right)\left(4b^{2}+3\mathrm{tr}\left(\Sigma_{\pi}\right)\right)+2b^{2} & \leq\varepsilon^{2}.
    \end{split}
\end{equation}
To satisfy the second inequality choose $k_{\textup{thin}}$ such
that
\[
k_{\textup{thin}}\geq\frac{1}{2\gamma}\log\left(\frac{2\Gamma^{2}\left(4b^{2}+3\mathrm{tr}\left(\Sigma_{\pi}\right)\right)}{\varepsilon^{2}-2b^{2}}\right).
\]
Note that, combined with the assumption in the Lemma statement,
this means that $k_{\textup{thin}}\geq\left(1/2\gamma\right)\log2$
and hence $\exp\left(-2\gamma k_{\textup{thin}}\right)\leq1/2$. Therefore,
since $x\left(1-x\right)^{-1}\leq2x$ for all $x\in\left[0,1/2\right]$,
for the first inequality it suffices to take $k_{\textup{burn}}$ such
that 
\[
\left(1+16\Gamma^{4}\exp\left(-2\gamma k_{\textup{thin}}\right)\right)W_{2}\left(\pi,\mu_{0}K^{k_{\textup{burn}}}\right)^{2}\leq\varepsilon^{2}.
\]
Substituting in our bound on $k_{\textup{thin}}$ shows that is suffices
to take 
\[
\left(1+8\Gamma^{4}\frac{\varepsilon^{2}-2b^{2}}{\Gamma^{2}\left(4b^{2}+3\mathrm{tr}\left(\Sigma_{\pi}\right)\right)}\right)W_{2}\left(\pi,\mu_{0}K^{k_{\textup{burn}}}\right)^{2}\leq\varepsilon^{2}.
\]
Now 
\[
\frac{\varepsilon^{2}-2b^{2}}{\Gamma^{2}\left(4b^{2}+3\mathrm{tr}\left(\Sigma_{\pi}\right)\right)}\leq1
\]
by hypothesis and hence it suffices to take
\[
\left(1+8\Gamma^{4}\right)W_{2}\left(\pi,\mu_{0}K^{k_{\textup{burn}}}\right)^{2}\leq\varepsilon^{2}.
\]
Applying the Wasserstein contraction result gives
\[
k_{\textup{burn}}\geq\frac{1}{\gamma}\log\left(\frac{\Gamma\sqrt{1+8\Gamma^{4}}W_{2}\left(\pi,\mu_{0}\right)}{\varepsilon-b\sqrt{1+8\Gamma^{4}}}\right).
\]
By hypothesis we have $\sqrt{1+8\Gamma^{4}}\leq3\Gamma^{2}$ from
which the result follows.

\subsection{Proof of Theorem \ref{thm:preconditioner_learn_complexity} Part 1.}\label{proof:preconditioner_learn_complexity_part_1}
Define $\{Y_t\}_{t=1}^N \sim \pi^{\otimes N}$ as an IID sample from $\pi$, let $\mu_\pi = \EE[Y_1]$, and let $Y \in \RR^{d\times N}$ be the matrix whose $j$th column is $Y_j - \mu_\pi$ for $j \in [N]$. 
Let $\{X_t\}_{t=1}^N$ be the $\sqrt{N}\varepsilon$-AIID output of the thinned Markov chain sampler in Algorithm \ref{alg:thinned_Markov_chain} and let $X\in\RR^{d\times N}$ be the matrix whose $j$th column is $X_j - \mu_\pi$ for $j\in [N]$. Then:
\begin{equation}
    \widehat{\Sigma}_\pi = \frac{1}{N}XX^\top - \left(\bar{X} - \mu_\pi\right)\left(\bar{X} - \mu_\pi\right)^\top\textup{ where }\bar{X}:=\frac{1}{N}\sum_{t=1}^NX_t.
\end{equation}
Then we have
\begin{equation}\label{eq:covariance_bound_decomposition_proof}
    \begin{split}
        \|\Sigma_\pi^{-1/2}\widehat{\Sigma}_\pi\Sigma_\pi^{-1/2}-\mathbf{I}_d\|&\leq \|\frac{1}{N}\Sigma_\pi^{-1/2}XX^\top\Sigma_\pi^{-1/2} - \frac{1}{N}\Sigma_\pi^{-1/2}YY^\top\Sigma_\pi^{-1/2}\| \\&\qquad+ \|\frac{1}{N}\Sigma_\pi^{-1/2}YY^\top\Sigma_\pi^{-1/2} - \mathbf{I}_d\|+ \|\Sigma_\pi^{-1/2}(\bar{X} - \mu_\pi)\|^2
    \end{split}
\end{equation}

We will bound the first and third terms of \cref{eq:covariance_bound_decomposition_proof} in expectation and the middle term with high probability.

\subsubsection{Bounding the first and third terms of {Eq. (\ref{eq:covariance_bound_decomposition_proof})}}

Let $\{\tilde{X}_t\}_{t=1}^N := \{\Sigma_\pi^{-1/2}(X_t - \mu_\pi)\}_{t=1}^N$. By Proposition \ref{prop:rootNIID_Lipschitz_transformation} $\{\tilde{X}_t\}_{t=1}^N$ is $\sqrt{N}\|\Sigma_\pi^{-1/2}\|\varepsilon$-AIID from $\tilde{\pi}:=\Sigma_\pi^{-1/2}\sharp\pi_0$ where $\pi_0$ is the centred version of $\pi$. 
By hypothesis we have that $\beta_{\Sigma_\pi}\mathbf{I}_d \preceq \Sigma_\pi$ and so $\|\Sigma_\pi^{-1/2}\| \leq \beta_{\Sigma_\pi}^{-1/2}$. 
Note that  $\Sigma_\pi^{-1/2}\mu_\pi = \mu_{\tilde \pi}$. Let $\tilde{X}\in\RR^{d\times N}$ be the matrix whose $j$th column is $\tilde{X}_j$. Let $\{\tilde{Y}_t\}_{t=1}^N := \{\Sigma_\pi^{-1/2}(Y_t - \mu_{\pi})\}_{t=1}^N$ and let $\tilde{Y}\in\RR^{d\times N}$ be the matrix whose $j$th column is $\tilde{Y}_j$. Since $\tilde{\pi}$ is centred, $\{\tilde{X}_t\}_{t=1}^N$ is $\sqrt{N}\beta_{\Sigma_\pi}^{-1/2}\varepsilon$-AIID from $\tilde{\pi}$ and $\{\tilde{Y}_t\}_{t=1}^N$ is sampled from $\tilde{\pi}^{\otimes N}$ we may apply Proposition \ref{prop:rootNIID_implies_empirical_covariance_bound}:
\begin{equation}
    \begin{split}
        \mathbb{\EE}\left[\|\frac{1}{N}\Sigma_\pi^{-1/2}XX^\top\Sigma_\pi^{-1/2} - \frac{1}{N}\Sigma_\pi^{-1/2}YY^\top\Sigma_\pi^{-1/2}\|\right] &= \mathbb{\EE}\left[\|\frac{1}{N}\tilde{X}\tilde{X}^\top - \frac{1}{N}\tilde{Y}\tilde{Y}^\top\|\right]\\
        &\leq \mathbb{\EE}\left[\|\frac{1}{N}\tilde{X}\tilde{X}^\top - \frac{1}{N}\tilde{Y}\tilde{Y}^\top\|_F\right]\\
        &\leq 2\sqrt{\mathrm{tr}\left(\Sigma_{\tilde{\pi}}\right)}\beta_{\Sigma_\pi}^{-1/2}\varepsilon+\beta_{\Sigma_\pi}^{-1}\varepsilon^{2}\\
        &= 2\sqrt{d}\beta_{\Sigma_\pi}^{-1/2}\varepsilon+\beta_{\Sigma_\pi}^{-1}\varepsilon^{2}\\
    \end{split}
\end{equation}
where in the second line we bound the 2-norm above by the Frobenius norm and the expectations are with respect to the $W_2$-optimal coupling between $\{\tilde{X}_t\}_{t=1}^N$ and $\{\tilde{Y}_t\}_{t=1}^N$.

Similarly we may bound the third term of \cref{eq:covariance_bound_decomposition_proof} in expectation using Proposition \ref{prop:L2_rootNepsilon_mean_error}:
\begin{equation}
    \begin{split}
        \EE[\|\Sigma_\pi^{-1/2}(\bar{X} - \mu_{\pi})\|^2] &= \EE[\|\frac{1}{N}\sum_{t=1}^N \tilde{X}_t - \mu_{\tilde \pi}\|^2]\\
        &\leq \left(\beta_{\Sigma_\pi}^{-1/2}\varepsilon + \sqrt{\frac{1}{N}\mathrm{tr}\left(\Sigma_{\tilde{\pi}}\right)}\right)^2\\
        &=\left(\beta_{\Sigma_\pi}^{-1/2}\varepsilon + \sqrt{\frac{d}{N}}\right)^2\\
    \end{split}
\end{equation}
where the expectations are under the $W_2$-optimal coupling between $\{\tilde{X}_t\}_{t=1}^N$ and $\{\tilde{Y}_t\}_{t=1}^N$.

\subsubsection{Bounding the second term of {Eq. (\ref{eq:covariance_bound_decomposition_proof})}}

Note that $N^{-1}\Sigma_\pi^{-1/2}YY^\top\Sigma_\pi^{-1/2} = N^{-1}\tilde{Y}\tilde{Y}^\top$ is an empirical covariance matrix (without centre-ing) of IID samples from $\tilde{\pi}$. We apply \citep[Exercise 4.7.3]{vershynin2018}, a concentration result for empirical covariances constructed using mean-$0$ IID samples. This gives:
\begin{equation}
    \|\frac{1}{N}\Sigma_\pi^{-1/2}YY^\top\Sigma_\pi^{-1/2} - \mathbf{I}_d\| \leq CK_\mathrm{cov}^2\left(\sqrt{\frac{d+u}{N}} + \frac{d+u}{N}\right)
\end{equation}
 with probability at least $1 - 2\exp(-u)$, for any $u\geq0$, where $C>0$ is an absolute constant and
$K_\mathrm{cov} > 0$ is such that for $Z\sim\pi$
\begin{equation}\label{eq:covariance_Orlicz_inequality}
    \forall z\in\Reals^d \qquad \frac{\|\langle \Sigma_\pi^{-1/2}(Z - \mu_{\pi}),z\rangle\|_{\psi_2}}{\sqrt{\EE_\pi[\langle \Sigma_\pi^{-1/2}(Z - \mu_{\pi}),z\rangle^2]}}\leq K_\mathrm{cov}.
\end{equation}
Here $\|\cdot\|_{\psi_2}$ is the \emph{sub-Gaussian} norm (a.k.a the Orlicz norm associated with $\psi_2:s\mapsto e^{s^2}-1$):
\begin{equation}
    \left\|X\right\|_{\psi_2}:=\inf\left\{t:\EE\left[\exp\left(\frac{X^2}{t^2}\right)\right]\leq 2\right\}.
\end{equation}
The denominator of the left-hand side of \cref{eq:covariance_Orlicz_inequality} is $\|z\|$. For the numerator we use the fact that $\Sigma_\pi^{-1/2}(Z - \EE_\pi[Z])$ is strongly log-concave with constant $\inf_{x\in\RR}\lambda_d(\Sigma_\pi^{1/2}\nabla^2U(x)\Sigma_\pi^{1/2})$. We may then use the following result:
\begin{lem}\label{lem:strong_log_concavity_implies_sub_Gaussian}
    If $X$ is an $\RR^d$-valued random variable whose density $\nu$ has strong log-concavity constant $m>0$ then $\|\langle X,x\rangle\|_{\psi_2}\leq \sqrt{\frac{8}{3}m^{-1}}\|x\|$ for all $x\in\RR^d$.
\end{lem}

\begin{proof}
    Since $\nu$ is $m$-strongly log-concave it satisfies a Bakry-Émery criterion with constant $m$ \citet[Theorem 1.2.31]{chewi2023}. This implies that it satisfies a log-Sobolev inequality with constant $m^{-1}$ \citet[Theorem 1.2.30]{chewi2023}. One can use the `Herbst argument' of the form presented in \citep[Section 2.3]{ledoux1999} to verify that $\nu$ is sub-Gaussian with variance proxy $m^{-1}$. Finally observe that
    \begin{equation}
\exp\left(\frac{s^2}{t^2}\right) = \frac{t}{2\sqrt{\pi}}\int_{-\infty}^\infty \exp\left(zs\right)\exp\left(-\frac{t^2z^2}{4}\right)\mathrm{d}z
\end{equation}
for all $t>0$ and $s\in\RR$. Hence
\begin{equation}
    \begin{split}
        \mathbb{E}\left[\exp\left(\frac{\left\langle X,x\right\rangle ^{2}}{t^{2}}\right)\right] & =\frac{t}{2\sqrt{\pi}}\int_{-\infty}^{\infty}\mathbb{E}\left[\exp\left(\left\langle zx,X\right\rangle \right)\right]\exp\left(-\frac{t^{2}z^{2}}{4}\right)\mathrm{d}z\\
 & \leq\frac{t}{2\sqrt{\pi}}\int_{-\infty}^{\infty}\exp\left(\frac{m^{-1}\left\Vert x\right\Vert ^{2}z^{2}}{2}\right)\exp\left(-\frac{t^{2}z^{2}}{4}\right)\mathrm{d}z
    \end{split}
\end{equation}
Letting $t>\sqrt{2m^{-1}}\left\Vert x\right\Vert $ this gives
\begin{equation}\begin{split}
\mathbb{E}\left[\exp\left(\frac{\left\langle X,x\right\rangle ^{2}}{t^{2}}\right)\right] & \leq\frac{t}{2\sqrt{\pi}}\sqrt{2\pi\left(\frac{t^{2}}{2}-m^{-1}\left\Vert x\right\Vert ^{2}\right)^{-1}}\\
 & =\frac{t}{\sqrt{t^{2}-2m^{-1}\left\Vert x\right\Vert ^{2}}}.
\end{split}\end{equation}
Setting the right-hand side above $= 2$ and solving for $t$ gives the result.
\end{proof}

Hence we can take
\begin{equation}
    K_\mathrm{cov} \propto \sqrt{\left(\inf_{x\in\RR}\lambda_d(\Sigma_\pi^{1/2}\nabla^2U(x)\Sigma_\pi^{1/2})\right)^{-1}} = \sqrt{m_{\Sigma_\pi^{-1}}^{-1}}.
\end{equation}

\subsubsection{Complexity result}\label{subsubsection:complexity_result_covariance_proof}

We have a bound in expectation on the first and third terms of \cref{eq:covariance_bound_decomposition_proof} and a high probability bound on the second term. Therefore we may use Markov's inequality with parameter $t>0$ on the first and third term. Combining this with a union bound gives us
\begin{equation}\begin{split}
\mathbb{P}\left(\left\Vert \Sigma_{\pi}^{-1/2}\widehat{\Sigma}\Sigma_{\pi}^{-1/2}-\mathbf{I}_{d}\right\Vert \leq t+CK_\mathrm{cov}^{2}\left(\sqrt{\frac{d+u}{N}}+\frac{d+u}{N}\right)\right)\\
\geq1-t^{-1}\left(2\sqrt{d}\beta_{\Sigma_\pi}^{-1/2}\varepsilon+\beta_{\Sigma_\pi}^{-1}\varepsilon^{2}+\left(\beta_{\Sigma_\pi}^{-1/2}\varepsilon+\sqrt{\frac{d}{N}}\right)^{2}\right)-2\exp\left(-u\right)
\end{split}\end{equation}
for all $t,u>0$. Requiring that
\begin{equation}
    \mathbb{P}\left(\left\Vert \Sigma_{\pi}^{-1/2}\widehat{\Sigma}\Sigma_{\pi}^{-1/2}-\mathbf{I}_{d}\right\Vert \leq \Delta\right) \geq 1-\delta
\end{equation}
gives us a family of constrained optimisation problems indexed by $t$ and $u$. Namely, find the maximal $\varepsilon$ and minimal $N$ satisfying
\begin{equation}\begin{split}
t+CK_\mathrm{cov}^{2}\left(\sqrt{\frac{d+u}{N}}+\frac{d+u}{N}\right) & \leq\Delta\\
t^{-1}\left(2\sqrt{d}\beta_{\Sigma_\pi}^{-1/2}\varepsilon+\beta_{\Sigma_\pi}^{-1}\varepsilon^{2}+\left(\beta_{\Sigma_\pi}^{-1/2}\varepsilon+\sqrt{\frac{d}{N}}\right)^{2}\right)+2\exp\left(-u\right) & \leq\delta
\end{split}\end{equation}
which can be rearranged into
\begin{equation}\begin{split}
N & \geq\frac{CK_\mathrm{cov}^{2}\left(d+u\right)\left(\sqrt{C}K_\mathrm{cov}+\sqrt{CK_\mathrm{cov}^{2}+4\left(\Delta-t\right)}\right)}{2\left(\Delta-t\right)}\\
0 & \geq2\beta_{\Sigma_\pi}^{-1}\varepsilon^{2}+2\sqrt{d}\left(1+\frac{1}{\sqrt{N}}\right)\beta_{\Sigma_\pi}^{-1/2}\varepsilon+\frac{d}{N}-t\left(\delta-2\exp\left(-u\right)\right).
\end{split}\end{equation}
We choose $t=\Delta/2$ and $u=\log\left(4/\delta\right)$
which give the constraints
\begin{equation}\begin{split}
N & \geq\frac{CK_\mathrm{cov}^{2}\left(d+\log\left(4/\delta\right)\right)\left(\sqrt{C}K_\mathrm{cov}+\sqrt{CK_\mathrm{cov}^{2}+2\Delta}\right)}{\Delta}\\
0 & \geq2\beta_{\Sigma_\pi}^{-1}\varepsilon^{2}+2\sqrt{d}\left(1+\frac{1}{\sqrt{N}}\right)\beta_{\Sigma_\pi}^{-1/2}\varepsilon+\frac{d}{N}-\frac{\delta\Delta}{4}.
\end{split}\end{equation}
The first constraint can be met by choosing 
\[
N\geq\frac{2CK_\mathrm{cov}^{2}\left(d+\log\left(4/\delta\right)\right)\sqrt{CK_\mathrm{cov}^{2}+2\Delta}}{\Delta}
\]
and the second constraint, which is an inequality involving a quadratic
in $\beta_{\Sigma_\pi}^{-1/2}\varepsilon$, can only be met if there exists a positive
root for this quadratic. Therefore we need
\[
N>\frac{4d}{\delta\Delta}
\]
for a positive root. Setting 
\[
N=\max\left\{ \frac{5d}{\delta\Delta},\frac{2CK_\mathrm{cov}^{2}\left(d+\log\left(4/\delta\right)\right)\sqrt{CK_\mathrm{cov}^{2}+2\Delta}}{\Delta}\right\} 
\]
satisfies both the first constraint and the positive root condition.
Given the positive root condition we must choose
\[
\beta_{\Sigma_\pi}^{-1/2}\varepsilon\leq\frac{\sqrt{d\left(1+\frac{1}{\sqrt{N}}\right)^{2}+2\left(\frac{\delta\Delta}{4}-\frac{d}{N}\right)}-\sqrt{d}\left(1+\frac{1}{\sqrt{N}}\right)}{2}.
\]
The upper bound on $\varepsilon$ can be expressed as
\[
\varepsilon\leq\beta_{\Sigma_\pi}^{1/2}\frac{2\left(\frac{\delta\Delta}{4}-\frac{d}{N}\right)}{2\left(\sqrt{d\left(1+\frac{1}{\sqrt{N}}\right)^{2}+2\left(\frac{\delta\Delta}{4}-\frac{d}{N}\right)}+\sqrt{d}\left(1+\frac{1}{\sqrt{N}}\right)\right)}.
\]
Therefore (using a lower bound on the upper bound above) it suffices
to take
\[
\varepsilon\leq\beta_{\Sigma_\pi}^{1/2}\frac{\frac{\delta\Delta}{4}-\frac{d}{N}}{2\sqrt{d\left(1+\frac{1}{\sqrt{N}}\right)^{2}+2\left(\frac{\delta\Delta}{4}-\frac{d}{N}\right)}}.
\]
We know that $N\leq5d/\left(\delta\Delta\right)$ hence it suffices
to take
\[
\varepsilon\leq\frac{1}{40}\delta\Delta\beta_{\Sigma_\pi}^{1/2}\frac{1}{\sqrt{d\left(1+\frac{1}{\sqrt{N}}\right)^{2}+2\left(\frac{\delta\Delta}{4}-\frac{d}{N}\right)}}.
\]
We have 
\begin{equation}\begin{split}
d\left(1+\frac{1}{\sqrt{N}}\right)^{2}+2\left(\frac{\delta\Delta}{4}-\frac{d}{N}\right) & \leq4d+\frac{\delta\Delta}{2}
  =\frac{1}{2}\left(8d+\delta\Delta\right)
 \leq\frac{1}{2}\left(8d+1\right)
 \leq\frac{9}{2}d.
\end{split}\end{equation}
Substituting this in gives that it suffices to take
\[
\varepsilon=\frac{\sqrt{2}}{120}\frac{\delta\Delta}{\sqrt{d}}\sqrt{\beta_{\Sigma_\pi}}.
\]

\subsection{Proof of Theorem \ref{thm:preconditioner_learn_complexity} Part 2.}\label{proof:preconditioner_learn_complexity_part_2}

Let $\{X_t\}_{t=1}^N$ be the $\sqrt{N}\varepsilon$-AIID output of the thinned Markov chain sampler in Algorithm \ref{alg:thinned_Markov_chain}. Let $S \in \RR^{d \times N}$ be the matrix whose $j$th column is $\nabla\log\pi(X_j)$ for $j\in[N]$ and let $T \in \RR^{d\times N}$ be the matrix whose $j$th column is $\nabla\log\pi(Y_t)$ with $\{Y_t\}_{t=1}^N\sim\pi^{\otimes N}$. Then we have
\[\label{eq:Fisher_bound_decomposition_proof}
    &\hspace{-2em}\|\Ff^{-1/2}\widehat{\Ff}\Ff^{-1/2} - \mathbf{I}_d\| \\
    & \leq \|\frac{1}{N}\Ff^{-1/2}SS^\top\Ff^{-1/2} - \frac{1}{N}\Ff^{-1/2}TT^\top\Ff^{-1/2}\| + \|\frac{1}{N}\Ff^{-1/2}TT^\top\Ff^{-1/2} - \mathbf{I}_d\|.
\]
As in the proof of Theorem \ref{thm:preconditioner_learn_complexity} part 1. we will bound the first term on the right-hand side in expectation and the second term with high probability.

\subsubsection{Bounding the first term of {Eq. (\ref{eq:Fisher_bound_decomposition_proof})}}

Let $\{\tilde{S}_t\}_{t=1}^N := \{\Ff^{-1/2}\nabla\log\pi(X_t)\}_{t=1}^N$. Using Proposition \ref{prop:rootNIID_Lipschitz_transformation} we have that $\{\tilde{S}_t\}_{t=1}^N$ is\\ $\sqrt{N}\|\Ff^{-1/2}\nabla\log\pi\|_\mathrm{Lip}\varepsilon$-approximately IID from $\tilde{\pi}$ in $W_2$ where $\tilde{\pi} = \Ff^{-1/2}\nabla\log\pi_\sharp (\pi)$. The potential of $\pi$ is $L$-smooth and hence the Lipschitz constant of $\nabla\log\pi$ is $L$. By hypothesis we have $\alpha_\Ff\mathbf{I}_d \preceq \Ff$ and so the Lipschitz constant of pre-multiplying with $\Ff^{-1/2}$ is $\alpha_\Ff^{-1/2}$. Hence $\|\Ff^{-1/2}\nabla\log\pi\|_\mathrm{Lip} \leq L\alpha_\Ff^{-1/2}$.

Let $\tilde{S}\in\RR^{d\times N}$ be the matrix whose $j$th column is $\tilde{S}_j$ for $j\in [N]$. Let $\{\tilde{T}_t\}_{t=1}^N := \{\Ff^{-1/2}\nabla\log\pi(Y_t)\}_{t = 1}^N$ and let $\tilde{T}\in\RR^{d\times N}$ be the matrix whose $j$th column is $T_j$. By the law of the unconscious statistician we have that $\tilde{\pi}$ is centred. Combining this with the fact that $\{\tilde{S}_t\}_{t=1}^N$ is $\sqrt{N}L\alpha_\Ff^{-1/2}\varepsilon$ approximately IID from $\tilde{\pi}$ in $W_2$ and the fact that $\{\tilde{T}_t\}_{t=1}^N$ is sampled from $\tilde{\pi}^{\otimes N}$ we can apply Proposition \ref{prop:rootNIID_implies_empirical_covariance_bound}:

\begin{equation}
\begin{split}
    \EE\left[\|\frac{1}{N}\Ff^{-1/2}SS^\top\Ff^{-1/2} - \frac{1}{N}\Ff^{-1/2}TT^\top\Ff^{-1/2}\|\right] &= \EE\left[\|\frac{1}{N}\tilde{S}\tilde{S}^\top - \frac{1}{N}\tilde{T}\tilde{T}^\top\|\right]\\
    &\leq \EE\left[\|\frac{1}{N}\tilde{S}\tilde{S}^\top - \frac{1}{N}\tilde{T}\tilde{T}^\top\|_F\right]\\
    &\leq 2\sqrt{\mathrm{tr}\left(\Sigma_{\tilde{\pi}}\right)}L\alpha_\Ff^{-1/2}\varepsilon+L^2\alpha_\Ff^{-1}\varepsilon^2\\
    &= 2\sqrt{d}L\alpha_\Ff^{-1/2}\varepsilon+L^2\alpha_\Ff^{-1}\varepsilon^2\\
\end{split}
\end{equation}
where the expectation is under the $W_2$-optimal coupling between $\{\tilde{S}_t\}_{t=1}^N$ and $\{\tilde{T}_t\}_{t=1}^N$.

\subsubsection{Bounding the second term of {Eq. (\ref{eq:Fisher_bound_decomposition_proof})}}

Note that $N^{-1}\Ff^{-1/2}TT^\top\Ff^{-1/2}$ is an empirical covariance matrix (without centre-ing) of IID samples from $\tilde{\pi}$. Therefore we may apply \citet[Exercise 4.7.3]{vershynin2018} giving:
\begin{equation}
    \|\frac{1}{N}\Ff^{-1/2}TT^\top\Ff^{-1/2} - \mathbf{I}_d\| \leq CK_\mathrm{Fisher}^2\left(\sqrt{\frac{d + u}{N}} + \frac{d + u}{N}\right)
\end{equation}
for all $u\geq 0$ with probability at least $1-2\exp(-u)$ where $C>0$ is an absolute constant and $K_\mathrm{Fisher} > 0$ is such that for $Z\sim\pi$
\begin{equation}\label{eq:Fisher_orlicz_inequality_proof}
    \frac{\left\Vert \left\langle \mathcal{F}^{-1/2}\nabla\log\pi\left(Z\right),z\right\rangle \right\Vert _{\psi_{2}}}{\sqrt{\mathbb{E}_{\pi}\left[\left\langle \mathcal{F}^{-1/2}\nabla\log\pi\left(Z\right),z\right\rangle ^{2}\right]}}\leq K_{\mathrm{Fisher}}
\end{equation}
for all $z\in\mathbb{R}^{d}$.
The denominator on the left hand side of \cref{eq:Fisher_orlicz_inequality_proof} is $\left\Vert z\right\Vert $.
To work out the sub-Gaussian norm we use a similar argument to the
proof of Lemma \ref{lem:strong_log_concavity_implies_sub_Gaussian}. However we can skip straight to the sub-Gaussian
constant of $\mathcal{F}^{-1/2}\nabla\log\pi\left(Z\right)$ with
\citep[Theorem 2.2]{negrea2022}. This result has that if you push a distribution
whose potential is $L$-smooth through its own score, the resulting
distribution is sub-Gaussian with variance proxy $L$. First note
that $\mathcal{F}^{-1/2}\nabla\log\pi$ is the score of $\tilde{Z}=\mathcal{F}^{1/2}Z$.
Therefore the smoothness constant is $\sup_{x\in\mathbb{R}^{d}}\lambda_{1}\left(\mathcal{F}^{-1/2}\nabla^{2}U\left(x\right)\mathcal{F}^{-1/2}\right)$.
Hence, by \citep[Theorem 2.2]{negrea2022}, $\mathcal{F}^{-1/2}\nabla\log\pi\left(Z\right)$
is sub-Gaussian with variance proxy $\sup_{x\in\mathbb{R}^{d}}\lambda_{1}\left(\mathcal{F}^{-1/2}\nabla^{2}U\left(x\right)\mathcal{F}^{-1/2}\right)$.
The rest of the proof of Lemma \ref{lem:strong_log_concavity_implies_sub_Gaussian} is the same,
replacing $X$ with $\mathcal{F}^{-1/2}\nabla\log\pi\left(Z\right)$
and $m^{-1}$ with $\sup_{x\in\mathbb{R}^{d}}\lambda_{1}\left(\mathcal{F}^{-1/2}\nabla^{2}U\left(x\right)\mathcal{F}^{-1/2}\right)$.
Therefore we get that 
\[
K_{\mathrm{Fisher}}\propto\sqrt{\sup_{x\in\mathbb{R}^{d}}\lambda_{1}\left(\mathcal{F}^{-1/2}\nabla^{2}U\left(x\right)\mathcal{F}^{-1/2}\right)}=\sqrt{L_\Ff}.
\]

\subsubsection{Complexity result}\label{subsubsection:complexity_result_Fisher_proof}

We have a bound in expectation on the first term of \cref{eq:Fisher_bound_decomposition_proof} and a high probability bound on the second term. Therefore we may use Markov's inequality with parameter $t>0$ on the first term. Combining this with a union bound gives us

\begin{equation}\begin{split}
\mathbb{P}\left(\left\Vert \mathcal{F}^{-1/2}\widehat{\Ff}\mathcal{F}^{-1/2}-\mathbf{I}_{d}\right\Vert \leq t+cK_\mathrm{Fisher}^{2}\left(\sqrt{\frac{d+u}{N}}+\frac{d+u}{N}\right)\right)\\
\qquad\geq1-t^{-1}\left(2\sqrt{d}\frac{L}{\sqrt{\alpha_\Ff}}\varepsilon+\frac{L^{2}}{\alpha_\Ff}\varepsilon^{2}\right)-2\exp\left(-u\right)
\end{split}\end{equation}
for all $t,u>0$. Requiring that
\begin{equation}
    \PP\left(\left\Vert \mathcal{F}^{-1/2}\widehat{\Ff}\mathcal{F}^{-1/2}-\mathbf{I}_{d}\right\Vert \leq \Delta\right)\geq1-\delta
\end{equation}
gives us a family of constrained optimisation problems indexed by $t$ and $u$. Namely, find the maximal $\varepsilon$ and minimal $N$ satisfying
\begin{equation}\begin{split}
t+cK_\mathrm{Fisher}^{2}\left(\sqrt{\frac{d+u}{N}}+\frac{d+u}{N}\right) & \leq\Delta\\
2\exp\left(-u\right)+\frac{1}{t}\left(2\sqrt{d}\frac{L}{\sqrt{\alpha_\Ff}}\varepsilon+\frac{L^{2}}{\alpha_\Ff}\varepsilon^{2}\right) & \leq\delta.
\end{split}\end{equation}
Hence
\begin{equation}\begin{split}
N & \geq\frac{cK^{2}\left(d+u\right)\left(\sqrt{c}K+\sqrt{cK^{2}+4\left(\Delta-t\right)}\right)}{2\left(\Delta-t\right)}\\
\varepsilon & \leq\frac{\sqrt{\alpha_\Ff d}}{L}\left(\sqrt{1+\frac{4\left(\delta-2\exp\left(-u\right)\right)t}{d}}-1\right)
\end{split}\end{equation}
Choosing $t = \delta/2$ and $u = \log(4/\delta)$ gives
\begin{equation}\begin{split}
N & \geq\frac{cK^{2}\left(d+\log(4/\delta)\right)\left(\sqrt{c}K+\sqrt{cK^{2}+2\Delta}\right)}{\Delta}\\
\varepsilon & \leq\frac{\sqrt{\alpha_\Ff d}}{L}\left(\sqrt{1+\frac{\delta\Delta}{d}}-1\right)
\end{split}\end{equation}
Hence it suffices to take
\begin{equation}\begin{split}
N &=\frac{2cK^{2}\left(d+\log(4/\delta)\right)\sqrt{cK^{2}+2\Delta}}{\Delta}\\
\varepsilon &=\frac{3}{8}\frac{\delta\Delta}{L\sqrt{d}}\sqrt{\alpha_\Ff}
\end{split}\end{equation}
where we use the fact that $\sqrt{1+s} - 1 \geq (3/8)s$ for $s\in[0,1/2]$ to set $\varepsilon$.

\subsection{Proof of Theorem \ref{thm:total_ULA_complexities} Part 1.}\label{subsection:proof_of_unpreconditioned_ULA_complexity}

Choose $h = 100^{-1}d^{-1}\kappa^{-2}\varepsilon^2$. The upper bound $\varepsilon \leq \frac{10\kappa\sqrt{d}}{\sqrt{L}}$ makes $h\leq2 / (L + m)$. Therefore the unpreconditioned ULA is a $(\Gamma, \gamma, b)$-$W_2$ contraction to $\pi$ by \citet[Theorem 1a)]{dalalyan2019} with $\Gamma = 1$, $\gamma = mh$, and $b = (33/20)\kappa\sqrt{dh}$ stated above. The choice of $h$ also automatically makes $3\Gamma^2b<\varepsilon$, satisfying the lower bound constraint from Theorem \ref{thm:Wasserstein_contraction_complexity}. The upper bound $\varepsilon \leq \sqrt{3\mathrm{tr}(\Sigma_\pi)}$ means we can apply Theorem \ref{thm:Wasserstein_contraction_complexity} with the appropriate values of $\Gamma$, $\gamma$, and $b$.

\subsection{Proof of Theorem \ref{thm:total_ULA_complexities} Part 2.}\label{proof:total_ULA_complexities_part_2}

Here we use Algorithm \ref{alg:preconditioned_Markov_chain} with the covariance-based preconditioner and ULA as the underlying Markov kernel.

\subsubsection{Time to construct the preconditioner}

Let $\Delta,\delta\in(0, 1)$ set the precision with which we learn the preconditioner in step 1. We wish to produce a $\sqrt{N}\varepsilon$-AIID output with which to construct our preconditioner so let

\begin{equation}
    \varepsilon = \frac{\sqrt{2}}{120}\frac{\delta\Delta}{\sqrt{dL}}
\end{equation}

in accordance with Theorem \ref{thm:preconditioner_learn_complexity} Part 1., where we take $\beta_{\Sigma_\pi}=L^{-1}$ under $L$-smoothness of the potential. Choose a step size $h_1 = k\varepsilon^2$ with $k>0$ and set the constraint
\begin{equation}
    k\leq\frac{120^2}{2}dL\frac{2}{L + m}
\end{equation}
so that $h_1 \leq 2 / (L + m)$. This makes the ULA kernel a $(\Gamma,\gamma,b)$-$W_2$ contraction to $\pi$ by \citet[Theorem 1a)]{dalalyan2019} with $\Gamma = 1$, $\gamma = mh$ and $b = (33/20)\kappa\sqrt{dh_1}$. The bias $b$ is parametrised by $h_1$ so choose 
\begin{equation}
    k = \frac{100}{99^2}\frac{1}{\kappa^2d}\leq \frac{120^2}{2}dL\frac{2}{L + m}
\end{equation}
to ensure that $3\Gamma^2b=(1/2)\varepsilon < \varepsilon$. The fact that $\delta\Delta < 1$ and $\mathrm{tr}(\Sigma_\pi)\geq L^{-1}d$ automatically give $\varepsilon\leq\sqrt{3\mathrm{tr}(\Sigma_\pi)}$. Therefore we have the conditions for Theorem \ref{thm:Wasserstein_contraction_complexity} to give a time complexity to achieve a $\sqrt{N}\varepsilon$-AIID ensemble with the thinned Markov chain sampler with an ULA kernel. Finally choose
\begin{equation}
    N = \max\left\{ \frac{5d}{\delta\Delta},\frac{2CK_\mathrm{cov}^{2}\left(d+\log\left(4/\delta\right)\right)\sqrt{CK_\mathrm{cov}^{2}+2\Delta}}{\Delta}\right\}
\end{equation}
in accordance with Theorem \ref{thm:preconditioner_learn_complexity} Part 1. For these values of $\varepsilon$ and $N$, Theorem \ref{thm:preconditioner_learn_complexity} part 1. gives us that
\begin{equation}
    \|\Sigma_\pi^{-1/2}\widehat{\Sigma}_\pi\Sigma_\pi^{-1/2} - \mathbf{I}_d\|\leq \Delta\textup{ with probability }\geq 1 - \delta.
\end{equation}
Choosing $\Delta = 1/2$ gives the learning complexity of
\begin{equation}
    \tilde{O}(d^3(d + \mathtt{G})\kappa^3\delta^{-2}\max\left\{\delta^{-1},K^3_\mathrm{cov}\right\})
\end{equation}
after substituting in the values of $N$, $\gamma$, $\Gamma$, and $b$.

\subsubsection{Time to achieve the \texorpdfstring{$\sqrt{N}\varepsilon$}{ε√N}-AIID output \texorpdfstring{$\{X_t\}_{t=1}^N$}{\{Xₜ\}ₜ₌₁ᴺ}} We now wish to evaluate the iteration complexity with the preconditioned kernel $K_{\widehat{\Sigma}_\pi^{-1}, \pi}$ defined in Definition \ref{def:preconditioned_kernel} where
\begin{equation}
    \|\Sigma_\pi^{-1/2}\widehat{\Sigma}_\pi\Sigma_\pi^{-1/2} - \mathbf{I}_d\|\leq \frac{1}{2}\textup{ with probability }\geq 1 - \delta.
\end{equation}
To do this we prove that $K_{(\widehat{\Sigma}_\pi^{-1/2})_\sharp \pi}$ is a $(\Gamma, \gamma, b)$-$W_2$ contraction to $(\widehat{\Sigma}_\pi^{-1/2})_\sharp\pi$ as in Definition \ref{def:W_2_contraction}, and ascertain the contraction parameters of $K_{\widehat{\Sigma}_\pi^{-1}, \pi}$according using Corollary \ref{cor:approximate_preconditioned_kernel_efficiency}. Choose $h = 100^{-1}d^{-1}\kappa_{\widehat{\Sigma}_\pi^{-1}}^{-2}\varepsilon^2$. We wish to verify that
\begin{equation}
    \varepsilon \leq \frac{10\kappa_{\widehat{\Sigma}_\pi^{-1}}\sqrt{d}}{\sqrt{L_{\widehat{\Sigma}_\pi^{-1}}}}
\end{equation}
in order to achieve $h \leq 2 / (L_{\widehat{\Sigma}_\pi^{-1}} + m_{\widehat{\Sigma}_\pi^{-1}})$ to get a $(\Gamma, \gamma, b)$-$W_2$ contraction to $(\widehat{\Sigma}_\pi^{-1/2})_\sharp\pi$ of $K_{(\widehat{\Sigma}_\pi^{-1/2})_\sharp\pi}$. Corollary \ref{cor:conditioning_comparison_covariance_case} gives us that $L_{\widehat{\Sigma}_\pi^{-1}} \leq (1-\Delta)^{-1}L_{{\Sigma}_\pi^{-1}}$ and $\kappa_{\Sigma_\pi^{-1}} \leq (1 + \Delta)(1 - \Delta)^{-1}\kappa_{\widehat{\Sigma}_\pi^{-1}}$ with $\Delta = 0.5$, and so the requirement
\begin{equation}
    \varepsilon\leq \frac{10}{3\sqrt{2}}\frac{\kappa_{\Sigma_\pi^{-1}}\sqrt{d}}{\sqrt{L_{\Sigma_\pi^{-1}}}}
\end{equation}
from the theorem statement gives us the assumption we wish to verify. Therefore $K_{(\widehat{\Sigma}_\pi^{-1/2})_\sharp\pi}$ is a $(\Gamma,\gamma,b)$-$W_2$ contraction to $(\widehat{\Sigma}_\pi^{-1/2})_\sharp\pi$ with
\[
\Gamma = 1\textup{, }\gamma = m_{\widehat{\Sigma}_\pi^{-1}}h\textup{ and }b = \frac{33}{20}\kappa_{\widehat{\Sigma}_\pi^{-1}}.
\]
Corollaries \ref{cor:conditioning_comparison_covariance_case} and \ref{cor:approximate_preconditioned_kernel_efficiency} give us that $K_{\widehat{\Sigma}_\pi^{-1}, \pi}$ is a $(\Gamma, \gamma, b)$-$W_2$ contraction to $\pi$ with
\[
\Gamma = \sqrt{3\kappa(\Sigma_\pi)}\textup{, }\gamma = \frac{2}{3}m_{\Sigma_\pi^{-1}}h\textup{ and }b =3\sqrt{2}\|\Sigma_\pi^{1/2}\| \frac{33}{20}\kappa_{\Sigma_\pi^{-1}}.
\]
Finally to apply Theorem \ref{thm:Wasserstein_contraction_complexity} we wish to verify that $\varepsilon \leq \sqrt{3\mathrm{tr}(\Sigma_{\tilde{\pi}})}$ where $\tilde{\pi}$ is the pushforward of $\pi$ through $\widehat{\Sigma}_\pi^{-1/2}$. For this we can use the fact that $\widehat{\Sigma}_\pi$ and $\Sigma_\pi$ are quantifiably close:
    \begin{equation}\begin{split}
\mathrm{tr}\left(\Sigma_{\tilde{\pi}}\right) & =\mathrm{tr}\left(\widehat{\Sigma}_{\pi}^{-1/2}\Sigma_{\pi}\widehat{\Sigma}_{\pi}^{-1/2}\right)\\
 & =\sum_{i=1}^{d}\lambda_{i}\left(\widehat{\Sigma}_{\pi}^{-1/2}\Sigma_{\pi}\widehat{\Sigma}_{\pi}^{-1/2}\right)\\
 & =\sum_{i=1}^{d}\lambda_{i}\left(\widehat{\Sigma}_{\pi}^{1/2}\Sigma_{\pi}^{-1}\widehat{\Sigma}_{\pi}^{1/2}\right)^{-1}\\
 & =\sum_{i=1}^{d}\lambda_{i}\left(\Sigma_{\pi}^{-1/2}\widehat{\Sigma}_{\pi}\Sigma_{\pi}^{-1/2}\right)^{-1}\\
 & \geq\frac{1}{1+\Delta}d\\
 & =\frac{2}{3}d.
\end{split}\end{equation}
Therefore the condition $\varepsilon \leq \sqrt{2d}$ from the theorem statement implies $\varepsilon \leq \sqrt{3\mathrm{tr}(\Sigma_{\tilde{\pi}})}$. Applying Theorem \ref{thm:Wasserstein_contraction_complexity} gives a time (in iterations) to collect the $\sqrt{N}\varepsilon$-AIID samples as
\[
T_\mathrm{collect} \in \tilde{O}\left(\frac{1}{m_{\Sigma_\pi^{-1}}}d\kappa^2_{\widehat{\Sigma}_\pi^{-1}}\frac{N}{\varepsilon^2}\right)
\]
and hence applying Corollary \ref{cor:conditioning_comparison_covariance_case} gives the result.

\subsection{Proof of Theorem \ref{thm:total_ULA_complexities} Part 3.}\label{proof:total_ULA_complexities_part_3}

Here we use Algorithm \ref{alg:preconditioned_Markov_chain} with the Fisher preconditioner and ULA as the underlying Markov kernel

\subsubsection{Time to construct the preconditioner}

Let $\Delta,\delta\in(0, 1)$ set the precision with which we learn the preconditioner in step 1. We wish to produce a $\sqrt{N}\varepsilon$-AIID output with which to construct our preconditioner so let

\begin{equation}
    \varepsilon = \frac{3}{8}\frac{\delta\Delta}{\sqrt{dL
    \kappa}}
\end{equation}

in accordance with Theorem \ref{thm:preconditioner_learn_complexity} Part 2., where we have taken $\alpha_\Ff = m$ under $m$-strong convexity of the potential. Choose a step size $h_1 = k\varepsilon^2$ with $k>0$ and set the constraint
\begin{equation}
    k\leq\frac{64}{9}dL\kappa\frac{2}{L + m}
\end{equation}
so that $h_1 \leq 2 / (L + m)$. This makes the ULA kernel a $(\Gamma,\gamma,b)$-$W_2$ contraction to $\pi$ by \citet[Theorem 1a)]{dalalyan2019} with $\Gamma = 1$, $\gamma = mh$ and $b = (33/20)\kappa\sqrt{dh_1}$. The bias $b$ is parametrised by $h_1$ so choose 
\begin{equation}
    k = \frac{100}{99^2}\frac{1}{\kappa^2d}\leq \frac{64}{9}dL\kappa\frac{2}{L + m}
\end{equation}
to ensure that $3\Gamma^2b=(1/2)\varepsilon < \varepsilon$. The fact that $\delta\Delta < 1$ and $\mathrm{tr}(\Sigma_\pi)\geq L^{-1}d$ automatically give $\varepsilon\leq\sqrt{3\mathrm{tr}(\Sigma_\pi)}$. Therefore we have the conditions for Theorem \ref{thm:Wasserstein_contraction_complexity} to give a time complexity to achieve a $\sqrt{N}\varepsilon$-AIID ensemble with the thinned Markov chain sampler with an ULA kernel. Finally choose
\begin{equation}
    N = \frac{2CK_\mathrm{Fisher}^{2}\left(d+\log\left(4/\delta\right)\right)\sqrt{CK_\mathrm{Fisher}^{2}+2\Delta}}{\Delta}
\end{equation}
in accordance with Theorem \ref{thm:preconditioner_learn_complexity} Part 2. For these values of $\varepsilon$ and $N$, Theorem \ref{thm:preconditioner_learn_complexity} part 2. gives us that
\begin{equation}
    \|\Ff^{-1/2}\widehat{\Ff}_\pi\Ff^{-1/2} - \mathbf{I}_d\|\leq \Delta\textup{ with probability }\geq 1 - \delta.
\end{equation}
Choosing $\Delta = 1/2$ gives the learning complexity of
\begin{equation}
    \tilde{O}(d^3(d + \mathtt{G})\kappa^4\delta^{-2}K^3_\mathrm{Fisher})
\end{equation}
after subsituting in the values of $N$, $\gamma$, $\Gamma$, and $b$.

\subsubsection{Time to achieve the \texorpdfstring{$\sqrt{N}\varepsilon$}{ε√N}-AIID output \texorpdfstring{$\{X_t\}_{t=1}^N$}{\{Xₜ\}ₜ₌₁ᴺ}}

We now wish to evaluate the iteration complexity with the preconditioned kernel $K_{\widehat{\Ff}, \pi}$ defined in Definition \ref{def:preconditioned_kernel} where
\begin{equation}
    \|\Ff^{-1/2}\widehat{\Ff}_\pi\Ff^{-1/2} - \mathbf{I}_d\|\leq \frac{1}{2}\textup{ with probability }\geq 1 - \delta.
\end{equation}
We can apply a similar proof technique to the covariance case. Choose $h = 100^{-1}d^{-1}\kappa_{\widehat{\Sigma}_\pi^{-1}}^{-2}\varepsilon^2$. First we wish to verify that
\begin{equation}
    \varepsilon \leq \frac{10\kappa_{\widehat{\Ff}}\sqrt{d}}{\sqrt{L_{\widehat{\Ff}}}}
\end{equation}
so that $h \leq 2 / (L_{\widehat{\Ff}} + m_{\widehat{\Ff}})$ which, by \citet{dalalyan2019}, means that $K_{\widehat{\Ff}^{1/2}_\sharp \pi}$ is a $(\Gamma,\gamma, b)$-$W_2$ contraction to $\widehat{\Ff}^{1/2}_\sharp \pi$. The Fisher version of Corollary \ref{cor:conditioning_comparison_covariance_case} gives us that $L_{\widehat{\Ff}} \leq (1+\Delta)L_{\Ff}$ and $\kappa_{\Ff} \leq (1 + \Delta)(1 - \Delta)^{-1}\kappa_{\widehat{\Ff}}$ with $\Delta = 0.5$, and so the requirement
\begin{equation}
    \varepsilon\leq \frac{10\sqrt{2}}{3\sqrt{3}}\frac{\kappa_{\Ff}\sqrt{d}}{\sqrt{L_{\Ff}}}
\end{equation}
from the theorem statement gives us the assumption we wish to verify. Therefore $h \leq 2 / (L_{\widehat{\Ff}} + m_{\widehat{\Ff}})$ and so $K_{\widehat{\Ff}^{1/2}_\sharp \pi}$ is a $(\Gamma, \gamma, b)$-$W_2$ contraction to $\widehat{\Ff}^{1/2}_\sharp\pi$ with
\[
\Gamma = 1\textup{, }\gamma = m_{\widehat{\Ff}}h\textup{ and }b = \frac{33}{20}\kappa_{\widehat{\Ff}}.
\]
Corollaries \ref{cor:approximate_preconditioned_kernel_efficiency} and \ref{cor:condition_number_comparison} give that $K_{\widehat{\Ff}, \pi}$ is a $(\Gamma, \gamma, b)$-$W_2$ contraction to $\pi$ with constants
\[
\Gamma = \sqrt{3\kappa(\Ff)}\textup{ , }\gamma = \frac{1}{2}m_\Ff h\textup{ and }b = \frac{3}{\sqrt{2}}\|\Ff^{-1/2}\|\frac{33}{20}\kappa_{\Ff}.
\]
Second to apply Theorem \ref{thm:Wasserstein_contraction_complexity} we wish to verify that $\varepsilon \leq \sqrt{3\mathrm{tr}(\Sigma_{\tilde{\pi}})}$ where $\tilde{\pi}$ is the pushforward of $\pi$ through $\widehat{\Ff}^{1/2}$. For this we can use the fact that $\widehat{\Ff}$ and $\Ff$ are quantifiably close:
\begin{equation}\begin{split}
\mathrm{tr}\left(\Sigma_{\tilde{\pi}}\right) & =\mathrm{tr}\left(\widehat{\mathcal{F}}^{1/2}\Sigma_{\pi}\widehat{\mathcal{F}}_{}^{1/2}\right)\\
 & \geq\mathrm{tr}\left(\widehat{\mathcal{F}}^{1/2}\mathcal{F}^{-1}\widehat{\mathcal{F}}^{1/2}\right)\\
 & =\mathrm{tr}\left(\mathcal{F}^{-1/2}\widehat{\mathcal{F}}\mathcal{F}^{-1/2}\right)\\
 & \geq\left(1-\Delta\right)d\\
 & =\frac{1}{2}d
\end{split}\end{equation}
where the second line follows by the Cramér-Rao lower bound. Therefore the condition $\varepsilon \leq \sqrt{(3/2)d}$ from the theorem statement suffices to verify $\varepsilon \leq \sqrt{3\mathrm{tr}(\Sigma_{\tilde{\pi}})}$. Applying Theorem \ref{thm:Wasserstein_contraction_complexity}
 gives an iteration complexity of
 \[
 T_\mathrm{collect} \in \tilde{O}\left(\frac{d\kappa^2_{\widehat{\Ff}}}{m_\Ff}\frac{N}{\varepsilon^2}\right).
 \]
 Using Corollary \ref{cor:conditioning_comparison_Fisher_case} to compare $\kappa_{\widehat{\Ff}}$ and $\kappa_{\Ff}$ gives the result.
\subsection{Proof of Theorem \ref{thm:total_unadjusted_underdamped_complexities}}\label{proof:total_unadjusted_underdamped_complexities}
    For the unpreconditioned case let the step-size be $h = k\varepsilon$ for $k>0$. Set
    \begin{equation}
        k=\frac{1}{1536}\frac{1}{\kappa}\sqrt{\frac{5}{2\Ee_K}}
    \end{equation}
    to satisfy $3\Gamma^2b = (1/2)\varepsilon < \varepsilon$. We need $h\leq 1$ for the Markov kernel to be a Wasserstein contraction of the form in Definition \ref{def:W_2_contraction}. Therefore we need
    \begin{equation}
        \varepsilon\leq1536\kappa\sqrt{\frac{2\Ee_K}{5}}
    \end{equation}
    which holds by hypothesis. We have $\varepsilon\leq\sqrt{3\mathrm{tr}(\Sigma_\pi)}$ and so we can apply Theorem \ref{thm:Wasserstein_contraction_complexity}, subbing in the values of $\Gamma$, $\gamma$, and $b$.
    
    Now we look at the covariance preconditioned case. This proof works similarly to the proof of the complexity with ULA as the underlying kernel in Section \ref{proof:total_ULA_complexities_part_2}. Let $\delta,\Delta>0$. For the preconditioner learning phase, set the step-size $h=k\varepsilon$ where
    \begin{equation}
        \varepsilon = \frac{\sqrt{2}}{120}\frac{\delta\Delta}{\sqrt{dL}}
    \end{equation}
    in accordance with Theorem \ref{thm:preconditioner_learn_complexity}. Impose that
    \begin{equation}
        k\leq\frac{120}{\sqrt{2}}\sqrt{dL}
    \end{equation}
    to make $h\leq 1$. Then the underlying kernel is a $(\Gamma, \gamma, b)$-$W_2$ contraction to $\pi$ with the $\Gamma$, $\gamma$ and $b$ as in the discussion before Theorem \ref{thm:total_unadjusted_underdamped_complexities}. Set
    \begin{equation}
        k = \frac{1}{1536}\frac{1}{\kappa}\sqrt{\frac{5}{2\Ee_K}}
    \end{equation}
    to make $3\Gamma^2b = (1/2)\varepsilon < \varepsilon$. To verify that this value of $k$ satisfies its upper bound we simply impose that $L\Dd^2 \geq 1$, which is done by hypothesis. The fact that $\delta\Delta<1$ and $\mathrm{tr}(\Sigma_\pi)\geq L^{-1}d$ gives $\varepsilon \leq\sqrt{3\mathrm{tr}(\Sigma_\pi)}$. Therefore we can apply Theorems \ref{thm:Wasserstein_contraction_complexity} and \ref{thm:preconditioner_learn_complexity} to get the complexity to achieve
    \begin{equation}
        \left\|\Sigma_\pi^{-1/2}\widehat{\Sigma}_\pi\Sigma_\pi^{-1/2} - \mathbf{I}_d\right\|\leq \Delta\textup{ with probability }\geq 1-\delta
    \end{equation}
    which is
    \begin{equation}
        \tilde{O}\left(\frac{\kappa^2d^{3/2}\sqrt{L}\sqrt{\Ee_K}}{\delta}\max\left\{\delta^{-1},K_\mathrm{cov}^3\right\}\right)
    \end{equation}
    after subbing in the values of $N$, $\Gamma$, $\gamma$, and $b$ in terms of the step-size. The rest of the proof works similarly to Section \ref{proof:total_ULA_complexities_part_2}: let the step-size be $h = k\varepsilon$ for $k>0$. Set
    \[
    k=\frac{1}{1536}\frac{1}{\kappa_{\widehat{\Sigma}_\pi^{-1}}}\sqrt{\frac{5}{2\Ee_K^{(\widehat{\Sigma}_\pi)}}}\textup{ where }\Ee_K^{(\widehat{\Sigma}_\pi)} = 26\left(\frac{d}{m_{\widehat{\Sigma}_\pi^{-1}}} + \mathcal{D}^2\right)
    \]
    to satisfy the assumption $3\Gamma^2b = (1/2)\varepsilon < \varepsilon$ from Theorem \ref{thm:Wasserstein_contraction_complexity} where
    \[
    \Gamma = 4\textup{ and }b = 16\kappa_{\widehat{\Sigma}_\pi^{-1}}\sqrt{\frac{2\mathcal{E}_K^{(\widehat{\Sigma}_\pi)}}{5}}h
    \]
    in the specifica case of the underdamped unadjusted kernel $K_{(\widehat{\Sigma}_\pi^{-1/2})_\sharp\pi}$. We need $h\leq 1$ for the Markov kernel $K_{(\widehat{\Sigma}_\pi^{-1/2})_\sharp\pi}$ to be a Wasserstein contraction. The upper bound on $\varepsilon$ in the Theorem statement suffices, after using the relation $m_{\widehat{\Sigma}_\pi^{-1}} \leq (3/2)m_{\Sigma_\pi^{-1}}$ from Corollary \ref{cor:conditioning_comparison_covariance_case}. Therefore $K_{(\widehat{\Sigma}_\pi^{-1/2})_\sharp\pi}$ is a $(\Gamma, \gamma, b)$-$W_2$ contraction to $(\widehat{\Sigma}_\pi^{-1/2})_\sharp\pi$ with
    \[
    \Gamma = 4\textup{ , }\gamma = \frac{h}{2\kappa_{\widehat{\Sigma}_\pi^{-1}}}\textup{ and }b = 16\kappa_{\widehat{\Sigma}_\pi^{-1}}\sqrt{\frac{2\mathcal{E}_K^{(\widehat{\Sigma}_\pi)}}{5}}h.
    \]
    Therefore using Corollary \ref{cor:approximate_preconditioned_kernel_efficiency} with $\Delta = 0.5$ gives us that $K_{\widehat{\Sigma}_\pi^{-1}, \pi}$ is a $(\Gamma, \gamma, b)$-$W_2$ contraction to $\pi$ with
    \[
    \Gamma = 4\sqrt{3\kappa(\Sigma_\pi)}\textup{ , }\gamma = \frac{h}{2\kappa_{\widehat{\Sigma}_\pi^{-1}}}\textup{ and }b = 16\sqrt{2}\|\Sigma_\pi^{1/2}\|\kappa_{\widehat{\Sigma}_\pi^{-1}}\sqrt{\frac{2\mathcal{E}_K^{(\widehat{\Sigma}_\pi)}}{5}}h.
    \]
    Using Proposition \ref{cor:matrix_conditioning_comparison} gives that $K_{\widehat{\Sigma}_\pi^{-1}, \pi}$ is a $(\Gamma, \gamma, b)$-$W_2$ contraction to $\pi$ with
    \[
    \Gamma = 4\sqrt{3\kappa(\Sigma_\pi)}\textup{ , }\gamma = \frac{h}{6\kappa_{\Sigma_\pi^{-1}}}\textup{ and }b = 54\sqrt{2}\|\Sigma_\pi^{1/2}\|\kappa_{\Sigma_\pi^{-1}}\sqrt{\frac{2\mathcal{E}_K^{(\widehat{\Sigma}_\pi)}}{5}}h.
    \]
    A further application of Corollary \ref{cor:conditioning_comparison_covariance_case} gives the result, after substituting $\Gamma, \gamma$ and $b$ in the statement of Theorem \ref{thm:Wasserstein_contraction_complexity}.
    
    We now look at the Fisher preconditioned case. This proof works similarly to the proof of the complexity with ULA as the underlying kernel in Section \ref{proof:total_ULA_complexities_part_3}. Let $\delta,\Delta>0$. For the preconditioner learning phase, set the step-size $h = k\varepsilon$ where
    \begin{equation}
        \varepsilon = \frac{3}{8}\frac{\delta\Delta}{\sqrt{dL\kappa}}
    \end{equation}
    in accordance with Theorem \ref{thm:preconditioner_learn_complexity}. Impose that
    \begin{equation}
        k\leq\frac{8}{3}\sqrt{dL\kappa}
    \end{equation}
    to make $h\leq 1$. Then the underlying kernel is a $(\Gamma, \gamma, d)$-$W_2$ contraction to $\pi$ with $\Gamma$, $\gamma$ and $b$ as in the discussion before Theorem \ref{thm:total_unadjusted_underdamped_complexities}. Set
    \begin{equation}
        k = \frac{1}{1536}\frac{1}{\kappa}\sqrt{\frac{5}{2\Ee_K}}
    \end{equation}
    to make $3\Gamma^2b=(1/2)\varepsilon < \varepsilon$. To verify that this value of $k$ satisfies its upper bound we impose that $L\kappa\Dd^2 \geq 1$, which is done by hypothesis. The fact that $\delta\Delta<1$ and $\mathrm{tr}(\Sigma_\pi)\geq L^{-1}d$ gives $\varepsilon\leq\sqrt{\mathrm{tr}(\Sigma_\pi)}$. Therefore we can apply Theorems \ref{thm:Wasserstein_contraction_complexity} and \ref{thm:preconditioner_learn_complexity} to get the complexity to achieve
    \begin{equation}
        \left\|\Ff^{-1/2}\widehat{\Ff}\Ff^{-1/2} - \mathbf{I}_d\right\|\leq \Delta\textup{ with probability }\geq 1-\delta
    \end{equation}
    which is 
    \begin{equation}
        \tilde{O}\left(\frac{\kappa^{5/2}d^{3/2}\sqrt{L}\sqrt{\Ee_K}}{\delta}K_\mathrm{Fisher}^3\right)
    \end{equation}
    after subbing in the values of $N$, $\Gamma$, $\gamma$, and $b$ in terms of the step-size. The rest of the proof works similarly to the the covariance case.

\subsection{Proof of Theorem \ref{thm:total_unadjusted_HMC_complexities}}\label{proof:total_unadjusted_HMC_complexities}

We first look at the unpreconditioned case. For the $W_2$-contraction style result in \cite{bou-rabee2025} we neet that $h \leq T$ and $T \leq (8L)^{-1/2}$ so set $T = (8L)^{-1/2}$ and let the step-size be $h = k^{2/3}\varepsilon^{2/3}$ for $k>0$.
The hypothesis that 
\[
\varepsilon \leq \frac{1}{k}\left(\frac{1}{\sqrt{8L}}\right)^{3/2}
\]
for the value of $k$ that we will go on to choose ensures that $h \leq T$.
Set
\[
k = \frac{m}{40897L^\frac{5}{4}\sqrt{d\kappa}}
\]
to satisfy
\[
3\Gamma^2b = \frac{40896 L^\frac{5}{4}\sqrt{d\kappa}h^\frac{3}{2}}{m} <\varepsilon.
\]

To apply Theorem \ref{thm:Wasserstein_contraction_complexity} we must have that $\varepsilon \leq \sqrt{3\mathrm{tr}(\Sigma_\pi)}$ which holds by hypothesis.

Having verified the assumptions on \ref{thm:Wasserstein_contraction_complexity} we may apply its conclusion. However for this particular algorithm we need $n := \lfloor T / h\rfloor$ leapfrog steps, each of which uses one gradient evaluation. Hence the FLOP complexity is $\tilde{O}(n\gamma^{-1}N(d + \mathtt{G}))$. Using the upper bound $n \leq T / h$ and subbing in the values for $T$, $h$, and $\gamma$ gives the result.

We now look at the covariance preconditioned case. This proof works similarly to the proof of the complexity with ULA as the underlying kernel in Section \ref{proof:total_ULA_complexities_part_2}. For the preconditioner learning complexity let $\delta, \Delta > 0$. Set $h = k^{2/3}\varepsilon^{2/3}$ for $k>0$ where
\[
\varepsilon = \frac{\sqrt{2}}{120}\frac{\delta\Delta}{\sqrt{dL}}
\]
in accordance with Theorem \ref{thm:preconditioner_learn_complexity}. We need $T \leq (8L)^{-1/2}$ and $h\leq T$ so let $T = (8L)^{-1/2}$ and impose that
\[
k\leq \frac{120}{\sqrt{2}}\sqrt{dL}\frac{1}{(8L)^{3/4}}
\]
to give $h \leq T$. If we let
\[
k = \frac{1}{5113}\frac{mT^2}{L^{1/4}\sqrt{d\kappa}}
\]
this suffices to satisfy both $3\Gamma^2b < \varepsilon$ and the upper bound we imposed on $k$. The fact that $\delta\Delta < 1$ and $\mathrm{tr}(\Sigma_\pi) \geq L^{-1}d$ gives $\varepsilon \leq \sqrt{3\mathrm{tr}(\Sigma_\pi)}$. Therefore we can apply Theorems \ref{thm:Wasserstein_contraction_complexity} and \ref{thm:preconditioner_learn_complexity} to get the complexity to achieve
\begin{equation}
        \left\|\Sigma_\pi^{-1/2}\widehat{\Sigma}_\pi\Sigma_\pi^{-1/2} - \mathbf{I}_d\right\|\leq \Delta\textup{ with probability }\geq 1-\delta
\end{equation}
which is
\[
\tilde{O}\left(\frac{(d + \mathtt{G})\kappa^2d^{5/3}}{\delta^{2/3}}\max\left\{\frac{1}{\delta}, K_\mathrm{cov}^3\right\}\right).
\]

The rest of the proof works similarly to Section \ref{proof:total_ULA_complexities_part_2} where we use the fact that $\kappa_{\Sigma_\pi^{-1}} \leq 3\kappa_{\widehat{\Sigma}_\pi^{-1}}$, $L_{\Sigma_\pi^{-1}}\leq 2L_{\widehat{\Sigma}_\pi^{-1}}$, $2m_{\widehat{\Sigma}_\pi^{-1}} \leq 3m_{\Sigma_\pi^{-1}}$, and $3\mathrm{tr}(\Sigma_{\tilde{\pi}}) \geq 2d$ where $\tilde{\pi} = (\widehat{\Sigma}_\pi^{-1/2})_\sharp\pi$ with probability $\geq 1-\delta$ to ensure that the 
\[
        \varepsilon \leq \min \left\{\frac{40897L^{1/2}\sqrt{d\kappa}}{8^{3/4}m},\sqrt{3\mathrm{tr}(\Sigma_\pi)}\right\}
\]
condition from the unpreconditioned proof holds. So as in the unpreconditioned proof set $T = (8L_{\widehat{\Sigma}_\pi^{-1}})^{-1/2}$ and let the step size be $h = k^{2/3}\varepsilon^{2/3}$ for $k > 0$. The upper bound on $\varepsilon$ in the Theorem statement with
\[
k = \frac{m_{\widehat{\Sigma}_\pi^{-1}}}{40897L_{\widehat{\Sigma}_\pi^{-1}}^{5/4}\sqrt{d\kappa_{\widehat{\Sigma}_\pi^{-1}}}}
\]
satisfies $3\Gamma^2b < \varepsilon$ where
\[
\Gamma = 1\textup{ and }b = 1704\frac{L_{\widehat{\Sigma}_\pi^{-1}}^{1/4}\sqrt{d\kappa_{\widehat{\Sigma}_\pi^{-1}}}}{m_{\widehat{\Sigma}_\pi^{-1}}T^2}h^{3/2}.
\]
The value of $k$ also ensures that $h \leq T$ and so $K_{(\widehat{\Sigma}_\pi^{-1/2})_\sharp\pi}$ is a $(\Gamma, \gamma, b)$-$W_2$ contraction to $(\widehat{\Sigma}_\pi^{-1/2})_\sharp\pi$ with
\[
\Gamma = 1\textup{ , }\gamma = \frac{m_{\widehat{\Sigma}_\pi^{-1}}T^2}{6}\textup{ and }b =  1704\frac{L_{\widehat{\Sigma}_\pi^{-1}}^{1/4}\sqrt{d\kappa_{\widehat{\Sigma}_\pi^{-1}}}}{m_{\widehat{\Sigma}_\pi^{-1}}T^2}h^{3/2}.
\]
We can apply Corollary \ref{cor:approximate_preconditioned_kernel_efficiency} to get the $(\Gamma, \gamma, b)$ contraction parameters for the preconditioned kernel $K_{\widehat{\Sigma}_\pi^{-1}, \pi}$, compare quantities involving $\widehat{\Sigma}_\pi^{-1}$ with those involving $\Sigma_\pi^{-1}$ using Corollary \ref{cor:conditioning_comparison_covariance_case} and submit the subsequent values of $\Gamma, \gamma$ and $b$ in to Theorem \ref{thm:Wasserstein_contraction_complexity} to get the iteration complexity. Finally we note that the number of integrator iterations is upper bounded by $T / h$ to give the FLOP complexity in the result.

We now look at the Fisher preconditioned case. This proof works similarly to the proof of the complexity with ULA as the underlying kernel in Section \ref{proof:total_ULA_complexities_part_3}. For the preconditioner learning complexity let $\delta,\Delta > 0$. Set $h = k^{2/3}\varepsilon^{2/3}$ for $k>0$ where
\[
\varepsilon = \frac{3}{8}\frac{\delta\Delta}{\sqrt{dL\kappa}}
\]
in accordance with Theorem \ref{thm:preconditioner_learn_complexity}. We need $T \leq (8L)^{-1/2}$ and $h \leq t$ so let $T = (8L)^{-1/2}$ and impose that
\[
k \leq \frac{8}{3}\sqrt{dL\kappa}\frac{1}{(8L)^{3/4}}
\]
to give $h \leq T$. If we let
\[
k = \frac{m}{40897L^{5/4}\sqrt{d\kappa}}
\]
this suffices to satisfy both $3\Gamma^2b < \varepsilon$ and the upper bound we imposed on $k$. The fact that $\delta\Delta < 1$ and $\mathrm{tr}(\Sigma_\pi) \geq L^{-1}d$ gives $\varepsilon \leq \sqrt{3\mathrm{tr}(\Sigma_\pi)}$. Therefore we can apply Theorems \ref{thm:preconditioner_learn_complexity} and \ref{thm:Wasserstein_contraction_complexity} to get the complexity to achieve
\begin{equation}
        \left\|\Ff^{-1/2}\widehat{\Ff}\Ff^{-1/2} - \mathbf{I}_d\right\|\leq \Delta\textup{ with probability }\geq 1-\delta
\end{equation}
which is
\[
\tilde{O}\left(\frac{(d + \mathtt{G})\kappa^{7/3}d^{5/3}}{\delta^{2/3}}K_\mathrm{Fisher}^3\right).
\]
The rest of the proof works similarly to Section \ref{proof:total_ULA_complexities_part_2} where we use the fact that $\kappa_{\Ff} \leq 3\kappa_{\widehat{\Ff}}$, $2L_{\Ff}\leq 3L_{\widehat{\Ff}}$, $m_{\widehat{\Ff}} \leq 2m_{\Ff}$, and $3\mathrm{tr}(\Sigma_{\tilde{\pi}}) \geq (3/2)d$ where $\tilde{\pi} = (\widehat{\Ff}^{1/2})_\sharp\pi$ with probability $\geq 1-\delta$ to ensure that the 
\[
        \varepsilon \leq \min \left\{\frac{40897L^{1/2}\sqrt{d\kappa}}{8^{3/4}m},\sqrt{3\mathrm{tr}(\Sigma_\pi)}\right\}
\]
condition from the unpreconditioned proof holds. The rest of the proof works similarly to the covariance case.

\subsection{Proof of Theorem \ref{thm:total_proximal_complexities}}\label{proof:total_proximal_complexities}

We first look at the unpreconditioned case. As detailed in \citet[Section 4.2]{chen2022a}, each iteration of the proximal sampler requires an exact sample from the restricted Gaussian oracle which, in their case, is taken using a rejection sampler. As they say: each rejection sample takes $\tilde{\kappa}^{d/2}$ expected iterations where $\tilde{\kappa}$ is the condition number of the restricted Gaussian oracle distribution. Therefore, since we have $m$-strong log-concavity and $L$-smoothness we get $\tilde{\kappa} = (Lh + 1) / (mh + 1)$. Combined with the results of Theorem \ref{thm:Wasserstein_contraction_complexity} it therefore takes
\[
O\left((d + \mathtt{U})\frac{1}{\log(1 + mh)}\left(\frac{Lh + 1}{mh + 1}\right)^{d/2}N\log(\varepsilon^{-1})\right)
\]
expected FLOPS to derive an $\sqrt{N}\varepsilon$ approximately IID ensemble of samples with Algorithm \ref{alg:thinned_Markov_chain}. It remains to select a sensible value of $h$.

Since $\log(1 + mh)$ is concave in $h$ we know that for all $h \in [0, h^*]$ we have 
\[
\log(1 + mh) \geq \frac{\log(1 + mh^*)}{h^*}h.
\]
We let $h^* = L^{-1}$ such that
\[
\log(1 + mh) \geq L\log(1 + \kappa^{-1})h.
\]
We use the fact that $\log(1 + x) \geq x(1 + x)^{-1}$ for all $x>-1$ to give
\[
\log(1 + mh) \geq \frac{L}{\kappa + 1}h.
\]
Therefore
\[
\frac{1}{\log(1 + mh)}\left(\frac{Lh + 1}{mh + 1}\right)^{d/2} \leq \frac{\kappa + 1}{Lh}\left(\frac{Lh + 1}{mh + 1}\right)^{d/2}
\]
Choosing $h = 2 / (Ld + 2m + md)$ to make
\[
\left(\frac{Lh + 1}{mh + 1}\right)^{d/2} = \left(1 + \frac{2}{d}\right)^{d / 2}
\]
and using the fact that $(1 + x)^{1/x} \leq e$ for all $x > 0$ gives the result.

We now look at the covariance preconditioned case. For the preconditioner learning portion of the algorithm let $\delta,\Delta > 0$. In accordance with Theorem \ref{thm:preconditioner_learn_complexity} set
\[
\varepsilon = \frac{\sqrt{2}}{120}\frac{\delta\Delta}{\sqrt{dL}}\textup{ and }N = \max\left\{ \frac{5d}{\delta\Delta},\frac{2CK_\mathrm{cov}^{2}\left(d+\log\left(4/\delta\right)\right)\sqrt{CK_\mathrm{cov}^{2}+2\Delta}}{\Delta}\right\}.
\]
The fact that $\delta\Delta < 1$ and $\mathrm{tr}(\Sigma_\pi)\geq L^{-1}d$ gives $\varepsilon \leq \sqrt{3\mathrm{tr}(\Sigma_\pi)}$ and so we can apply part 1) of Theorem \ref{thm:total_proximal_complexities}. In a similar fashion to the proof of the unpreconditioned algorithm, set $h = 2 / (Ld + 2m + md)$ so that we get a preconditioner learning complexity of
\[
O\left((d + \mathtt{U})d\kappa N\log\varepsilon^{-1}\right) = \tilde{O}\left((d + \mathtt{U})d^2 \max\left\{\frac{1}{\delta}, K_\mathrm{cov}^3\right\}\right)
\]
after setting $\Delta = 0.5$.

For the time it takes to collect the $\sqrt{N}\varepsilon$ samples we can again use Theorem \ref{thm:total_proximal_complexities} part 1) along with the fact that $3\mathrm{tr}(\Sigma_{\tilde{\pi}}) \geq 2d$ where $\tilde{\pi}$ is the preconditioned target. Therefore we set the new step size for $K_{(\widehat{\Sigma}_\pi^{-1})_\sharp\pi}$ to be $2 / (\tilde{L}d + 2\tilde{m} + \tilde{m}d)$ where $\tilde{L}$ and $\tilde{m}$ are the smoothness and strong-convexity constants after preconditioning with $\widehat{\Sigma}_\pi^{-1}$. This gives us that $K_{(\widehat{\Sigma}_\pi^{-1})_\sharp\pi}$ is a $(\Gamma, \gamma, b)$-$W_2$ contraction to $\pi$ with $\Gamma = 1$, $\gamma = \log(1 + \tilde{m}h)$ and $b = 0$. As in the proof of Theorem \ref{thm:total_proximal_complexities} part 1) the expected per-iteration complexity of the rejection sampler can be upper bounded by $e$. Applying Corollary \ref{cor:approximate_preconditioned_kernel_efficiency} gives us that $K_{\widehat{\Sigma}_\pi^{-1}, \pi}$ is a $(\Gamma, \gamma, b)$-$W_2$ contraction to $\pi$ with $\Gamma = \sqrt{3\kappa(\Sigma_\pi)}$, $\gamma = \log(1 + \tilde{m}h)$ and $b = 0$. The expected per-iteration complexity is the same because the rejection sampler is the same between $K_{(\widehat{\Sigma}_\pi^{-1})_\sharp\pi}$ and $K_{\widehat{\Sigma}_\pi^{-1}, \pi}$. This gives us a complexity for $K_{\widehat{\Sigma}_\pi^{-1}, \pi}$ of
\[
O\left((d^2 + \mathtt{U})d\kappa_{\widehat{\Sigma}_\pi^{-1}}N\log\varepsilon^{-1}\right).
\]
We can subsequently use the fact that $\kappa_{\widehat{\Sigma}_\pi^{-1}} \leq 3\kappa_{\Sigma_\pi^{-1}}$, which stems from Corollary \ref{cor:condition_number_comparison}, to give the result.

Finally we look at the Fisher preconditioned case. For the preconditioner learning portion of the algorithm let $\delta,\Delta>0$. In accordance with Theorem \ref{thm:preconditioner_learn_complexity} set
\[
\varepsilon = \frac{3}{8}\frac{\delta\Delta}{\sqrt{dL\kappa}}\textup{ and }N = \frac{2cK_\mathrm{Fisher}^{2}\left(d+\log(4/\delta)\right)\sqrt{cK_\mathrm{Fisher}^{2}+2\Delta}}{\Delta}
\]

The fact that $\delta\Delta < 1$ and $\mathrm{tr}(\Sigma_\pi)\geq L^{-1}d$ gives $\varepsilon \leq \sqrt{3\mathrm{tr}(\Sigma_\pi)}$ and so we can apply part 1) of Theorem \ref{thm:total_proximal_complexities}. In a similar fashion to the proof of the unpreconditioned algorithm, set $h = 2 / (Ld + 2m + md)$ so that we get a preconditioner learning complexity of
\[
O\left((d + \mathtt{U})d\kappa N\log\varepsilon^{-1}\right) = \tilde{O}\left((d + \mathtt{U})d^2\kappa K_\mathrm{Fisher}^3\right)
\]
after setting $\Delta = 0.5$.

For the time it takes to collect the $\sqrt{N}\varepsilon$ samples we follow an analogous procedure to the covariance case, giving us a complexity for $K_{\widehat{\Ff}, \pi}$ of
\[
O\left((d^2 + \mathtt{U})d\kappa_{\widehat{\Ff}}N\log\varepsilon^{-1}\right).
\]
We can subsequently use the fact that $\kappa_{\widehat{\Ff}} \leq 3\kappa_{\Ff}$, which stems from Corollary \ref{cor:condition_number_comparison}, to give the result.

\subsection{Proofs of Results Concerning the \texorpdfstring{$\sqrt{N}\varepsilon$}{ε√N}-approximate IID from \texorpdfstring{$\pi$}{π} in \texorpdfstring{$W_2$}{W₂} condition}\label{subsection:proofs_for_rootNepsilonIID}

\subsubsection{Proof of Proposition \ref{prop:rootNIID_Lipschitz_transformation}}\label{proof:rootNIID_Lipschitz_transformation}

Let $X\in\mathbb{R}^{N\times d}$ be the matrix with columns $X_{t}$
for $t\in\left[N\right]$ and let $Y\in\mathbb{R}^{d\times N}$ be
a matrix whose columns are sampled independently from $\pi$. Let
$X$ and $Y$ be sampled according to the $W_{2}$-optimal coupling.
Defining $\tilde{X}$ as $X$ with $F$ applied to its columns and
$\tilde{Y}$ as $Y$ with $F$ applied to its columns, the optimal
coupling between $X$ and $Y$ therefore generates a coupling $\Gamma$
between $\tilde{X}$ and $\tilde{Y}$. Therefore
\begin{equation}\begin{split}
W_{2}\left(\tilde{X},\tilde{Y}\right)^{2} & =\inf_{\gamma\in\mathcal{C}(\mathcal{L}(\tilde{X}),(F_{\sharp}\pi)^{\otimes N})}\mathbb{E}_{\gamma}\left[\left\Vert \tilde{X}-\tilde{Y}\right\Vert _{F}^{2}\right]\\
 & \leq\mathbb{E}_{\Gamma}\left[\left\Vert \tilde{X}-\tilde{Y}\right\Vert _{F}^{2}\right]\\
 & =\sum_{t=1}^{N}\mathbb{E}_{\Gamma}\left[\left\Vert F\left(X_{t}\right)-F\left(Y_{t}\right)\right\Vert _{2}^{2}\right]\\
 & \leq L^{2}\mathbb{E}\left[\left\Vert X-Y\right\Vert _{F}^{2}\right]\\
 & \leq N\left(L\varepsilon\right)^{2}.
\end{split}\end{equation}
where $\mathcal{C}(\mathcal{L}(\tilde{X}),(F_{\sharp}\pi)^{\otimes N})$ is the set of couplings between $\mathcal{L}(\tilde{X})$ and $(F_{\sharp}\pi)^{\otimes N}$.

\subsubsection{Proof of Proposition \ref{prop:L2_rootNepsilon_mean_error}}\label{proof:L2_rootNepsilon_mean_error}
Let $\left\{ Y_{t}\right\} _{t=1}^{N}\sim\pi^{\otimes N}$ be coupled
optimally with $\left\{ X_{t}\right\} _{t=1}^{N}$ and define $\overline{Y}:=\left(1/N\right)\sum_{t}Y_{t}$.
Then
\begin{equation}\begin{split}
\mathbb{E}\left[\left\Vert \overline{X}-\mu_{\pi}\right\Vert ^{2}\right] & =\mathbb{E}\left[\left\Vert \overline{X}-\overline{Y}\right\Vert ^{2}\right]+2\mathbb{E}\left[\left\langle \overline{X}-\overline{Y},\overline{Y}-\mu_{\pi}\right\rangle \right]+\mathbb{E}\left[\left\Vert \overline{Y}-\mu_{\pi}\right\Vert ^{2}\right]\\
 & \leq\mathbb{E}\left[\left\Vert \overline{X}-\overline{Y}\right\Vert ^{2}\right]+2\sqrt{\mathbb{E}\left[\left\Vert \overline{X}-\overline{Y}\right\Vert ^{2}\right]}\sqrt{\mathbb{E}\left[\left\Vert \overline{Y}-\mu_{\pi}\right\Vert ^{2}\right]}+\mathbb{E}\left[\left\Vert \overline{Y}-\mu_{\pi}\right\Vert ^{2}\right]\\
 & =\left(\sqrt{\mathbb{E}\left[\left\Vert \overline{X}-\overline{Y}\right\Vert ^{2}\right]}+\sqrt{\mathbb{E}\left[\left\Vert \overline{Y}-\mu_{\pi}\right\Vert ^{2}\right]}\right)^{2}
\end{split}\end{equation}
where the second line follows by an application of Cauchy-Schwarz.
Now we have
\begin{equation}\begin{split}
\mathbb{E}\left[\left\Vert \overline{X}-\overline{Y}\right\Vert ^{2}\right] & =\mathbb{E}\left[\left\Vert \frac{1}{N}\sum_{t=1}^{N}\left(X_{t}-Y_{t}\right)\right\Vert ^{2}\right]\\
 & \leq\mathbb{E}\left[\frac{1}{N}\sum_{t=1}^{N}\left\Vert \left(X_{t}-Y_{t}\right)\right\Vert ^{2}\right]\\
 & \leq\varepsilon^{2}
\end{split}\end{equation}
where the second line comes from applying Jensen's inequality and
the final line follows by hypothesis.

We also have
\begin{equation}\begin{split}
\mathbb{E}\left[\left\Vert \overline{Y}-\mu_{\pi}\right\Vert ^{2}\right] & =\mathbb{E}\left[\left\langle \frac{1}{N}\sum_{t=1}^{n}\left(Y_{t}-\mu_{\pi}\right),\frac{1}{N}\sum_{s=1}^{n}\left(Y_{s}-\mu_{\pi}\right)\right\rangle \right]\\
 & =\frac{1}{N^{2}}\sum_{t,s=1}^{N}\mathbb{E}\left[\left\langle Y_{t}-\mu_{\pi},Y_{s}-\mu_{\pi}\right\rangle \right]\\
 & =\frac{1}{N^{2}}\sum_{t=1}^{N}\mathbb{E}\left[\left\Vert Y_{t}-\mu_{\pi}\right\Vert ^{2}\right]\\
 & =\frac{1}{N^{2}}\sum_{t=1}^{N}\mathbb{E}\left[\mathrm{tr}\left(\left(Y_{t}-\mu_{\pi}\right)\left(Y_{t}-\mu_{\pi}\right)^{\top}\right)\right]\\
 & =\frac{1}{N}\mathrm{tr}\left(\Sigma_{\pi}\right).
\end{split}\end{equation}
\subsubsection{Proof of Proposition \ref{prop:rootNepsilon_bounded_estimator_variance}}\label{proof:rootNepsilon_bounded_estimator_variance}
We prove the following Lemma from which the full result follows:
\begin{prop}
    Let $\left\{ X_{t}\right\} _{t=1}^{N}\in\mathbb{R}^{d\times N}$
be $\sqrt{N}\varepsilon$-approximately IID from $\pi$ in $W_{2}$.
Then
\begin{equation}\begin{split}
\mathrm{tr}\left(\sum_{t=1}^{N}\textup{Var}\left(X_{t}\right)\right) & \leq N\left(\sqrt{\mathrm{tr}\left(\Sigma_{\pi}\right)}+\varepsilon\right)^{2}\\
\textup{and }\left|\mathrm{tr}\left(\sum_{\underset{t\neq s}{t,s=1}}^{N}\textup{Cov}\left(X_{t},X_{s}\right)\right)\right| & \leq N^{2}\left(2\sqrt{\mathrm{tr}\left(\Sigma_{\pi}\right)}+\varepsilon\right)\varepsilon
\end{split}\end{equation}
where $\Sigma_{\pi}:=\textup{Var}_{\pi}\left(Y\right)$.
\end{prop}
For the sum of the variances we have that $\textup{Var}\left(X_{t}\right)\preceq\mathbb{E}\left[X_{t}X_{t}^{\top}\right]$.
Therefore, denoting by $X\in\mathbb{R}^{d\times N}$ the matrix with
columns $X_{t}$ for $t\in\left[N\right]$, we have
\begin{equation}\begin{split}
\mathrm{tr}\left(\sum_{t=1}^{N}\textup{Var}\left(X_{t}\right)\right) & \leq\sum_{t=1}^{N}\mathbb{E}\left[\left\Vert X_{t}\right\Vert ^{2}\right]
=\mathbb{E}\left[\left\Vert X\right\Vert _{F}^{2}\right]
 =\left\Vert X\right\Vert _{\mathrm{L}^{2}\left(\gamma\right)}^{2}
\end{split}\end{equation}
where $\gamma\in\mathcal{C}\left(\mathcal{L}\left(\left\{ X_{t}\right\} _{t=1}^{N}\right),\pi^{\otimes N}\right)$
is the $W_{2}$-optimal coupling. Let $Y\in\mathbb{R}^{d\times N}$
have columns $Y_{t}$ for $t\in\left[N\right]$ sampled identically
and independently from $\pi$. By the reverse triangle inequality,
we have
\begin{equation}\begin{split}
\mathrm{tr}\left(\sum_{t=1}^{N}\textup{Var}\left(X_{t}\right)\right)  & \leq\left(\left\Vert X-Y\right\Vert _{\mathrm{L}^{2}\left(\gamma\right)}+\left\Vert Y\right\Vert _{\mathrm{L}^{2}\left(\gamma\right)}\right)^{2}\\
 & \leq\left(\sqrt{N}\varepsilon+\sqrt{\sum_{t=1}^{N}\mathbb{E}\left[\left\Vert Y_{t}\right\Vert ^{2}\right]}\right)^2\\
 & =N\left(\varepsilon+\sqrt{\mathrm{tr}\left(\Sigma_{\pi}\right)}\right)^{2}
\end{split}\end{equation}
where the second line follows by hypothesis and we let $\Sigma_{\pi}:=\textup{Var}_{\pi}\left(Y\right)$.

For the sum of the covariances we use the decomposition
\[
\textup{Cov}\left(X_{t},X_{s}\right)=\mathbb{E}\left[X_{t}X_{s}^{\top}\right]-\mathbb{E}\left[Y_{t}Y_{s}^{\top}\right]+\mathbb{E}\left[Y_{t}\right]\mathbb{E}\left[Y_{s}\right]^{\top}-\mathbb{E}\left[X_{t}\right]\mathbb{E}\left[X_{s}\right]^{\top}
\]
for $t\neq s$, and hence by the triangle inequality we have
\[
\left|\mathrm{tr}\left(\sum_{\underset{t\neq s}{t,s=1}}^{N}\textup{Cov}\left(X_{t},X_{s}\right)\right)\right|\leq\sum_{t\neq s}\left|\mathbb{E}\left[X_{t}^{\top}X_{s}^ {}\right]-\mathbb{E}\left[Y_{t}^{\top}Y_{s}\right]\right|+\left|\mathbb{E}\left[Y_{t}\right]^{\top}\mathbb{E}\left[Y_{s}\right]-\mathbb{E}\left[X_{t}\right]^{\top}\mathbb{E}\left[X_{s}\right]\right|.
\]
Using the representation
\[
X_{t}^{\top}X_{s}-Y_{t}^{\top}Y_{s}^{\top}=\left(X_{t}-Y_{t}\right)^{\top}X_{s}^ {}+Y_{t}^{\top}\left(X_{s}-Y_{s}\right)
\]
the first term in the summand can be bounded as follows:
\[
\left|\mathbb{E}\left[X_{t}^{\top}X_{s}^ {}\right]-\mathbb{E}\left[Y_{t}^{\top}Y_{s}\right]\right|\leq\sqrt{\mathbb{E}\left[\left\Vert X_{s}\right\Vert ^{2}\right]}\sqrt{\mathbb{E}\left[\left\Vert X_{t}-Y_{t}\right\Vert ^{2}\right]}+\sqrt{\mathbb{E}\left[\left\Vert Y_{t}\right\Vert ^{2}\right]}\sqrt{\mathbb{E}\left[\left\Vert X_{s}-Y_{s}\right\Vert ^{2}\right]}
\]
where we have used Jensen's inequality and Cauchy-Schwarz. We use
a similar representation for the second summand:
\[
\mathbb{E}\left[Y_{t}\right]^{\top}\mathbb{E}\left[Y_{s}\right]-\mathbb{E}\left[X_{t}\right]^{\top}\mathbb{E}\left[X_{s}\right]=\mathbb{E}\left[X_{t}-Y_{t}\right]^{\top}\mathbb{E}\left[X_{s}\right]+\mathbb{E}\left[Y_{t}\right]^{\top}\mathbb{E}\left[X_{s}-Y_{s}\right]
\]
and so the second summand can be bounded as
\[
&\hspace{-2em} \left|\mathbb{E}\left[Y_{t}\right]^{\top}\mathbb{E}\left[Y_{s}\right]-\mathbb{E}\left[X_{t}\right]^{\top}\mathbb{E}\left[X_{s}\right]\right| \\ 
& \leq\sqrt{\mathbb{E}\left[\left\Vert X_{s}\right\Vert ^{2}\right]}\sqrt{\mathbb{E}\left[\left\Vert X_{t}-Y_{t}\right\Vert ^{2}\right]}+\sqrt{\mathbb{E}\left[\left\Vert Y_{t}\right\Vert ^{2}\right]}\sqrt{\mathbb{E}\left[\left\Vert X_{s}-Y_{s}\right\Vert ^{2}\right]}.
\]
Hence
\[
&\hspace{-2em} \left|\mathrm{tr}\left(\sum_{\underset{t\neq s}{t,s=1}}^{N}\textup{Cov}\left(X_{t},X_{s}\right)\right)\right| \\
& \leq2\sum_{t\neq s}\left(\sqrt{\mathbb{E}\left[\left\Vert X_{s}\right\Vert ^{2}\right]}\sqrt{\mathbb{E}\left[\left\Vert X_{t}-Y_{t}\right\Vert ^{2}\right]}+\sqrt{\mathbb{E}\left[\left\Vert Y_{t}\right\Vert ^{2}\right]}\sqrt{\mathbb{E}\left[\left\Vert X_{s}-Y_{s}\right\Vert ^{2}\right]}\right)
\]
We decompose the right hand side into two sums. Dealing with the first
sum:
\begin{equation}\begin{split}
\sum_{t\neq s}\sqrt{\mathbb{E}\left[\left\Vert X_{s}\right\Vert ^{2}\right]}\sqrt{\mathbb{E}\left[\left\Vert X_{t}-Y_{t}\right\Vert ^{2}\right]} & \leq\sum_{s=1}^{N}\sqrt{\mathbb{E}\left[\left\Vert X_{s}\right\Vert ^{2}\right]}\sum_{t=1}^{N}\sqrt{\mathbb{E}\left[\left\Vert X_{t}-Y_{t}\right\Vert ^{2}\right]}\\
 & \leq N^{2}\sqrt{\frac{1}{N}\sum_{s=1}^{N}\mathbb{E}\left[\left\Vert X_{s}\right\Vert ^{2}\right]}\sqrt{\frac{1}{N}\sum_{t=1}^{N}\mathbb{E}\left[\left\Vert X_{t}-Y_{t}\right\Vert ^{2}\right]}\\
 & =N\left\Vert X\right\Vert _{\mathrm{L}^{2}\left(\gamma\right)}\left\Vert X-Y\right\Vert _{\mathrm{L}^{2}\left(\gamma\right)}\\
 & \leq N^{3/2}\left(\left\Vert X-Y\right\Vert _{\mathrm{L}^{2}\left(\gamma\right)}+\left\Vert Y\right\Vert _{\mathrm{L}^{2}\left(\gamma\right)}\right)\varepsilon\\
 & \leq N^{2}\left(\sqrt{\mathrm{tr}\left(\Sigma_{\pi}\right)}+\varepsilon\right)\varepsilon
\end{split}\end{equation}
where in the second line we use the concavity of the square-root and
the fourth line follows by hypothesis. Dealing with the second sum:
\begin{equation}\begin{split}
\sum_{t\neq s}\sqrt{\mathbb{E}\left[\left\Vert Y_{t}\right\Vert ^{2}\right]}\sqrt{\mathbb{E}\left[\left\Vert X_{s}-Y_{s}\right\Vert ^{2}\right]} & \leq\sum_{t=1}^{N}\sqrt{\mathbb{E}\left[\left\Vert Y_{t}\right\Vert ^{2}\right]}\sum_{s=1}^{N}\sqrt{\mathbb{E}\left[\left\Vert X_{s}-Y_{s}\right\Vert ^{2}\right]}\\
 & \leq N^{2}\sqrt{\mathrm{tr}\left(\Sigma_{\pi}\right)}\sqrt{\frac{1}{N}\sum_{s=1}^{N}\mathbb{E}\left[\left\Vert X_{s}-Y_{s}\right\Vert ^{2}\right]}\\
 & =N^{3/2}\sqrt{\mathrm{tr}\left(\Sigma_{\pi}\right)}\left\Vert X-Y\right\Vert _{\mathrm{L}^{2}\left(\gamma\right)}\\
 & \leq N^{2}\sqrt{\mathrm{tr}\left(\Sigma_{\pi}\right)}\varepsilon
\end{split}\end{equation}
where in the second line we use the concavity of the square-root and
the fourth line follows by hypothesis. Putting everything together
gives the result

\subsubsection{Proof of Proposition \ref{prop:empirical_distribution_distance_rootNepsilonIID}}\label{proof:empirical_distribution_distance_rootNepsilonIID}

Let
\[
\nu_{X}:=\frac{1}{N}\sum_{t=1}^{N}\delta_{X_{t}}\textup{ and }\nu_{Y}:=\frac{1}{N}\sum_{t=1}^{N}\delta_{Y_{t}}.
\]
Sampling from $\nu_{X}$ can be done by generating $T\sim\mathrm{Uniform}\left\{ 1,\ldots,N\right\} $
and outputting $X_{T}$. A similar construction holds of $\nu_{Y}$.
Therefore we can couple $\nu_{X}$ and $\nu_{Y}$ by generating $T\sim\mathrm{Uniform}\left\{ 1,\ldots,N\right\} $
and outputting $\left(X_{T},Y_{T}\right)$. Let $\Gamma\in\mathcal{P}\left(\mathbb{R}^{d}\times\mathbb{R}^{d}\right)$
be this coupling. Then
\begin{equation}\begin{split}
W_{2}\left(\nu_{X},\nu_{Y}\right) & =\inf_{\gamma\in\mathcal{C}\left(\nu_{X},\nu_{Y}\right)}\int_{\mathbb{R}^{d}}\left\Vert x-y\right\Vert ^{2}\gamma\left(\mathrm{d}x,\mathrm{d}y\right)\\
 & \leq\int_{\mathbb{R}^{d}}\left\Vert x-y\right\Vert ^{2}\Gamma\left(\mathrm{d}x,\mathrm{d}y\right)\\
 & =\int_{\mathbb{R}^{d}}\left\Vert x-y\right\Vert ^{2}\sum_{t=1}^{N}\frac{1}{N}\delta_{X_{t}}\left(\mathrm{d}x\right)\delta_{Y_{t}}\left(\mathrm{d}y\right)\\
 & =\frac{1}{N}\sum_{t=1}^{N}\left\Vert X_{t}-Y_{t}\right\Vert ^{2}.
\end{split}\end{equation}
Taking expectations with respect to the optimal coupling between $\left\{ X_{t}\right\} _{t=1}^{N}$
and $\left\{ Y_{t}\right\} _{t=1}^{N}$ gives
\begin{equation}\begin{split}
\mathbb{E}\left[W_{2}\left(\nu_{X},\nu_{Y}\right)\right] & \leq\frac{1}{N}\mathbb{E}\left[\sum_{t=1}^{N}\left\Vert X_{t}-Y_{t}\right\Vert ^{2}\right]\\
 & =\varepsilon.
\end{split}\end{equation}

\subsubsection{Proof of Proposition \ref{prop:rootNIID_implies_empirical_covariance_bound}}\label{proof:rootNIID_implies_empirical_covariance_bound}

Let $X\in\mathbb{R}^{d\times N}$ be the matrix whose $j$th column
is $N^{-1/2}X_{j}$ and let $Y\in\mathbb{R}^{d\times N}$ be the matrix
whose $j$th column is $N^{-1/2}Y_{j}$ such that $\mathbb{E}\left[\left\Vert X-Y\right\Vert _{F}^{2}\right]\leq\varepsilon$
where the expectation is with respect to the $W_{2}$-optimal coupling
between $\left\{ X_{t}\right\} _{t=1}^{N}$ and $\left\{ Y_{t}\right\} _{t=1}^{N}$.
Then we use the following decomposition
\[
XX^{\top}-YY^{\top}=\left(X-Y\right)X^{\top}+Y\left(X-Y\right)^{\top}
\]
 to give
\begin{equation}\begin{split}
\left\Vert \frac{1}{N}\sum_{t=1}^{N}X_{t}X_{t}^{\top}-\frac{1}{N}\sum_{t=1}^{N}Y_{t}Y_{t}^{\top}\right\Vert _{F} & =\left\Vert XX^{\top}-YY^{\top}\right\Vert _{F}\\
 & \leq\left\Vert \left(X-Y\right)X^{\top}\right\Vert _{F}+\left\Vert Y\left(X-Y\right)^{\top}\right\Vert _{F}\\
 & \leq\left\Vert X-Y\right\Vert _{F}\left(\left\Vert X\right\Vert _{F}+\left\Vert Y\right\Vert _{F}\right)\\
 & \leq\left\Vert X-Y\right\Vert _{F}\left(\left\Vert X-Y\right\Vert _{F}+2\left\Vert Y\right\Vert _{F}\right)
\end{split}\end{equation}
where in the second line we use the triangle inequality, in the third
line we use the submultiplicativity of the Frobenius norm, and in
the final line we use the reverse triangle inequality. Taking expectations
with respect to the $W_{2}$-optimal coupling between $\left\{ X_{t}\right\} _{t=1}^{N}$
and $\left\{ Y_{t}\right\} _{t=1}^{N}$ we get
\begin{equation}\begin{split}
\mathbb{E}\left[\left\Vert \frac{1}{N}\sum_{t=1}^{N}X_{t}X_{t}^{\top}-\frac{1}{N}\sum_{t=1}^{N}Y_{t}Y_{t}^{\top}\right\Vert _{F}\right] & \leq\mathbb{E}\left[\left\Vert X-Y\right\Vert _{F}^{2}\right]+2\mathbb{E}\left[\left\Vert X-Y\right\Vert _{F}\left\Vert Y\right\Vert _{F}\right]\\
 & \leq\varepsilon^{2}+2\sqrt{\mathbb{E}\left[\left\Vert X-Y\right\Vert _{F}^{2}\right]}\sqrt{\mathbb{E}\left[\left\Vert Y\right\Vert _{F}^{2}\right]}\\
 & \leq\varepsilon^{2}+2\varepsilon\sqrt{\mathbb{E}\left[\left\Vert Y\right\Vert _{F}^{2}\right]}.
\end{split}\end{equation}
Finally we have that
\begin{equation}\begin{split}
\mathbb{E}\left[\left\Vert Y\right\Vert _{F}^{2}\right] & =\frac{1}{N}\sum_{t=1}^{N}\mathbb{E}\left[\left\Vert Y_{t}\right\Vert ^{2}\right]\\
 & =\frac{1}{N}N\mathrm{tr}\left(\Sigma_{\pi}\right)=\mathrm{tr}\left(\Sigma_{\pi}\right).
\end{split}\end{equation}

\subsection{Proofs of Results in \cref{section:improving_the_conditioning}}

\subsubsection{Proof of \cref{lem:Fisher_improves_multiplicative_Hessian}}\label{proof:Fisher_improves_multiplicative_Hessian}

Define $\Lambda=\mathbb{E}_{\pi}\left[\Lambda\left(X\right)\right]$
such that $\mathcal{F}=Z^{\top}\Lambda Z$. Then
\begin{align*}
\lambda_{1}\left(\mathcal{F}^{-1/2}\nabla^{2}U\left(x\right)\mathcal{F}^{-1/2}\right) & =\lambda_{1}\left(\left(Z^{\top}\Lambda Z\right)^{-1/2}Z^{\top}\Lambda\left(x\right)Z\left(Z^{\top}\Lambda Z\right)^{-1/2}\right)\\
 & \leq\lambda_{1}\left(\left(Z^{\top}\Lambda Z\right)^{-1/2}Z^{\top}Z\left(Z^{\top}\Lambda Z\right)^{-1/2}\right)C\\
 & =\lambda_{1}\left(\left(Z^{\top}Z\right)^{1/2}\left(Z^{\top}\Lambda Z\right)^{-1}\left(Z^{\top}Z\right)^{1/2}\right)C\\
 & =\lambda_{d}\left(\left(Z^{\top}Z\right)^{-1/2}Z^{\top}\Lambda Z\left(Z^{\top}Z\right)^{-1/2}\right)^{-1}C\\
 & \leq\lambda_{d}\left(\left(Z^{\top}Z\right)^{-1/2}Z^{\top}Z\left(Z^{\top}Z\right)^{-1/2}\right)^{-1}\frac{C}{\lambda_{d}\left(\Lambda\right)}\\
 & =\frac{C}{\lambda_{d}\left(\Lambda\right)}
\end{align*}
where in the second and penultimate lines we use the rectangular form
of Ostrowski's theorem. One may carry out similar calculations to
reveal that
\[
\lambda_{d}\left(\mathcal{F}^{-1/2}\nabla^{2}U\left(x\right)\mathcal{F}^{-1/2}\right)\geq\frac{c}{\lambda_{1}\left(\Lambda\right)}
\]
from which we derive that
\[
\kappa_{\mathcal{F}}\leq\frac{C}{c}\kappa\left(\Lambda\right).
\]
Finally note that $\lambda_{1}\left(\Lambda\right)=\lambda_{1}\left(\mathbb{E}_{\pi}\left[\Lambda\left(X\right)\right]\right)$
and $\lambda_{d}\left(\Lambda\right)=\lambda_{d}\left(\mathbb{E}_{\pi}\left[\Lambda\left(X\right)\right]\right)$.
The top and bottom eigenvalue functions are convex and concave respectively.
Therefore applying Jensen's inequality twice gives
\[
\kappa\left(\Lambda\right)\leq\frac{C}{c}.
\]

\subsubsection{Proof of \cref{lem:covariance_condition_bounded_by_Fisher}}\label{proof:covariance_condition_bounded_by_Fisher}

The Brascamp-Lieb and Cramér-Rao inequalities give
\[
\mathcal{F}^{-1}\preceq\Sigma_{\pi}\preceq\mathbb{E}_{\pi}\left[\nabla^{2}U\left(X\right)^{-1}\right]
\]
with respect to the Loewner order. Using \cref{cor:condition_number_comparison}
gives
\begin{align*}
\kappa_{\Sigma^{-1}_{\pi}} & \leq\kappa\left(\Sigma^{1/2}_{\pi}\mathcal{F}\Sigma^{1/2}_{\pi}\right)\kappa_{\mathcal{F}}\\
 & =\frac{\lambda_{1}\left(\Sigma^{1/2}_{\pi}\mathcal{F}\Sigma^{1/2}_{\pi}\right)}{\lambda_{d}\left(\Sigma^{1/2}_{\pi}\mathcal{F}\Sigma^{1/2}_{\pi}\right)}\kappa_{\mathcal{F}}\\
 & =\frac{\lambda_{1}\left(\mathcal{F}^{1/2}\Sigma_{\pi}\mathcal{F}^{1/2}\right)}{\lambda_{d}\left(\mathcal{F}^{1/2}\Sigma_{\pi}\mathcal{F}^{1/2}\right)}\kappa_{\mathcal{F}}\\
 & \leq\frac{\lambda_{1}\left(\mathcal{F}^{1/2}\mathbb{E}_{\pi}\left[\nabla^{2}U\left(X\right)^{-1}\right]\mathcal{F}^{1/2}\right)}{\lambda_{d}\left(\mathcal{F}^{1/2}\mathcal{F}^{-1}\mathcal{F}^{1/2}\right)}\kappa_{\mathcal{F}}
\end{align*}
from which the result follows.

\end{document}